\documentclass[a4paper,12pt]{article}
\usepackage{geometry}
\usepackage{amsmath}
\usepackage{amsthm}
\usepackage{authblk}
\usepackage{hyperref}
\usepackage{bm}
\usepackage{cite}
\usepackage{mathrsfs}
\usepackage{amssymb}
\usepackage{amsfonts}
\usepackage{graphicx}
\usepackage[justification=centering]{caption}
\usepackage{subcaption}
\usepackage{tikz}
\usepackage{mathtools}
\usepackage{xcolor}
\usepackage{authblk}
\usepackage[labelfont=bf]{caption}
\usetikzlibrary{intersections}
\allowdisplaybreaks

\theoremstyle{plain}
\newtheorem{prop}{Proposition}[section]

\newtheorem{theo}{Theorem}[section]
\newtheorem{lem}{Lemma}[section]

\newtheorem{assu}{Assumption}[section]
\newtheorem{rhp}{RH problem}[section]
\newtheorem*{unnumberednum*}{Numerical Verifications}
\def\rd{{\rm d}}

\allowdisplaybreaks
\theoremstyle{remark}
\newtheorem*{remark}{Remark}

\title{Long-time behaviors of the two-component nonlinear Klein-Gordon equation:
higher-order asymptotics}

\author[$^{\dagger}$]{Deng-Shan Wang}
\author[$^{\dagger}$]{Yingmin Yang\thanks{Corresponding author: yangym@mail.bnu.edu.cn}}
\author[$^{\ddagger}$]{Liming Zang}

\affil[$^{\dagger}$]{\small School of Mathematical Sciences,
	Beijing Normal University, Beijing 100875, China}
\affil[$^{\ddagger}$]{School of Applied Science, Beijing Information Science and Technology University, Beijing 100192, China}

\date{}

\begin{document}
\maketitle
\begin{abstract}
This work investigates the long-time asymptotic behaviors of solutions to the initial value problem of the two-component nonlinear Klein-Gordon equation by inverse scattering transform and Riemann-Hilbert formulism. Two reflection coefficients are defined and their properties are analyzed in detail. The Riemann-Hilbert problem associated with the initial value problem is constructed in term of the two reflection coefficients. The Deift-Zhou nonlinear steepest descent method is then employed to analyze the Riemann-Hilbert problem, yielding the long-time asymptotics of the solution in different regions. Specifically, a higher-order asymptotic expansion of the solution inside the light cone is provided, and the leading term of this asymptotic solution is compared with results from direct numerical simulations, showing excellent agreement. This work not only provides a comprehensive analysis of the long-time behaviors of the two-component nonlinear Klein-Gordon equation but also offers a robust framework for future studies on similar nonlinear systems with third-order Lax pair.
	\end{abstract}

\section{Introduction}

\ \ \ \
In relativistic quantum mechanics, the mass-energy equation for a free particle is
$E^2=c^2p^2+m^2c^4$, where $E, p$ and $m$ are the energy, momentum and rest mass of the particle, and $c$ is the speed of light. For the Planck constant $\hbar$, the Schr\"{o}dinger correspondence principle shows that $E\rightarrow i\hbar\partial_t$ and $p\rightarrow -i\hbar \nabla$, which result in the Klein-Gordon (KG) equation
$$
\hbar^2\partial_t^2\Psi=c^2\hbar^2\triangle\Psi-m^2c^4\Psi,
$$
for wave function $\Psi=\Psi(t,{\bf{x}})$. Initially, physicists Oskar Klein and Walter Gordon introduced the KG equation to investigate relativistic electrons. Later, it was discovered that the KG equation accurately describes the spinless pion \cite{Messiah-1999}, even though the Dirac equation was already renowned for its effective description of spinning electrons. The nonlinear KG equation is readily derived by incorporating the term $V'(|\Psi|^2)\Psi$, where
$V$ represents the potential energy density of the fields:
$$
\frac{\hbar^2}{2mc^2}\partial_t^2\Psi-\frac{\hbar^2}{2m}\triangle\Psi
+\frac{mc^2}{2}\Psi+V'(|\Psi|^2)\Psi=0.
$$
Different potentials correspond different kinds of nonlinear KG equations such as the linear wave equation for $V=|\Psi|^2$ and the $\Psi^4$-field equation for $V=|\Psi|^4$ \cite{Kleinert-2001}, which has direct connection with the SU(2) Yang-Mills theory. Moreover, both sine-Gordon equation $u_{tt} - u_{xx} =\sin(u)$ \cite{Takhatajian-1974,Zhou-sG} and Tzitz\'eica equation \cite{tzitzeica,tzitzeica-2024}
\begin{align}\label{Tt}
		u_{tt} - u_{xx} = {\rm{e}}^{-2u} - {\rm{e}}^u,
\end{align}
belong to the nonlinear KG type equation. The Tzitz\'eica equation (\ref{Tt}) was first derived to described affine surface with negative affine mean curvature \cite{tzitzeica} and can also be reduced from the anti-self-dual Yang-Mills in the space $\mathbb{R}^{2,2}$ by gauge group ${\rm{SL}}(3,\mathbb{R})$ \cite{Dunajski-2009}. Besides there are also various variants of the one-dimensional nonlinear KG equation. For example, C${\rm\hat{o}}$te, Martel and Yuan \cite{Martel-2021} studied the long-time asymptotics of the one-dimensional damped nonlinear KG equation
$$
u_{tt}+2\alpha u_t - u_{xx} +u-|u|^{p-1}u=0,\quad \alpha>0,\quad p>2.
$$
\par
In the early 1980s, Fordy and Gibbons \cite{Fordy1,Fordy2} proposed a class of coupled nonlinear KG equations solvable by means of inverse scattering transform, among which the following two-component nonlinear KG equation
\begin{equation}\label{1}
		\left\{
		\begin{aligned}
			u_{\xi\tau}&={\rm{e}}^{2u}-{\rm{e}}^{-u}\cosh{3v},\\
			v_{\xi\tau}&={\rm{e}}^{-u}\sinh{3v},
		\end{aligned}\right.
\end{equation}
is surprising to be completely integrable and has Lax pair of the form
	\begin{equation}\label{Lax-O}
		\begin{aligned}
			\Phi_\xi=U\Phi,\quad \Phi_\tau=V\Phi,\quad \Phi=(\phi_1,\phi_2,\phi_3)^{\rm T},
		\end{aligned}
	\end{equation}
	where the matrices $U$ and $V$ are
	\begin{equation*}
			\begin{aligned}
			U&=
			\begin{pmatrix}
				u_\xi+v_\xi & \lambda   &0\\
				0   &-2v_\xi    &\lambda\\
				\lambda &0  &-u_\xi+v_\xi
			\end{pmatrix},\quad V=
			\frac{1}{\lambda}\begin{pmatrix}
				0   &   0   & {\rm{e}}^{2u}\\
				{\rm{e}}^{-u-3v}   &0  &0\\
				0   &{\rm{e}}^{3v-u}   &0
			\end{pmatrix},
		\end{aligned}
	\end{equation*}
where $\lambda$ is the spectral parameter.
\par
Taking the laboratory coordinates $x=\xi+\tau$ and $ t=\xi-\tau$, it is deduced that
	\begin{equation*}\begin{aligned}
		\partial_x=\frac{1}{2}\left(\partial_\xi+\partial_\tau\right),\quad \partial_t=\frac{1}{2}\left(\partial_\xi-\partial_\tau\right).
	\end{aligned}\end{equation*}
Applying these transformations, the Lax pair (\ref{Lax-O}) is converted into
	\begin{equation}\begin{aligned}\label{2}
		\Phi_x=\tilde{L}\Phi,\quad \Phi_t=\tilde{A}\Phi,\quad \Phi=(\phi_1,\phi_2,\phi_3)^{\rm T},
	\end{aligned}\end{equation}
	where the matrices $\tilde{L}$ and $\tilde{A}$ are
	
	\begin{align*}
		\tilde{L}=\frac{1}{2}(U+V)&=\frac{\lambda}{2}\begin{pmatrix}
			0   &1  &0\\
			0   &0  &1\\
			1   &0  &0
		\end{pmatrix}+\frac{1}{2\lambda}\begin{pmatrix}
			0   &0 & {\rm{e}}^{2u}\\
			{\rm{e}}^{-u-3v}   &0  &0\\
			0&{\rm{e}}^{3v-u}  &0
		\end{pmatrix}\\
		&\quad+\frac{1}{2}\begin{pmatrix}
			(u+v)_x+(u+v)_t &0  &0\\
			0   &-2v_x-2v_t &0\\
			0   &0  &(v-u)_x+(v-u)_t
		\end{pmatrix},
		\\
		\tilde{A}=\frac{1}{2}(U-V)&=\frac{\lambda}{2}\begin{pmatrix}
			0   &1  &0\\
			0   &0  &1\\
			1   &0  &0
		\end{pmatrix}-\frac{1}{2\lambda}\begin{pmatrix}
			0   &0  &{\rm{e}}^{2u}\\
			{\rm{e}}^{-u-3v}   &0  &0\\
			0&{\rm{e}}^{3v-u}  &0
		\end{pmatrix}\\
		&\quad+\frac{1}{2}\begin{pmatrix}
			(u+v)_x+(u+v)_t &0  &0\\
			0   &-2v_x-2v_t &0\\
			0   &0  &(v-u)_x+(v-u)_t
		\end{pmatrix},\end{align*}
	and the equation \eqref{1} is converted into
	\begin{equation}\label{two-component-KG}
		\left\{
		\begin{aligned}
			&u_{tt}-u_{xx}=-{\rm{e}}^{2u}+{\rm{e}}^{-u}\cosh{3v},\\
			&v_{tt}-v_{xx}=-{\rm{e}}^{-u}\sinh{3v},
		\end{aligned}
		\right.
	\end{equation}
which can be reduced to the Tzitz\'eica equation (\ref{Tt}) \cite{tzitzeica} by taking $v\to u$ and then $u\to \frac{1}{2}u.$ In fact, the mappings $u\to u+\frac{1}{2}v$ and $v\to -\frac{1}{2}v$ indicate that the two-component nonlinear KG equation (\ref{two-component-KG}) is equivalent to the $\mathbb{Z}_3$ two dimensional Toda chain \cite{Toda-2012}
\begin{equation}\label{two-component-Toda}
		\left\{
		\begin{aligned}
			&u_{tt}-u_{xx}={\rm{e}}^{v-u}-{\rm{e}}^{2u+v},\\
			&v_{tt}-v_{xx}={\rm{e}}^{-u-2v}-{\rm{e}}^{v-u},
		\end{aligned}
		\right.
	\end{equation}
given by Mikhailov \cite{Mikhailov-1981} through reducing Toda lattice $u_{\alpha,tt}-u_{\alpha,xx}={\rm{e}}^{u_{\alpha+1}-u_\alpha}-{\rm{e}}^{u_{\alpha}-u_{\alpha-1}}$ for $\alpha\in\mathbb{Z}/\mathbb{Z}_3$ and $u_1+u_2+u_3=0.$ On the other hand, Dunajski and Plansangkate \cite{Dunajski-2009} re-derived the $\mathbb{Z}_3$ two dimensional Toda chain (\ref{two-component-Toda}) through reducing the Hitchin equations (also called the ``self dual equation on a Riemann surface") in gauge theory \cite{Hitchin-1987}. \"{O}zer \cite{Ozer-1998} derived a new integrable coupled nonlinear Schr\"{o}dinger equation from the two-component nonlinear KG equation (\ref{two-component-KG}) by
the method of multiple scales. Wu, He and Geng \cite{Wu-Geng-2013} introduce a trigonal curve associated with the two-component nonlinear KG equation and derive its quasi-periodic solutions in terms of the Baker-Akhiezer function \cite{Kotlyarov-Shepelsky-2017}. Thus the two-component nonlinear KG equation (\ref{two-component-KG}) is an important integrable model in quantum field theory, gauge field theory and low-dimensional geometry.
\par
So far, the understanding of the solutions and asymptotic properties of the two-component nonlinear KG equation (\ref{two-component-KG}) remains quite limited. A primary reason is that its Lax pair is of third order, which significantly complicates the inverse scattering transform \cite{Constantin-2010} required to solve it, thereby rendering the task highly intricate and challenging as Deift proposed in the 8th of the list of sixteen open problems \cite{Deift-2008}. In recent years, Charlier, Lenells and one of the present authors \cite{good-boussinesq,lenells-wang-good} have endeavored to explore the long-time asymptotics of the good Boussinesq equation, which also features a third-order Lax pair, by developing the Deift-Zhou nonlinear steepest descent method \cite{Deift-Zhou-1993} for oscillatory Riemann-Hilbert (RH) problems \cite{Charlier-Lenells-2022,Grava-Minakov-2020,Girotti-Grava-2021,WangJMP,Charlier-Lenells-miura,Charlier-Lenells-2024,Wang-Zhu-2025}. This offers a potential approach for exploring the inverse scattering theory and examining the long-time asymptotics of the two-component nonlinear KG equation (\ref{two-component-KG}).
\par
This work develops a formalism of the inverse scattering transform for the initial value problem of two-component nonlinear KG equation \eqref{two-component-KG} on the line, i.e.,
\begin{equation}\label{3}
		\left\{
		\begin{aligned}
			&u_{tt}-u_{xx}=-{\rm{e}}^{2u}+{\rm{e}}^{-u}\cosh{3v},\\
			&v_{tt}-v_{xx}=-{\rm{e}}^{-u}\sinh{3v},\\
			&u(x,0)=u_0(x)\in \mathcal{S}(\mathbb{R}),\quad v(x,0)=v_0(x)\in \mathcal{S}(\mathbb{R}).
		\end{aligned}
		\right.
	\end{equation}
Assuming the existence of a solution, it is demonstrated that the solution to this problem can be expressed in terms of the solution to a $3\times3$ matrix-valued RH problem, which is formulated by using two reflection coefficients and is defined solely in terms of the initial data. Subsequently, the long-time asymptotic behaviors of the solution to (\ref{3}) are established through the Deift-Zhou nonlinear steepest descent analysis and the theory of high-order asymptotics \cite{higher-order}.
\par
The organization of this paper is as follows: Section \ref{sec2} presents the main results of this paper. In Section \ref{sec3}, starting from the Lax pair of the two-component nonlinear KG equation \eqref{3}, a thorough spectral analysis is carried out, and the corresponding RH problem is constructed. Section \ref{sec4} employs the nonlinear steepest descent method to meticulously examine the long-time asymptotic behaviors of the solution across four distinct regions as time approaches infinity. A significant focus is placed on deriving higher-order asymptotic expressions for the solution within the light cone, thereby offering a more nuanced understanding for the two-component nonlinear KG equation \eqref{3} and its long-term dynamics.

	\section{Main results}\label{sec2}
	\quad\quad This section lists the main theorems and conclusions of the present work. Here, we first provide the definitions of the scattering matrices that will be used in the following conclusions. For the initial value problem (\ref{3}), direct scattering analysis suggests that one can define the scattering matrices $s(\lambda)=(s_{ij}(\lambda))_{3\times3}$ and $s^A(\lambda)=(s^A_{ij}(\lambda))_{3\times3}$ for the regions $D_1$ and $D_4$ in Fig. \ref{figSigma}, which are given in \eqref{slambda} and \eqref{sAlambda} below, respectively.  First of all, the following assumption is formulated.

    \begin{assu}\label{assu1}
    	$($Absence of solitons$)$.
    	In the context of the initial value problem \eqref{3}, it is presupposed that the elements $(s(\lambda))_{11}$ and $(s^A(\lambda))_{11}$ are nonzero for $\lambda\in \bar{D}_1\setminus\{0\}$ and $\lambda\in \bar{D}_4\setminus\{0\}$, respectively.
    \end{assu}
\par
	That is, we suppose there do not exist solitons in the problem \eqref{3}. The reflection coefficients $r_1(\lambda)$ and $r_2(\lambda)$ are defined by
	\begin{equation}\label{4}
		\left\{
		\begin{aligned}
			&r_1(\lambda):=\frac{s_{12}(\lambda)}{s_{11}(\lambda)},\quad \lambda \in (0,\infty),\\
			&r_2(\lambda):=\frac{s^A_{12}(\lambda)}{s^A_{11}(\lambda)},\quad \lambda \in (-\infty,0).
		\end{aligned}
		\right.
	\end{equation}
	
	\begin{theo}\label{theo21}
		
		Assuming that the initial data $u_0$, $v_0\in \mathcal{S}(\mathbb{R})$, the reflection coefficients $r_1(\lambda)$ and $r_2(\lambda)$ are rigorously defined for $\lambda\in (0,\infty)$ and $\lambda \in (-\infty,0)$, respectively and exhibit the subsequent properties:
    \begin{enumerate}
    	
    	\item The functions $r_1(\lambda)$ and $r_2(\lambda)$ are smooth within their respective domains of definition and rapidly decay as $\lambda\rightarrow\pm\infty$.
    	\item Both the functions $r_1(\lambda)$, $r_2(\lambda)$ and their derivatives rapidly decay as $\lambda\rightarrow0$, and $r_1(0)=r_2(0)=0$.
    \end{enumerate}
	\end{theo}

	\begin{rhp}\label{rhp}
	Find a $3\times3$ matrix-valued function $M(x,t,\lambda)$ with the following properties:
	\begin{itemize}
		\item The function $M(x,t,\lambda)$ is holomorphic for $\lambda\in\mathbb{C}\setminus\Sigma$, where 	$\Sigma$ is defined by Fig. \ref{figSigma}.
		\item The function $M(x,t,\lambda)$ is analytic for $\lambda\in\mathbb{C}\setminus\Sigma$, and as $\lambda$ approaches $\Sigma$ from the left and right, the boundary values $M_+(x,t,\lambda)$ and $M_-(x,t,\lambda)$ of $M(x,t,\lambda)$ exist and satisfy the following relationship:
		\begin{equation}\label{5}
			M_+(x,t,\lambda)=M_-(x,t,\lambda)v_j(x,t,\lambda), \quad \lambda\in \Sigma_j,
		\end{equation}
		 for $j=1,2,\cdots,6$, and
	\begin{equation}\label{6}\begin{aligned}
	v_1 &\!\!=\!\!\begin{pmatrix}
		1	&	-r_1(\lambda){\rm{e}}^{-\theta_{21}}	&0\\
		r_1^*(\lambda){\rm{e}}^{\theta_{21}}	&1-r_1(\lambda)r_1^*(\lambda)	&0\\
		0	&0&1
	\end{pmatrix}\!\!,&&
	v_2\!\!=\!\!\begin{pmatrix}
		1	\!\!&0\!\!	&0\\
		0	\!\!&1\!-\!r_2(\omega\lambda)r_2^*(\omega\lambda)\!\!	&-r_2^*(\omega\lambda){\rm{e}}^{-\theta_{32}}\\
		0	\!\!&r_2(\omega\lambda){\rm{e}}^{\theta_{32}}\!\!	&1
	\end{pmatrix}\!\!,\\
	v_3 &\!\!=\!\!\begin{pmatrix}
		1-r_1(\omega^2\lambda)r_1^*(\omega^2\lambda)	&	\!\!0\!\!	&r_1^*(\omega^2\lambda){\rm{e}}^{-\theta_{31}}\\
		0 &\!\! 1\!\!& 0\\
		-r_1(\omega^2\lambda){\rm{e}}^{\theta_{31}}	&\!\!0\!\!	&1\\
	\end{pmatrix}\!\!,\!&&\!
	v_4\!\! =\!\!\begin{pmatrix}
		1-r_2(\lambda)r_2^*(\lambda)	&-r_2^*(\lambda){\rm{e}}^{-\theta_{21}}		&0\\
		r_2(\lambda){\rm{e}}^{\theta_{21}}	&1	&0\\
		0	&0&1
	\end{pmatrix}\!\!,\\
	v_5 &\!\!=\!\!\begin{pmatrix}
		1	&0	&0\\
		0	&1	&-r_1(\omega\lambda){\rm{e}}^{-\theta_{32}}\\
		0	&r_1^*(\omega\lambda){\rm{e}}^{\theta_{32}}&1-r_1(\omega\lambda)r_1^*(\omega\lambda)
	\end{pmatrix}\!\!,&&
	v_6 \!\!=\!\!\begin{pmatrix}
		1\!\!	&\!\!0\!\!	&r_2(\omega^2\lambda){\rm{e}}^{-\theta_{31}}\\
		0\!\!	&\!\!1\!\!	&0\\
		-r^*_2(\omega^2\lambda){\rm{e}}^{\theta_{31}}\!\!	&\!\!0\!\!	&1\!-\!r_2(\omega^2\lambda)r^*_2(\omega^2\lambda)
	\end{pmatrix}\!\!,
\end{aligned}\end{equation}
		where $\omega=\text{e}^{\frac{2}{3}\pi i}$, $\theta_{ij}=(l_i-l_j)x+(a_i-a_j)t$, $l_j=\frac{\omega^j\lambda+(w^j\lambda)^{-1}}{2}$, $a_j=\frac{\omega^j\lambda-(w^j\lambda)^{-1}}{2}$.
		
		\item  As $\lambda\rightarrow\infty$, $M(x,t,\lambda)=I+\mathcal{O}(\frac{1}{\lambda})$ for $\lambda\in\mathbb{C}\setminus\Sigma$.
		\item  As $\lambda\rightarrow0$, $M(x,t,\lambda)=G(x,t)+\mathcal{O}({\lambda})$ for $\lambda\in\mathbb{C}\setminus\Sigma$, where
		\begin{equation*}\begin{aligned}
				G(x,t)=
				\frac{{\rm{e}}^{2u}}{3}\begin{pmatrix}
					1	&	\omega	&	\omega^2\\
					\omega^2	&	1	&\omega\\
					\omega	&\omega^2	&1
				\end{pmatrix}+\frac{{\rm{e}}^{u-3v}}{3}\begin{pmatrix}
					1	&\omega^2	&\omega\\
					\omega	&1	&\omega^2\\
					\omega^2	&\omega	&1
				\end{pmatrix}+\frac{1}{3}
				\begin{pmatrix}
					1	&1	&1\\
					1	&1	&1\\
					1	&1	&1
				\end{pmatrix}.
		\end{aligned}\end{equation*}
		\item  $M(x,t,\lambda)=\tilde{\sigma}_1^{-1}M(x,t,\omega\lambda)\tilde{\sigma}_1=\tilde{\sigma}_2\overline{M(x,t,\overline{\lambda})}\tilde{\sigma}_2^{-1}$, where
		\begin{equation*}
				\tilde{\sigma}_1=\begin{pmatrix}
				0   &   1   &0\\
				0   &   0   &1\\
				1   &   0   &0
			\end{pmatrix},\quad	\tilde{\sigma}_2=\begin{pmatrix}
				0   &   1   &0\\
				1   &   0   &0\\
				0   &   0   &1
			\end{pmatrix}.
		\end{equation*}
	\end{itemize}
	\end{rhp}
	
	\begin{theo}\label{reconstruction formula}
		Under the conditions of Assumption \ref{assu1}, let $u(x,t)$ and $v(x,t)$ be solutions to the initial value problem \eqref{3}. Define the reflection coefficients $r_1(\lambda)$ and $r_2(\lambda)$ according to the equation \eqref{4} and an existence time $T\in (0,\infty)$. The RH problem \ref{rhp} admits a unique solution $M(x,t,\lambda)$ for each point $(x,t)\in \mathbb{R}\times \left[0,T\right)$ if it exists, and moreover, the solution to the initial value problem of two-component nonlinear KG equation \eqref{3} can be expressed by
		\begin{equation}\label{7}
			\left\{
			\begin{aligned}
				&u(x,t)=\frac{1}{2}\lim_{\lambda\rightarrow0}\ln[(\omega,\omega^2,1)M(x,t,\lambda)]_{13},\\
				&v(x,t)=\frac{1}{6}\lim_{\lambda\rightarrow0}\ln[(\omega,\omega^2,1)M(x,t,\lambda)]_{13}-\frac{1}{3}\lim_{\lambda\rightarrow0}\ln[(\omega^2,\omega,1)M(x,t,\lambda)]_{13}.
			\end{aligned}\right.
		\end{equation}
	\end{theo}
	
	\begin{lem}
		Under the condition that Theorem \ref{reconstruction formula} holds, the solution $M(x,t,\lambda)$ of RH problem \ref{rhp} is unique, if it exits.
	\end{lem}
	\begin{theo}\label{region}
		Under the conditions in Theorem \ref{reconstruction formula} and reconstruction formula \eqref{7}, the solution to the initial value problem of the two-component nonlinear KG equation \eqref{3} exhibits different asymptotic behaviors in distinct regions with the division of regions as shown in Fig. \ref{asy}:
		\begin{description}
			\item[Region I:] In this region, $\vert\frac{x}{t}\vert\rightarrow\infty$, and when $\vert x\vert>>t$, the solutions $u(x,t)$ and $v(x,t)$ decay rapidly to zero and behave as $\mathcal{O}(|x|^{-N-3/2})$ $(N\geq0)$.
			\item[Region II:] This region satisfies $1\leq\vert\frac{x}{t}\vert<\infty$, and the solutions $u(x,t)$ and $v(x,t)$ decay rapidly to zero with the behavior of $\mathcal{O}(t^{-N-3/2})$ $(N\geq0)$.
			\item[Region III:] In this region, $\vert\frac{x}{t}\vert\rightarrow1^-$,  it represents a transition region, and the leading terms of the solutions $u(x,t)$ and $v(x,t)$ are the same as those in equations \eqref{8} and \eqref{9} below, but they have different error of $\mathcal{O}\left(t^{-N-3/2} +\frac{C_N(\lambda_0)(\ln t)^{N+1}}{t^{(N+2)/2}}\right) $ where $N\geq0$ and $C_N(\lambda_0)$ is a smooth non-negative function that vanishes to any order as $\lambda_0=0$ and $\lambda_0=\infty$.
			\item[Region IV:] In this region, $\vert\frac{x}{t}\vert<1$, the higher-order asymptotic solution to the initial value problem of the two-component nonlinear KG equation \eqref{3} is formulated as
			\begin{align}\label{8}
					u(x,t)=&\frac{1}{2}\ln \!\!\left(\! 1\!\!+\!\!\sum\limits_{0\leq q\leq p\leq N}\!\!\frac{(\ln t)^q}{(bt)^{p/2}}\frac{2{\rm{Re}}\,\left( \omega^2\delta^0_{-\lambda_0}{\rm{e}}^{-\theta_{21}(-\lambda_0)}b_{pq}(x,t)\right) -2{\rm{Re}}\,\left(\omega \delta^0_{\lambda_0}{\rm{e}}^{\theta_{21}(\lambda_0)}a^*_{pq}(x,t)\right)}{\sqrt{\frac{2\sqrt{3}t\lambda_0}{1+\lambda_0^2}}}\right)\nonumber\\
					& +\mathcal{O}\left( \frac{(\ln t)^{N+1}}{t^{(N+2)/2}}\right),
				\end{align}
			\begin{align}\label{9}
					v(x,t)
					&=\frac{1}{6}\ln\!\! \left( \!1\!\!+\!\!\sum\limits_{0\leq q\leq p\leq N}\!\!\frac{(\ln t)^q}{(bt)^{p/2}}\frac{2{\rm{Re}}\,\left( \omega^2\delta^0_{-\lambda_0}{\rm{e}}^{-\theta_{21}(-\lambda_0)}b_{pq}(x,t)\right) -2{\rm{Re}}\,\left(\omega \delta^0_{\lambda_0}{\rm{e}}^{\theta_{21}(\lambda_0)}a^*_{pq}(x,t)\right)}{\sqrt{\frac{2\sqrt{3}t\lambda_0}{1+\lambda_0^2}}}\!\right)\nonumber\\
					& \quad -\frac{1}{3}\ln \!\!\left(\!\! 1\!\!+\!\!\sum\limits_{0\leq q\leq p\leq N}\!\!\frac{(\ln t)^q}{(bt)^{p/2}}\frac{2{\rm{Re}}\,\left( \omega\delta^0_{-\lambda_0}{\rm{e}}^{-\theta_{21}(-\lambda_0)}b_{pq}(x,t)\right)\!-\!2{\rm{Re}}\,\left(\omega^2 \delta^0_{\lambda_0}{\rm{e}}^{\theta_{21}(\lambda_0)}a^*_{pq}(x,t)\right)}{\sqrt{\frac{2\sqrt{3}t\lambda_0}{1+\lambda_0^2}}}\!\right)\nonumber\\
					&\quad +\mathcal{O}\left( \frac{(\ln t)^{N+1}}{t^{(N+2)/2}}\right).
				\end{align}
	The definitions of the symbols involved can be found in Section \ref{sec42}.

					\end{description}
				\end{theo}

	\begin{figure}[htbp]
		\centering
	\begin{tikzpicture}
		\draw[thick] (-4.5,4.5) -- (0,0) -- (4.5,4.5);

		\fill[black!20!blue!20] (-4.5,4.5) -- (0,0) -- (-4.5,0.5) -- cycle;
		\fill[black!20!blue!20] (4.5,4.5) -- (0,0) -- (4.5,0.5) -- cycle;
		\fill[black!20!blue!40] (4,4.5) -- (0,0) -- (-4,4.5) -- cycle;
		
		\draw[very thick,-latex] (-4.7,0) -- (4.7,0) node[right] {$x$};
		\draw[very thick,-latex] (0,0) -- (0,4.7) node[above] {$t$};
	    \node at (0,-0.3) {0};
		\node at (-3,2) {II};
		\node at (3,2) {II};
		\node at (-4,0.2) {I};
		\node at (4,0.2) {I};
		\node at (0,3) {IV};
		\node at (-3.7,3.9) {III};
		\node at (3.7,3.9) {III};
	\end{tikzpicture}
	\caption{Division of asymptotic regions in the upper $(x,t)$-half plane}
	\label{asy}
	\end{figure}
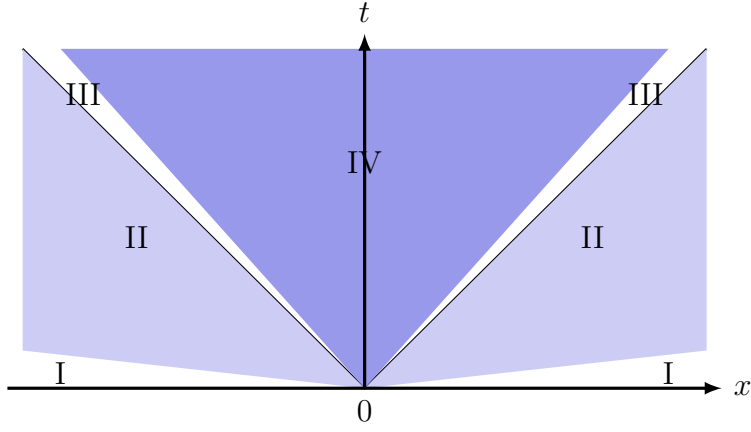

	The above theorem involves results for four regions. The results for Regions ${\rm{I}}$ and ${\rm{II}}$ are jointly provided and proved by Lemma \ref{lem41} and Lemma \ref{lem42}. The result for Region ${\rm{III}}$ is given and proved by Lemma \ref{RegionIII}. The result for the last region, Region ${\rm{IV}}$, is discussed in detail and rigorously in Section \ref{sec42} and is presented in Lemma \ref{RegionIV}.
	
	\begin{lem}
		As $\left| \frac{x}{t}\right|\to1 $, the solutions $u(x,t)$ and $v(x,t)$ in Theorem \ref{region} from the inside of the light cone in Region ${\rm{III}}$ can be matched to the long-time asymptotics in Region ${\rm{II}}$.
	\end{lem}
	\begin{proof}
		The method in Ref. \cite{tzitzeica-2024} can directly provide the proof of this lemma.
	\end{proof}
	\ \ \ \ For parameter $N=0$, the specific expressions for the long-time asymptotic behavior of the solutions \eqref{8} and \eqref{9} to the two-component nonlinear KG equation \eqref{3} are presented as follows:
	\begin{align}\label{N01}
		u(x,t)
		=&\frac{1}{2}\ln\!  \left(\! 1+3^{-\frac{1}{4}}\sqrt{\frac{2\left( 1+\lambda_0^2\right) }{t\lambda_0}}\left( \sqrt{\nu_1}\cos\left( \frac{5\pi}{12}-\frac{2\sqrt{3}t\lambda_0}{1+\lambda_0^2}-\nu_1\ln \left( \frac{6\sqrt{3}t\lambda_0}{1+\lambda_0^2}\right) \!\!+m_1\right)\right.\right. \nonumber\\
		&\left.\left.-\sqrt{\nu_4}\cos\!\left(\! \frac{13\pi}{12}\!-\!\frac{2\sqrt{3}t\lambda_0}{1+\lambda_0^2}\!-\!\nu_4\ln \left( \frac{6\sqrt{3}t\lambda_0}{1+\lambda_0^2}\right) +m_2  \right) \right)  \right)   +\mathcal{O}\left( \frac{\ln t}{t}\right),	
	\end{align}
	\begin{align}\label{N02}
		v(x,t)
		=&\frac{1}{6}\ln\!\! \left( 1+3^{-\frac{1}{4}}\sqrt{\frac{2\left( 1+\lambda_0^2\right) }{t\lambda_0}}\left( \sqrt{\nu_1}\cos\left( \frac{5\pi}{12}-\frac{2\sqrt{3}t\lambda_0}{1+\lambda_0^2}-\nu_1\ln \left( \frac{6\sqrt{3}t\lambda_0}{1+\lambda_0^2}\right) +m_1\right) \right.\right.\nonumber\\
		&\left.\left.-\sqrt{\nu_4}\cos\left( \frac{13\pi}{12}-\frac{2\sqrt{3}t\lambda_0}{1+\lambda_0^2}-\nu_4\ln \left( \frac{6\sqrt{3}t\lambda_0}{1+\lambda_0^2}\right) +m_2\right)\right)  \right) \nonumber\\
		&-\frac{1}{3}\ln \left(\! 1\!+\!3^{-\frac{1}{4}}\sqrt{\frac{2\left( 1+\lambda_0^2\right) }{t\lambda_0}}\!\left(\! \sqrt{\nu_1}\cos\left( \frac{13\pi}{12}\!-\!\frac{2\sqrt{3}t\lambda_0}{1+\lambda_0^2}\!-\!\nu_1\ln \left( \frac{6\sqrt{3}t\lambda_0}{1+\lambda_0^2}\right) \!+\!m_1\right) \right.\right.\nonumber\\
		&\left.\left.-\sqrt{\nu_4}\cos\left( \frac{5\pi}{12}-\frac{2\sqrt{3}t\lambda_0}{1+\lambda_0^2}-\nu_4\ln \left( \frac{6\sqrt{3}t\lambda_0}{1+\lambda_0^2}\right) +m_2\right) \right) \right) +\mathcal{O}\left( \frac{\ln t}{t}\right) 	,
	\end{align}
	where
	\begin{equation*}
		\begin{aligned}
			m_1=&-\left( \arg r_1(\lambda_0)+\arg \Gamma(-i\nu_1)+\nu_4\ln 4\right) +\frac{1}{\pi}\int_{0}^{-\lambda_0}\log_0\frac{\vert s-\omega \lambda_0\vert}{\vert s-\lambda_0\vert}{\rm{d}} \ln(1-\vert r_2(s)\vert^2)\\
			&+\frac{1}{\pi}\int_{0}^{\lambda_0}\log_{-\pi}\frac{\vert s-\lambda_0\vert}{\vert s-\omega \lambda_0\vert}{\rm{d}} \ln(1-\vert r_1(s)\vert^2),\\
		\end{aligned}
	\end{equation*}
	\begin{equation*}
		\begin{aligned}
			m_2=&-(\arg r_2(-\lambda_0)+\arg \Gamma(-i\nu_4)+\nu_1\ln 4)+\frac{1}{\pi}\int_{0}^{\lambda_0}\log_{-\pi}\frac{\vert s+\omega \lambda_0\vert}{\vert s+\lambda_0\vert}{\rm{d}} \ln(1-\vert r_1(s)\vert^2)\\
			&+\frac{1}{\pi}\int_{0}^{-\lambda_0}\log_0\frac{\vert s+\lambda_0\vert}{\vert s+\omega \lambda_0\vert}{\rm{d}} \ln(1-\vert r_2(s)\vert^2),
		\end{aligned}
	\end{equation*}
	where $\lambda_0$ and $\nu_1, \nu_4$ are defined by \eqref{lambda0} and \eqref{nu14}, respectively. Next, we provide the following numerical verifications and further validate the accuracy of the asymptotic solution we have given. 
	\begin{unnumberednum*}
			To substantiate the precision of Theorem \ref{region}, we examine an initial value problem exemplified by a Gaussian wave packet of the form:
			\begin{equation}\label{10}\left\lbrace
				\begin{aligned}
					&u(x,0)=-0.1{\rm{e}}^{-\frac{x^2}{2}},\quad &&u_t(x,0)=0,\\
					&v(x,0)=-0.05{\rm{e}}^{-\frac{x^2}{2}},\quad &&v_t(x,0)=0,\\
				\end{aligned}\right.
			\end{equation}
which satisfies the Assumption \ref{assu1}. Fig. \ref{figrhp-direct} illustrates the comparisons between the asymptotic solutions \eqref{N01} and \eqref{N02} and the results obtained from numerical simulations by using the initial conditions specified in \eqref{10} at time $t=50$ and $t=100$, respectively. The figures confirm that the long-time asymptotic solutions approximate the numerical solutions within an acceptable margin of error. Furthermore, a detailed examination of Fig. \ref{figrhp-direct} indicates that within Regions {\rm{I}} and {\rm{II}}, the solution $u(x,t)$ and $v(x,t)$ tend toward zero when $\vert x\vert>t$. This trend is in alignment with the theoretical predictions for Regions {\rm{I}} and {\rm{II}}.
			\begin{figure}[ht]
				\centering
				\begin{subfigure}[b]{0.49\textwidth}
					\includegraphics[width=\textwidth]{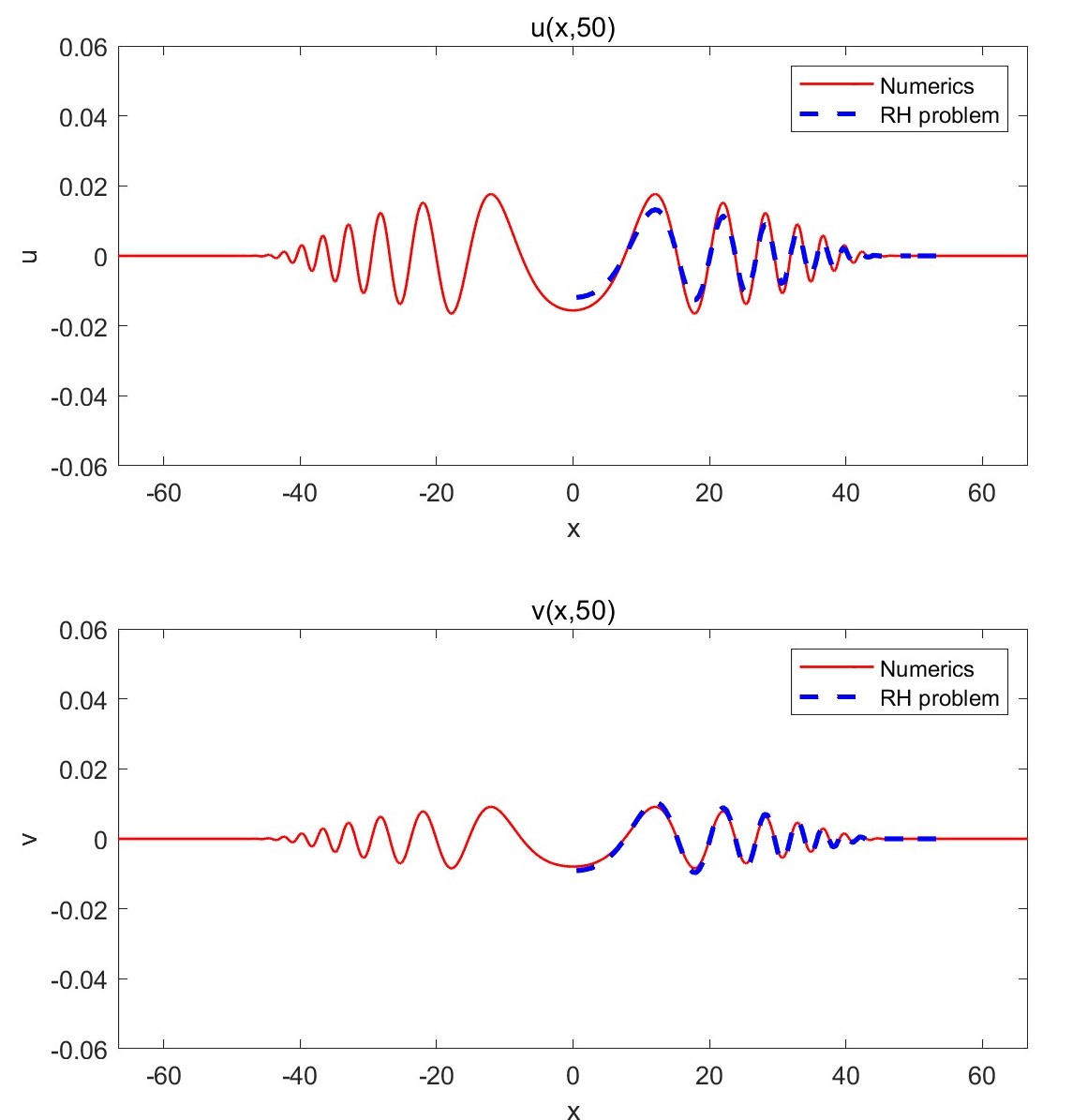}
					\label{figrhp-direct50}
				\end{subfigure}
				\begin{subfigure}[b]{0.48\textwidth}
					\includegraphics[width=\textwidth]{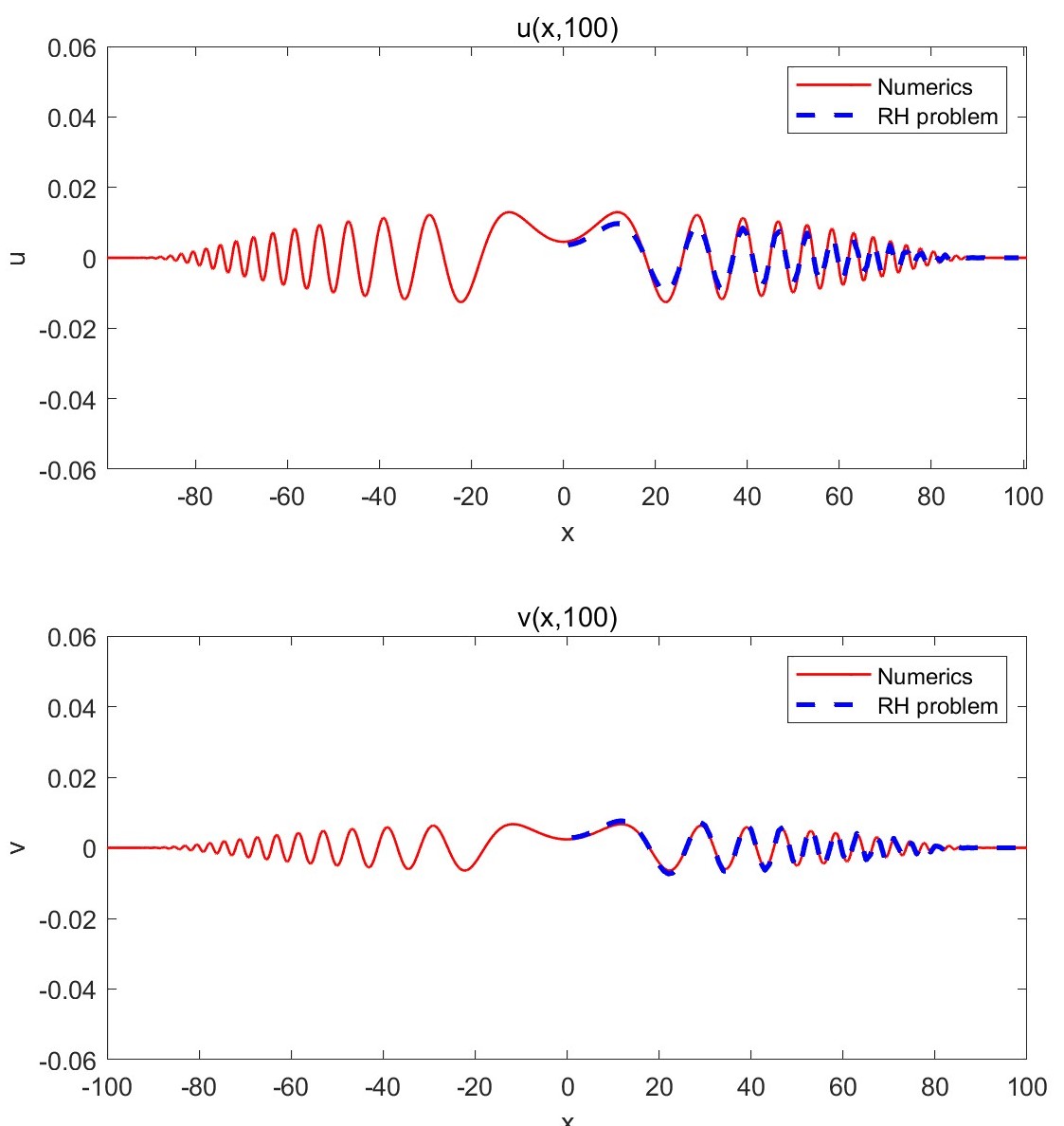}
					\label{figrhp-direct100}
				\end{subfigure}
				\caption{The comparisons between the theoretical results given by Theorem \ref{region} and the full numerical simulations of the two-component nonlinear KG equation \eqref{3} under the initial value condition \eqref{10} for $t=50$ (left) and $t=100$ (right)}
				\label{figrhp-direct}
			\end{figure}

			In summary, the numerical analysis confirms Theorem \ref{region}, highlighting the reliability and accuracy of its asymptotic formulations.
			
		\end{unnumberednum*}

	\section{Spectral analysis}\label{sec3}
		\quad\quad
		In this section, the spectral analysis and inverse scattering transform of the two-component nonlinear KG equation \eqref{3} are discussed based on its Lax pair \eqref{2}.
	\subsection{Direct scattering problem}
	\quad\quad
	In order to treat the three linearly independent solutions of the Lax pair \eqref{2} simultaneously, formulate it in matrix form as
	\begin{equation}\label{11}
		\left\{
		\begin{aligned}
			\tilde{X}_x&=\tilde{L}\tilde{X},\\
			\tilde{X}_t&=\tilde{A}\tilde{X},
		\end{aligned}
		\right.
	\end{equation}
	where $\tilde{X}(x,t,\lambda)$ is a $3\times 3$ matrix-valued function consisting of three linearly independent solutions of the Lax pair \eqref{2}. Introduce a gauge transformation $\tilde{X}(x,t,\lambda)=P\check{X}(x,t,\lambda)$, where the matrix $P$ is given by
	\begin{equation*}\begin{aligned}
		P=\begin{pmatrix}
			\omega  &\omega^2   &1\\
			\omega^2    &\omega &1\\
			1   &1  &1
		\end{pmatrix},\quad \omega=\text{e}^{\frac{2}{3}\pi i},
	\end{aligned}\end{equation*}
then the Lax pair \eqref{11} becomes
	\begin{equation}\label{12}
		\left\{
		\begin{aligned}
			\check{X}_x&={L}\check{X},\\
			\check{X}_t&={A}\check{X},
		\end{aligned}
		\right.
	\end{equation}
where the matrices $L$ and $A$ satisfy the relations $L=P^{-1}\tilde{L}P$ and $A=P^{-1}\tilde{A}P$, and take the forms
	\begin{equation*}\begin{aligned}
		L(x,t,\lambda)&=\frac{1}{2}\left( \lambda J+\frac{1}{\lambda}J^2\right) +\frac{1}{6\lambda}M+N,
		\\
		A(x,t,\lambda)&=\frac{1}{2}\left( \lambda J-\frac{1}{\lambda}J^2\right) -\frac{1}{6\lambda}M+N,
	\end{aligned}\end{equation*}
	where
	
	\begin{align*}
		J&=\begin{pmatrix}
			\omega   &0  &0\\
			0   &\omega^2    &0\\
			0&0&1
		\end{pmatrix},\quad
		N=\frac{u_t+u_x}{6}
		\begin{pmatrix}
			0   &   \omega-1    &\omega^2-1\\
			\omega^2-1  &0  &\omega-1\\
			\omega-1    &   \omega^2-1  &0
		\end{pmatrix}-
		\frac{v_t+v_x}{2}
		\begin{pmatrix}
			0   &   \omega^2  &\omega\\
			\omega    &0  &\omega^2\\
			\omega^2  &\omega   &0
		\end{pmatrix},\\
		M&=({\rm{e}}^{-u-3v}+\omega {\rm{e}}^{3v-u}+\omega^2{\rm{e}}^{2u})
		\begin{pmatrix}
			0   &1  &0\\
			0   &0  &\omega^2\\
			\omega &0   &0
		\end{pmatrix}+
		(\omega {\rm{e}}^{-u-3v}+ {\rm{e}}^{3v-u}+\omega^2{\rm{e}}^{2u})
		\begin{pmatrix}
			0   &0  &1\\
			\omega^2    &0  &0\\
			0   &\omega&0
		\end{pmatrix}\\
		&\quad+({\rm{e}}^{-u-3v}+{\rm{e}}^{3v-u}+{\rm{e}}^{2u}-3)J^2.
	\end{align*}

	Denote
	\begin{equation*}\begin{aligned}
		\mathfrak{L}=\frac{1}{2}\left(\lambda J+\frac{1}{\lambda}J^2\right),\quad
		\mathfrak{A}=\frac{1}{2}\left(\lambda J-\frac{1}{\lambda}J^2\right),\quad
	\end{aligned}\end{equation*}
	then the diagonal entries of $\mathfrak{L}={\rm diag} (l_1,l_2,l_3)$ and $\mathfrak{A}={\rm diag} (a_1,a_2,a_3)$ are given by
	\begin{equation*}
		l_j(\lambda)=\frac{\omega^j \lambda+(\omega^j \lambda)^{-1}}{2
	},\quad a_j(\lambda)=\frac{\omega^j \lambda-(\omega^j \lambda)^{-1}}{2},\quad j=1,2,3.
	\end{equation*}
	So the matrix-valued functions $L$ and $A$ can be written as
	\begin{equation*}
	L:=\mathfrak{L}+L_1,\quad 	A:=\mathfrak{A}+A_1,
	\end{equation*}
	where
	\begin{equation*}\begin{aligned}
		L_1&=N+\frac{1}{6\lambda}M,\\
		A_1&=N-\frac{1}{6\lambda}M.
	\end{aligned}\end{equation*}
	According to $u_0$, $v_0 \in \mathcal{S}(\mathbb{R})$, it is directly verified that
	\begin{equation*}\begin{aligned}
		\lim_{\left|x\right|\to \infty} L_1=O,\quad \lim_{\left|x\right|\to \infty} A_1=O.
	\end{aligned}\end{equation*}
\par
Obviously, the matrices $L$ and $A$ satisfy the $\mathbb{Z}_3$ symmetry
	\begin{equation*}\begin{aligned}
		L(k)=\tilde{\sigma}_1^{-1}L(\omega \lambda) \tilde{\sigma}_1, \quad
		A(k)=\tilde{\sigma}_1^{-1}A(\omega \lambda) \tilde{\sigma}_1, \quad
		\tilde{\sigma}_1=\begin{pmatrix}
			0   &   1   &0\\
			0   &   0   &1\\
			1   &   0   &0
		\end{pmatrix},
	\end{aligned}\end{equation*}
	and the $\mathbb{Z}_2$ symmetry
	\begin{equation*}\begin{aligned}
		L(k)=\tilde{\sigma}_2\overline{L(\overline{\lambda})} \tilde{\sigma}_2^{-1}, \quad
		A(k)=\tilde{\sigma}_2\overline{A(\overline{\lambda})} \tilde{\sigma}_2^{-1}, \quad
	\tilde{\sigma}_2=\begin{pmatrix}
			0   &   1   &0\\
			1   &   0   &0\\
			0   &   0   &1
		\end{pmatrix}.
	\end{aligned}\end{equation*}
	
	Before proceeding to further analysis, introduce the eigenfunction $X$ by the following transformation
	\begin{equation*}\begin{aligned}
		\check{X}=X {\rm{e}}^{\mathfrak{L}x+\mathfrak{A}t},
	\end{aligned}\end{equation*}
	 and substitute it into \eqref{12}. Then we obtain that
	\begin{equation}\label{13}
		\left\{
		\begin{aligned}
			X_x-[\mathfrak{L},X]&=L_1X,\\
			X_t-[\mathfrak{A},X]&=A_1X.
		\end{aligned}
		\right.
	\end{equation}
\subsubsection{The eigenfunctions $X_{\pm}(x,\lambda)$}
\ \ \ \
	Now fix $t=0$ and denote $L(x,\lambda)$ in place of $L(x,0,\lambda)$ for notational simplicity. Then analyze the $x$-part of the Lax pair \eqref{13} for $t=0$:
	\begin{equation}\begin{aligned}\label{14}
		X_x-[\mathfrak{L},X]&=L_1X,
	\end{aligned}\end{equation}
whose solution is formalized by introducing two matrix-valued function $X_+(x,\lambda)$  and $X_-(x,\lambda)$ which are the solutions of the linear Volterra integral equations
	\begin{equation}\label{15}
		\begin{aligned}
			X_+(x,\lambda)&=I-\int_x^{\infty} {\rm{e}}^{(x-y)\widehat{\mathfrak{L}(\lambda)}}(L_1X_+)(y,\lambda)\rd y,\\
			X_-(x,\lambda)&=I+\int_{-\infty}^{x} {\rm{e}}^{(x-y)\widehat{\mathfrak{L}(\lambda)}}(L_1	X_-)(y,\lambda)\rd y,
			\end{aligned}
	\end{equation}
	where ${\rm{e}}^{\hat{\mathfrak{L}}}$ is an operator acting on the $3\times 3$ matrix by ${\rm{e}}^{\hat{\mathfrak{L}}}A={\rm{e}}^{{\mathfrak{L}}}A{\rm{e}}^{{-\mathfrak{L}}}$.
	
	Decompose the complex $\lambda$-plane by $\Sigma:=\{\mathbb{R},\omega\mathbb{R},\omega^2\mathbb{R}\}$,
	then the contour $\Sigma$ divides the plane $\mathbb{C}$ into the six open subsets $\{D_n\}^6_1$ as shown in Fig. \ref{figSigma}.
	
    \begin{figure}[htbp]
    	\centering
    	\begin{tikzpicture}[scale=1]
    		\draw[very thick, black!20!blue, -latex] (0,0) -- (2,0);
    		\draw [very thick,black!20!blue](0,0) -- (3,0);
    		\draw[very thick, black!20!blue, latex-] (-2,0) -- (0,0);
    		\draw [very thick,black!20!blue](-3,0) -- (0,0);
    		\draw[very thick, black!20!blue, -latex] (0,0) -- (1,1.732);
    		\draw [very thick,black!20!blue](0,0) -- (1.5,1.5*1.732);
    		\draw[very thick, black!20!blue, latex-] (-1,-1.732) -- (0,0);
    		\draw [very thick,black!20!blue](-1.5,-1.5*1.732) -- (0,0);
    		\draw[very thick, black!20!blue, -latex] (0,0) -- (-1,1.732);
    		\draw [very thick,black!20!blue](0,0) -- (-1.5,1.732*1.5);
    		\draw[very thick, black!20!blue, latex-] (1,-1.732) -- (0,0);
    		\draw [very thick,black!20!blue](0,0) -- (1.5,-1.5*1.732);
    		\node[red!70!black,below] at (1.8,0) {$1$};
    		\node[red!70!black,right] at (1,1.5) {$2$};
    		\node[red!70!black,left] at (-1,1.5) {$3$};
    		\node[red!70!black,below] at (-1.8,0) {$4$};
    		\node[red!70!black,left] at (-1,-1.5) {$5$};
    		\node[red!70!black,right] at (1,-1.5) {$6$};
    		\node[right] at (3,0) {$\Sigma$};
    		\draw (0:0.4cm) arc (0:60:0.4cm);
    		\node[right] at (0.25,0.2) {$\pi/3$};
    		\node at (1.732,0.8) {$D_1$};
    		\node[right] at (-0.3,1.7) {$D_2$};
    		\node[left] at (-1.4,0.8) {$D_3$};
    		\node[below] at (-1.6,-0.7) {$D_4$};
    		\node[left] at (0.2,-1.7) {$D_5$};
    		\node[right] at (1.5,-0.9) {$D_6$};
    	\end{tikzpicture}
    	\caption{The jump contour $\Sigma$ decomposes the complex $\lambda$-plane into six parts}
    	\label{figSigma}
    \end{figure}
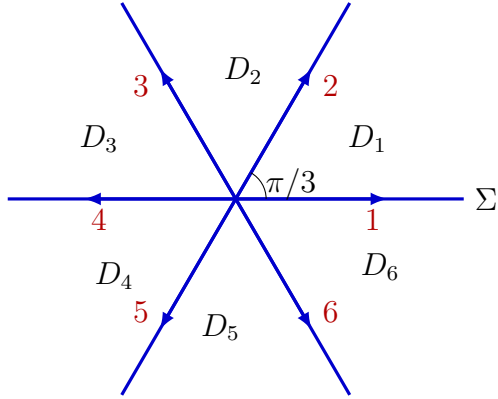
	Further let $S=\{\lambda\in\mathbb{C}|\arg \lambda\in (\frac{2\pi}{3},\frac{4\pi}{3})\}$ denote the interior of $\bar{D}_3\cup\bar{D}_4$, where $\bar{D}_j$ means the closure of set $D_j$ for $j=3,4$.
	\begin{prop}
		$($Basic properties of the eigenfunctions $X_{\pm}(x,\lambda)$$)$. Suppose $u_0,v_0\in \mathcal{S}(\mathbb{R})$, then the equations \eqref{15} uniquely define two solutions $X_\pm(x,\lambda)$ in the equation \eqref{13} with the following properties:
		\begin{enumerate}
			\item The function $X_+(x,\lambda)$ is well-defined for $x\in \mathbb{R}$ and $\lambda\in (\omega^2 \bar{S},\omega \bar{S},\bar{S})\setminus\{0\}$. For each $\lambda\in (\omega^2\bar{S},\omega\bar{S},\bar{S})\setminus\{0\}$, the function $X_+(\cdot,\lambda)$ is smooth and satisfies \eqref{14}.
            Similarly, $X_-(x,\lambda)$ is well-defined for $x\in \mathbb{R}$ and $\lambda\in (-\omega^2 \bar{S},-\omega \bar{S},-\bar{S})\setminus\{0\}$. For each $\lambda\in (-\omega^2\bar{S},-\omega\bar{S},-\bar{S})\setminus\{0\}$, $X_-(\cdot,\lambda)$ is also smooth and satisfies \eqref{14}.
			\item For each $x\in\mathbb{R}$, the function $X_+(x,\cdot)$ is continuous for $\lambda\in(\omega^2 \bar{S},\omega \bar{S},\bar{S})\setminus\{0\}$ and analytic for $\lambda\in (\omega^2 S,\omega S,S)\setminus\{0\}$. The function $X_-(x,\cdot)$ is continuous for $\lambda\in(-\omega^2 \bar{S},-\omega \bar{S},-\bar{S})\setminus\{0\}$ and analytic for $\lambda\in (-\omega^2 S,-\omega S,-S)\setminus\{0\}$.
			\item For each $x\in \mathbb{R}$ and $ j \in \mathbb{N}_+$, the partial derivative $\frac{\partial^j X_+(x,\cdot)}{\partial \lambda^j}$ has a continuous extension to $(\omega^2\bar{S},\omega \bar{S},\bar{S})\setminus \{0\}$, and  $\frac{\partial^j X_-(x,\cdot)}{\partial \lambda^j}$ has a continuous extension to $(-\omega^2\bar{S},-\omega \bar{S},-\bar{S})\setminus \{0\}$.
			\item For each $n\geq 1$ and $\epsilon > 0$, there exist functions $f_+(x)$ and $f_-(x)$ that are smooth, positive, and have bounded support over $\mathbb{R}$ with rapidly decay as $x\rightarrow+\infty$ and
			$x\rightarrow-\infty$, respectively. Under these conditions,  the following inequalities hold for $x\in \mathbb{R}$ and $j=0,1,\dots, n$:
     	\begin{equation*}\begin{aligned}
				&\left|\frac{\partial^j }{\partial \lambda^j} \left(X_+(x,\lambda)-I\right)\right|\leq f_+(x),\quad \lambda\in(\omega^2 \bar{S},\omega \bar{S},\bar{S}),\quad |\lambda|>\epsilon,\\
				&\left|\frac{\partial^j }{\partial \lambda^j} \left(X_-(x,\lambda)-I\right)\right|\leq f_-(x),\quad \lambda\in(-\omega^2 \bar{S},-\omega \bar{S},-\bar{S}),\quad |\lambda|>\epsilon.
			\end{aligned}\end{equation*}
			\item The eigenfunctions $X_+(x,\lambda)$ and $X_-(x,\lambda)$ satisfy the following symmetry relationships
			\begin{equation*}
				X_{\pm}(x,\lambda)=\tilde{\sigma}_1^{-1}X_{\pm}(x,\omega\lambda)\tilde{\sigma}_1=\tilde{\sigma}_2\overline{X_{\pm}(x,\overline{\lambda})}\tilde{\sigma}_2^{-1},\quad \lambda\in({\pm}\omega^2\bar{S},{\pm}\omega\bar{S},{\pm}\bar{S})\setminus\{0\}.			
			\end{equation*}
			\item Assume $u_0(x)$ and $v_0(x)$ are of compact support, for each $x\in\mathbb{R}$, then the functions $X_\pm(x,\lambda)$ are defined, analytic for $\lambda\in\mathbb{C}\setminus\{0\}$, and $\det X_+=\det X_-=1$.
		\end{enumerate}
	\end{prop}
	\begin{proof}
		The proof of this proposition is achieved through a direct examination of the Volterra integral equation presented in \eqref{15}.
	\end{proof}

		\subsubsection{Asymptotics of the eigenfunctions $X_+(x,\lambda)$ and $X_-(x,\lambda)$ as $\lambda\to \infty$}
		\quad\quad
		In this subsection, we give the behaviors of the eigenfunctions $X_+(x,\lambda)$ and $X_-(x,\lambda)$ as $\lambda\to \infty$. For the equation \eqref{14}, perform the Wentzel-Kramers-Brillouin (WKB) expansion of $X$ as $\lambda\to \infty$, then the following expansions of power series are gotten
		\begin{equation}\begin{aligned}\label{16}
				X_{\pm}(x,\lambda)=I+\frac{X_{\pm}^{(1)}(x)}{\lambda}+\frac{X_{\pm}^{(2)}(x)}{\lambda^2}+\cdots,
		\end{aligned}\end{equation}
		where the coefficients $\{X^{(j)}_{\pm}\}_1^\infty$ satisfy $
		\lim\limits_{x\to\pm\infty}X^{(j)}_{\pm}(x)=0$ for $j\geq 1.$
		Since it satisfies equation \eqref{14}, substitute equation \eqref{16} into \eqref{14} to get the recursive formula
     	\begin{equation}\label{17}
    	\left\{
    	\begin{aligned}
		&\bigg[\frac{J}{2},X_{\pm}^{(n+1)}\bigg]+\bigg[\frac{J^2}{2},X_{\pm}^{(n-1)}\bigg]=(\partial_x X_{\pm}^{(n)})^{(o)}-\frac{1}{6}(MX_{\pm}^{(n-1)})^{(o)}-(NX_{\pm}^{(n)})^{(o)},\\
		&(\partial_x X_{\pm}^{(n)})^{(d)}=\frac{1}{6}(MX_{\pm}^{(n-1)})^{(d)}
		+(NX_{\pm}^{(n)})^{(d)},
    	\end{aligned}
     	\right.
       \end{equation}
       where $X_{\pm}^{(-1)}=O$ and $X_{\pm}^{(0)}=I$, the notations $(A)^{(d)}$ and $(A)^{(o)}$ represent the diagonal and non-diagonal parts of the matrix $A$, respectively. The coefficients $\left\lbrace X_{\pm}^{(j)}\right\rbrace $ are uniquely determined recursively from \eqref{17}, and
       \begin{equation*}\begin{aligned}
       			\left( X_{\pm}^{(1)}\right) ^{(o)}=&\frac{1}{6}\!\!\left(  \frac{\partial}{\partial t}u +\frac{\partial}{\partial x}u\right) \!\!\left(
       		\begin{array}{ccc}
       			0 & \omega^{2} & \omega^{2}
       			\\
       			\omega  & 0 & \omega
       			\\
       			1 & 1 & 0
       		\end{array}\right)\!
       		-\!\frac{1}{6}\!\!\left( \frac{\partial}{\partial t}v +\frac{\partial}{\partial x}v\right) \!\!\left(
       		\begin{array}{ccc}
       			0 & 1-\omega  & \omega -1
       			\\
       			1-\omega^{2} & 0 & \omega^{2}-1
       			\\
       			1+2\omega  & -1-2\omega  & 0
       		\end{array}\right),\\
       		\left( X_{\pm}^{(1)}\right) ^{(d)}=&\frac{1}{6}\int_{-\infty}^x({\mathrm e}^{-u_0(x^{\prime})-3 v_0(x^{\prime})}+{\mathrm e}^{3 v_0(x^{\prime})-u_0(x^{\prime})}+{\mathrm e}^{2 u_0(x^{\prime})}-3+\\
       		&(u_{0t}(x^{\prime})+u_{0x}(x^{\prime}))^{2}+3\left(v_{0t}(x^{\prime})+v_{0x}(x^{\prime})\right)^{2})\rd x^{\prime} {J}^{2}.
       		\end{aligned}\end{equation*}
       		
     The subsequent proposition presents the asymptotic properties of $X_\pm(x,\lambda)$ for $\lambda\rightarrow\infty$.
       		
		\begin{prop} $($Asymptotics of $X_+(x,\lambda)$ and $X_-(x,\lambda)$ as $\lambda\to\infty$$)$.
			Assume $u_0(x),\ v_0(x)\in \mathcal{S}(\mathbb{R})$. The matrices $X_{\pm}(x,\lambda)$ align completely with their respective series expansions to any order as $\lambda\rightarrow\infty$. To be more specific, for $p\geq0$, the functions
			\begin{equation*}\begin{aligned}
				&(X_+)_{(p)}(x,\lambda):=I+\frac{X_+^{(1)}(x)}{\lambda}+\cdots+\frac{X_+^{(p)}(x)}{\lambda^p},\\
				&(X_-)_{(p)}(x,\lambda):=I+\frac{X_-^{(1)}(x)}{\lambda}+\cdots+\frac{X_-^{(p)}(x)}{\lambda^p},
			\end{aligned}\end{equation*}
			are well-defined, and for $\forall j\geq 0$, we have
			\begin{equation*}\begin{aligned}
				\left|  \frac{\partial^j}{\partial \lambda^j}\left(X_{+}-(X_{+})_{(p)}\right)\right| &\leq\frac{f_+(x)}{\mid\lambda\mid^{p+1}}
				,\quad x\in\mathbb{R},\ \lambda \in (\omega^2 \bar{S},\omega \bar{S},\bar{S}),\ \mid\lambda\mid\geq 2,\\
				\left|  \frac{\partial^j}{\partial \lambda^j}\left(X_{-}-(X_{-})_{(p)}\right)\right| &\leq\frac{f_-(x)}{\mid\lambda\mid^{p+1}}
				,\quad x\in\mathbb{R},\ \lambda \in (-\omega^2 \bar{S},-\omega \bar{S},-\bar{S}),\ \mid\lambda\mid\geq 2,
			\end{aligned}\end{equation*}
		where $f_+(x)$ and $f_-(x)$ are smooth, positive, and have bounded support over $\mathbb{R}$ with rapidly decay properties as $x\rightarrow+\infty$ and
		$x\rightarrow-\infty$, respectively.
		\end{prop}
		\begin{proof}
			The proof here is analogous to that of the good Boussinesq equation and the sine-Gordon equation, referring to references \cite{good-boussinesq} and \cite{lenells-sG}. For the sake of brevity, further elaboration is omitted here.
		\end{proof}
		
		\subsubsection{Asymptotics of the eigenfunctions $X_+(x,\lambda)$ and $X_-(x,\lambda)$ as $\lambda\to 0$}
		\quad\quad
		It is complicated that the kernel of the Jost function has a singularity at $\lambda =0$, which will cause bad behavior and the framework of the Volterra integral equation is not appropriate for analyzing the behavior of the Jost function for $\lambda\rightarrow 0$. To solve this problem, take a gauge transformation $\check{X}(x,t,\lambda)=G(x,t)\hat{{X}}(x,t,\lambda)$, where
		\begin{equation*}\begin{aligned}
			G(x,t)=
			\frac{{\rm{e}}^{2u}}{3}\begin{pmatrix}
				1	&	\omega	&	\omega^2\\
				\omega^2	&	1	&\omega\\
				\omega	&\omega^2	&1
			\end{pmatrix}+\frac{{\rm{e}}^{u-3v}}{3}\begin{pmatrix}
				1	&\omega^2	&\omega\\
				\omega	&1	&\omega^2\\
				\omega^2	&\omega	&1
			\end{pmatrix}+\frac{1}{3}
			\begin{pmatrix}
				1	&1	&1\\
				1	&1	&1\\
				1	&1	&1
			\end{pmatrix},
		\end{aligned}\end{equation*}
then the new Lax pair is derived in the following form
		\begin{equation*}\begin{aligned}
			\hat{X}_x&=G^{-1}(LG-G_x)\hat{X}=\widehat{L}\hat{X},\\
			\hat{X}_t&=G^{-1}(AG-G_t)\hat{X}=\widehat{A}\hat{X},
		\end{aligned}\end{equation*}
where the expressions for $\widehat{L}$ and $\widehat{A}$ are as follows
		\begin{equation*}\begin{aligned}
			\widehat{L}=
			&\frac{1}{2}\left(\lambda J+\frac{1}{\lambda}J^2\right)+\widehat{N}^{(x)}+\frac{\lambda}{6} \widehat{M},\\
			\widehat{A}=
			&\frac{1}{2}\left(\lambda J-\frac{1}{\lambda}J^2\right)+\widehat{N}^{(t)}-\frac{\lambda}{6} \widehat{M},
		\end{aligned}\end{equation*}
		where
		
		\begin{equation*}\begin{aligned}
			\widehat{N}^{(x)}&=N-\partial_x\begin{pmatrix}
				u-v	&\frac{\omega-1}{3}u-\omega^2v	&\frac{\omega^2-1}{3}u-\omega v\\
				\frac{\omega^2-1}{3}u-\omega v	&u-v	&\frac{\omega-1}{3}u-\omega^2v\\
				\frac{\omega-1}{3}u-\omega^2v	&\frac{\omega^2-1}{3}u-\omega v	&u-v
			\end{pmatrix},\\
			\widehat{N}^{(t)}&=N-\partial_t\begin{pmatrix}
				u-v	&\frac{\omega-1}{3}u-\omega^2v	&\frac{\omega^2-1}{3}u-\omega v\\
				\frac{\omega^2-1}{3}u-\omega v	&u-v	&\frac{\omega-1}{3}u-\omega^2v\\
				\frac{\omega-1}{3}u-\omega^2v	&\frac{\omega^2-1}{3}u-\omega v	&u-v
			\end{pmatrix},\\
			\widehat{M}&= (\omega^2{\rm{e}}^{2u}+\omega {\rm{e}}^{-u+3v}+{\rm{e}}^{-u-3v})\!\!
			\begin{pmatrix}
				0   &1  &0\\
				0   &0  &\omega\\
				\omega^2 &0   &0
			\end{pmatrix}\!+\!
			({\rm{e}}^{2u}+\omega {\rm{e}}^{-u+3v}+\omega^2 {\rm{e}}^{-u-3v})\!\!
			\begin{pmatrix}
				0   &0  &1\\
				\omega    &0  &0\\
				0   &\omega^2&0
			\end{pmatrix}\\
			&\quad +({\rm{e}}^{2u}+{\rm{e}}^{-u+3v}+{\rm{e}}^{-u-3v}-3)J.
		\end{aligned}\end{equation*}
		Similarly, we have
		\begin{equation*}\begin{aligned}
			\widehat{\mathfrak{L}}=\lim_{\left|x\right|\to \infty} \widehat{L}=\frac{1}{2}\left(\lambda J+\frac{1}{\lambda}J^2\right),\quad
			\widehat{\mathfrak{A}}=\lim_{\left|x\right|\to \infty} \widehat{A}=\frac{1}{2}\left(\lambda J-\frac{1}{\lambda}J^2\right),\quad
		\end{aligned}\end{equation*}
		together with
		\begin{equation*}\begin{aligned}
			\widehat{L_1}&\equiv \widehat{L}-\widehat{\mathfrak{L}}=\widehat{N}^{(x)}+\frac{\lambda}{6}\widehat{M},\\
			\widehat{A_1}&\equiv \widehat{L}-\widehat{\mathfrak{L}}=\widehat{N}^{(t)}-\frac{\lambda}{6}\widehat{M}.
		\end{aligned}\end{equation*}
		
		Then introduce the eigenfunction $Y$ by the  transformation
	$\hat{X}=Y {\rm{e}}^{\widehat{\mathfrak{L}}x+\widehat{\mathfrak{A}}t}$, and
		obtain another new Lax pair:
		\begin{equation}\label{18}
			\left\{
			\begin{aligned}
					Y_x-[\mathfrak{L},Y]&=\widehat{L_1}Y,\\
					Y_t-[\mathfrak{A},Y]&=\widehat{A_1}Y.
				\end{aligned}
			\right.
		\end{equation}
		In the case of fixing $t=0$, the Volterra integral equations of the $x$-part of the Lax pair \eqref{18} are given
			\begin{equation*}\begin{aligned}
				Y_+(x,\lambda)&=I-\int_x^{\infty} {\rm{e}}^{(x-y)\widehat{\mathfrak{L}(\lambda)}}(\widehat{L_1}Y_+)(y,\lambda)\rd y,\\
				Y_-(x,\lambda)&=I+\int_{-\infty}^{x} {\rm{e}}^{(x-y)\widehat{\mathfrak{L}(\lambda)}}(\widehat{L_1}	Y_-)(y,\lambda)\rd y.
			\end{aligned}\end{equation*}
		As we already defined, $X(x,t,\lambda)=G(x,t)Y(x,t,\lambda)$, it is found that $Y(x,\lambda)$ satisfies the spectral problem without singularity at the origin.
		Consider the following expansion for $Y(x,\lambda)$ at $\lambda=0$ as
		\begin{equation*}\begin{aligned}
			Y_{\pm}(x,\lambda)=	Y_{\pm}^{(0)}+\lambda 	Y_{\pm}^{(1)}+\cdots,\quad \lambda\rightarrow 0.
		\end{aligned}\end{equation*}
		By bringing the above expansion into the Lax pair \eqref{18} satisfied by $Y(x,\lambda)$, the recursion relationship of the following expansion coefficients is obtained
			\begin{equation*}
			\left\{
			\begin{aligned}
				&\bigg[\frac{J}{2},Y_{\pm}^{(n-1)}\bigg]+\bigg[\frac{J^2}{2},Y_{\pm}^{(n+1)}\bigg]=(\partial_x Y_{\pm}^{(n)})^{(o)}-\frac{1}{6}(MX_{\pm}^{(n-1)})^{(o)}-(NX_{\pm}^{(n)})^{(o)},\\
				&(\partial_x Y_{\pm}^{(n)})^{(d)}=\frac{1}{6}(MY_{\pm}^{(n-1)})^{(d)}
				+(NY_{\pm}^{(n)})^{(d)},
			\end{aligned}
			\right.
		\end{equation*}
		where $Y_{\pm}^{(-1)}=O$ and $Y_{\pm}^{(0)}=I$. Under this relation, the behavior of $X(x,\lambda)$ as $\lambda\rightarrow0$ will be described by the transformation $X=GY$, with the following expansion
		\begin{equation}\label{19}
			X_{\pm}(x,\lambda)=G(x)+G(x)Y_{\pm}^{(1)}\lambda+G(x)Y_{\pm}^{(2)}\lambda^2\cdots,\quad \lambda\rightarrow0.
		\end{equation}

		\begin{remark}  The function $G(x,t)$ satisfies two properties:
		\begin{itemize}
				\item 	$G(x,t)$ satisfies the following symmetry
				\begin{equation*}
					\tilde{\sigma}_1^{-1}G(x,t)\tilde{\sigma}_1=G(x,t),\quad \tilde{\sigma}_2 \overline{G(x,t)}\tilde{\sigma}_2=G(x,t).
				\end{equation*}
				\item 	It is noted that
				\begin{equation*}\begin{aligned}
					\begin{pmatrix}
						\omega&\omega^2&1
					\end{pmatrix}G={\rm{e}}^{2u}\begin{pmatrix}
						\omega&\omega^2&1
					\end{pmatrix},\quad
					\begin{pmatrix}
						\omega^2&\omega&1
					\end{pmatrix}G={\rm{e}}^{u-3v}\begin{pmatrix}
						\omega^2&\omega&1
					\end{pmatrix}.\\
				\end{aligned}\end{equation*}
				\end{itemize}
				
			According to these properties, the reconstruction formula \eqref{7} of solution to the initial value problem of the two-component nonlinear KG equation \eqref{3} can be obtained:
				\begin{equation*}
				\left\{
				\begin{aligned}
					&u(x,t)=\frac{1}{2}\lim_{\lambda\rightarrow0}\ln[(\omega,\omega^2,1)M(x,t,\lambda)]_{13},\\
					&v(x,t)=\frac{1}{6}\lim_{\lambda\rightarrow0}\ln[(\omega,\omega^2,1)M(x,t,\lambda)]_{13}-\frac{1}{3}\lim_{\lambda\rightarrow0}\ln[(\omega^2,\omega,1)M(x,t,\lambda)]_{13}.
				\end{aligned}\right.
			\end{equation*}
		\end{remark}

		\subsection{The scattering matrix}\ \ \ \
		Assuming that the initial value functions $u_0(x)$ and $v_0(x)$ are of compact support, it follows that a function $s(\lambda)$ exists for $\lambda\in \mathbb{C}\setminus\{0\}$ , which is independent of both $x$ and $t$. Indeed, the subsequent relationship holds
		\begin{equation}\label{20}
			X_+(x,\lambda)=X_-(x,\lambda){\rm{e}}^{x\widehat{\mathfrak{L}(\lambda)}} s(\lambda),\quad \lambda \in \mathbb{C} \backslash \{0\},
		\end{equation}
		which describes the analytic continuation of the scattering data in the complex plane except for the origin. In the broader scenario where $u_0(x)$, $v_0(x)\in\mathcal{S}(\mathbb{R})$, the above mentioned relation \eqref{20} for $s(\lambda)$ still holds in its domain.

	  \begin{prop}\label{prop33}
				$($Basic properties of $s(\lambda)$$)
				$.
				Given that $u_0(x), v_0(x)\in \mathcal{S}(\mathbb{R})$, the matrix-valued function $s(\lambda)$ exhibits the following  properties:
				\begin{enumerate}
					\item The domain of $s(\lambda)$ is
					\begin{equation*}
						\left(\begin{array}{ccc}
						\omega^2\bar{S} & \mathbb{R}_+ & \omega\mathbb{R}_+
						\\
						\mathbb{R}_+ & \omega\bar{S} & 	\omega^2 \mathbb{R}_+
							\\
						\omega \mathbb{R}_+ & \omega^2 \mathbb{R}_+ &\bar{S}
						\end{array}\right)
						\backslash\{0\},
					\end{equation*}
					where $\mathbb{R}_+$ and $ \mathbb{R}_-$ denote the positive and negative real axis, respectively. The matrix-valued function $s(\lambda)$ maintains continuity right up to the boundary of its domain and is analytic throughout the interior of this domain.
					\item  The function $s(\lambda)$ behaves as
					\begin{equation*}
						s(\lambda)=I-\sum_{j=1}^{N}\frac{s_j}{\lambda^j}+\mathcal{O}\left( \frac{1}{\lambda^{N+1}}\right) ,\quad \lambda\rightarrow\infty,
					\end{equation*}
					and
					\begin{equation*}
						s(\lambda)=I+s^{(1)}\lambda+s^{(2)}\lambda^2+\cdots,\quad \lambda\rightarrow0.
					\end{equation*}
					\item  The function $s(\lambda)$ satisfies the symmetries
					\begin{equation*}
						s(\lambda)=\tilde{\sigma}_1^{-1}s(\omega\lambda)\tilde{\sigma}_1=\tilde{\sigma}_2\overline{s(\overline{\lambda})} \tilde{\sigma}_2^{-1}.
					\end{equation*}
				\end{enumerate}
			\end{prop}
			\begin{proof}
				Examining the value of $\lambda$ that resides within the intersection of the pertinent domains, we can discuss the relationship between the two Jost functions
				\begin{equation}\label{21}	X_+(x,\lambda)=X_-(x,\lambda){\rm{e}}^{x\widehat{\mathfrak{L}(\lambda)}} s(\lambda).
				\end{equation}
				It is clear that as $x\rightarrow-\infty$, $X_-(x,\lambda)$ rapidly converges to $I$. So we can represent $s(\lambda)$ by the following integral equation
					\begin{equation}\label{slambda}\begin{aligned}
					s(\lambda)=I-\int_\mathbb{R} {\rm{e}}^{-y\widehat{\mathfrak{L}(\lambda)}}(L_1X_+)(y,\lambda)\rd y.
				\end{aligned}\end{equation}
			By conducting a detailed analysis of ${\rm{e}}^{l_i-l_j}$, it is immediate that $s(\lambda)$ is well-defined for
				\begin{equation*}
					\lambda\in	\left(\begin{array}{ccc}
						\omega^2\bar{S} & \mathbb{R}_+ & \omega\mathbb{R}_+
						\\
						\mathbb{R}_+ & \omega\bar{S} & 	\omega^2 \mathbb{R}_+
						\\
						\omega \mathbb{R}_+ & \omega^2 \mathbb{R}_+ &\bar{S}
					\end{array}\right)
					\backslash\{0\}.
				\end{equation*}
				
				In addition, by analyzing the asymptotic behavior of $X(x,\lambda)$ as $\lambda\rightarrow0$, and setting $x=0$, we can obtain
				\begin{equation*}
					G(0)+G(0)X_+^{(1)}\lambda+\mathcal{O}(\lambda^2)=(G(0)+G(0)X_-^{(1)}\lambda+\mathcal{O}(\lambda^2))s(\lambda),
				\end{equation*}
				which gives the expansion of $s(\lambda)$ as $\lambda\rightarrow0$. In other words, the expansion of $s(\lambda)$ can be given
				\begin{equation*}
					s(\lambda)=I+s^{(1)}\lambda+s^{(2)}\lambda^2+\cdots,\quad \lambda\rightarrow0,
				\end{equation*}
				where $s^{(j)}=(X_-^{(j)})^{-1}X_+^{(j+1)}$, $j=1,2,\cdots$.
			\end{proof}
			\begin{remark}
				According to Proposition \ref{prop33}, the function
				$s(\lambda)$ can be analytically continued to the point  $\lambda=0$.
					\end{remark}
			
			\subsection{The cofactor matrix}
			\quad\quad
			Let $M^A=(M^{-1})^T$ be the cofactor matrix of matrix $M$. Utilizing $(X^A)_x=-X^A(X_x)^TX^A$ and the $x$-part of the Lax pair satisfied by $X(x,\lambda)$ in (\ref{14}), we can get the Lax representation corresponding to the cofactor matrix $X^A=X^A(x,\lambda)$ as
			\begin{equation*}
				X^A_x+[\mathfrak{L},X^A]=-L_1^TX^A.
			\end{equation*}
			Similar to the above discussion of $X(x,\lambda)$ satisfying Lax pair in (\ref{14}), the Volterra integral equations for $X^A(x,\lambda)$ are also expressed below
			\begin{equation}\label{22}
				\begin{aligned}
					&X_+^A(x,\lambda)=I+\int_{x}^{\infty}{\rm{e}}^{-(x-y)\widehat{\mathfrak{L}(\lambda)}}(L_1^TX_+^A)(y,\lambda)\rd y,\\
				   &X_-^A(x,\lambda)=I-\int_{-\infty}^{x}{\rm{e}}^{-(x-y)\widehat{\mathfrak{L}(\lambda)}}(L_1^TX_-^A)(y,\lambda)\rd y.
				\end{aligned}
			\end{equation}
			
		\begin{prop}
		$($Basic properties of the eigenfunctions $X^A_{\pm}(x,\lambda)$$)$. Suppose $u_0(x), v_0(x)\in \mathcal{S}(\mathbb{R})$. Then the equations \eqref{22} uniquely define two $3\times 3$ matrix-valued solutions $X^A_\pm(x,\lambda)$ of \eqref{22} with the following properties:
	\begin{enumerate}
		\item The function $X^A_+(x,\lambda)$ is well-defined for $x\in \mathbb{R}$ and $\lambda\in (-\omega^2 \bar{S},-
		\omega \bar{S},-\bar{S})\setminus\{0\}$. For each $\lambda\in (-\omega^2\bar{S},-\omega\bar{S},-\bar{S})\setminus\{0\}$, $X_+(\cdot,\lambda)$ is smooth and satisfies \eqref{22}. Similarly $X^A_-(x,\lambda)$ is well-defined for $x\in \mathbb{R}$ and $\lambda\in (\omega^2 \bar{S},\omega \bar{S},\bar{S})\setminus\{0\}$. For each $\lambda\in (\omega^2\bar{S},\omega\bar{S},\bar{S})\setminus\{0\}$, $X_-^A(\cdot,\lambda)$ is smooth and satisfies \eqref{22}.
		\item For each $x\in\mathbb{R}$, the function $X^A_+(x, \lambda)$ is continuous for $\lambda\in(-\omega^2 \bar{S},-\omega \bar{S},-\bar{S})\setminus\{0\}$ and analytic for $\lambda\in (-\omega^2 S,-\omega S,-S)\setminus\{0\}$. The function $X^A_-(x, \lambda)$ is continuous for $\lambda\in(\omega^2 \bar{S},\omega \bar{S},\bar{S})\setminus\{0\}$ and analytic for $\lambda\in (\omega^2 S,\omega S,S)\setminus\{0\}$.
		\item For each $x\in \mathbb{R}$ and $j\in \mathbb{N}_+$, the partial derivative $\frac{\partial^j X^A_+(x,\cdot)}{\partial \lambda^j}$ has a continuous extension to $(-\omega^2\bar{S},-\omega \bar{S},-\bar{S})\setminus \{0\}$, and $\frac{\partial^j X^A_-(x,\cdot)}{\partial \lambda^j}$ has a continuous extension to $(\omega^2\bar{S},\omega \bar{S},\bar{S})\setminus \{0\}$.
		\item For each $n\geq 1$ and $\epsilon > 0$, there exist functions $f_+(x)$ and $f_-(x)$ that are smooth, positive, and have bounded support over $\mathbb{R}$ with rapidly decay properties as $x\rightarrow+\infty$ and
		$x\rightarrow-\infty$, respectively. Under this condition, the following inequalities hold for $x\in \mathbb{R}$ and $j=0,1,\dots, n$:
		\begin{equation*}\begin{aligned}
			&\left|\frac{\partial^j }{\partial \lambda^j} \left(X^A_+(x,\lambda)-I\right)\right|\leq f_+(x),\quad \lambda\in(-\omega^2 \bar{S},-\omega \bar{S},-\bar{S}), &&&|\lambda|>\epsilon,\\
			&\left|\frac{\partial^j }{\partial \lambda^j} \left(X^A_-(x,\lambda)-I\right)\right|\leq f_-(x),\quad \lambda\in(\omega^2 \bar{S},\omega \bar{S},\bar{S}), &&&|\lambda|>\epsilon.
		\end{aligned}\end{equation*}
		\item $X^A_+(x,\lambda)$ and $X^A_-(x,\lambda)$ satisfy the following symmetry relationships
		\begin{equation*}
			X^A_{\pm}(x,\lambda)=\tilde{\sigma}_1^{-1}X^A_{\pm}(x,\omega\lambda)\tilde{\sigma}_1=\tilde{\sigma}_2\overline{X^A_{\pm}(x,\overline{\lambda})}\tilde{\sigma}_2^{-1},\quad \lambda\in({\mp}\omega^2\bar{S},{\mp}\omega\bar{S},{\mp}\bar{S})\setminus\{0\}.			
		\end{equation*}
		\item	Assume $u_0(x)$ and $v_0(x)$ are of compact support for each $x\in\mathbb{R}$, then the functions $X^A_\pm(x,\lambda)$ are defined, analytic for $\lambda\in\mathbb{C}\setminus\{0\}$, and $\det X^A_+=\det X^A_-=1$.
	\end{enumerate}
	\end{prop}

	\begin{prop} $($Asymptotics of $X^A_+(x,\lambda)$ and $X^A_-(x,\lambda)$ as $\lambda\to\infty$$)$.
		Assume $u_0(x),\ v_0(x)\in \mathcal{S}(\mathbb{R})$. The matrices $X^A_{\pm}(x,\lambda)$ align completely with their respective series expansions to any order as $\lambda\rightarrow\infty$. To be more specific, for $p\geq0$, the functions
		\begin{equation*}\begin{aligned}
			&(X_+^A)_{(p)}(x,\lambda):=I+\frac{(X^A_+)^{(1)}(x)}{\lambda}+\cdots+\frac{(X^A_+)^{(p)}(x)}{\lambda^p},\\
			&(X^A_-)_{(p)}(x,\lambda):=I+\frac{(X^A_-)^{(1)}(x)}{\lambda}+\cdots+\frac{(X^A_-)^{(p)}(x)}{\lambda^p},
		\end{aligned}\end{equation*}
		are well-defined, and for any $\geq 0$, we have
		\begin{equation*}\begin{aligned}
			\left|  \frac{\partial^j}{\partial \lambda^j}\left(X^A_{+}-(X^A_{+})_{(p)}\right)\right| &\leq\frac{f_+(x)}{\mid\lambda\mid^{p+1}}
			,\quad x\in\mathbb{R},\ \lambda \in (-\omega^2 \bar{S},-\omega \bar{S},-\bar{S}),\ \mid\lambda\mid\geq 2,\\
			\left|  \frac{\partial^j}{\partial \lambda^j}\left(X^A_{-}-(X^A_{-})_{(p)}\right)\right| &\leq\frac{f_-(x)}{\mid\lambda\mid^{p+1}}
			,\quad x\in\mathbb{R},\ \lambda \in (\omega^2 \bar{S},\omega \bar{S},\bar{S}),\ \mid\lambda\mid\geq 2,
		\end{aligned}\end{equation*}
		where $f_+(x)$ and $f_-(x)$ are bounded smooth positive functions of $x\in \mathbb{R}$ with rapidly decay properties as $x\rightarrow+\infty$ and
		$x\rightarrow-\infty$, respectively.
	\end{prop}
	
	Based on the definition of $X^A$, it is natural to get $X^A=G^AY^A$. In other words, the function $Y^A$ is subject to the following equation
	\begin{equation*}
		Y_x^A+[\mathfrak{L},Y^A]=-(G^{-1}LG-G^{-1}G_x-\mathfrak{L})^TY^A=-\widehat{L_1}^TY^A.
	\end{equation*}
	The Volterra integral equation for the cofactor matrix function $Y^A$ are expressed by
	\begin{equation*}
		\begin{aligned}
			&Y_+^A(x,\lambda)=I+\int_{x}^{\infty}{\rm{e}}^{-(x-y)\widehat{\mathfrak{L}(\lambda)}}(\widehat{L_1}^TY_+^A)(y,\lambda)\rd y,\\
			&Y_-^A(x,\lambda)=I-\int_{-\infty}^{x}{\rm{e}}^{-(x-y)\widehat{\mathfrak{L}(\lambda)}}(\widehat{L_1}^TY_-^A)(y,\lambda)\rd y.
		\end{aligned}
	\end{equation*}
	Similar to the analysis of $Y_{\pm}(x,\lambda)$ when $\lambda\rightarrow0$, the expansions for $Y_{\pm}^A(x,\lambda)$ near $\lambda=0$ as follows
	\begin{equation*}
		\begin{aligned}
			Y_{\pm}^A(x,\lambda)&=(Y^A_{\pm})^{(0)}(x)+\lambda(Y^A_{\pm})^{(1)}(x)+\cdots\\
			&=I+\lambda(Y^A_{\pm})^{(1)}(x)+\cdots,\qquad \lambda\rightarrow0.
		\end{aligned}
	\end{equation*}
	
	The analysis of the two scattering matrices $s(\lambda)$ and $s^A(\lambda)$ is similar, and the Jost functions $X^A_+(x,\lambda)$ and $X^A_-(x,\lambda)$ satisfy the following relation
	\begin{equation*}
		X^A_+(x,\lambda)=X^A_-(x,\lambda){\rm{e}}^{-x\widehat{\mathfrak{L}(\lambda)}}s^A(\lambda),\quad \lambda\in \mathbb{C}\setminus\{0\},
	\end{equation*}
	where the scattering matrix $s^A(\lambda)$ is given by
		\begin{equation}\label{sAlambda}\begin{aligned}
		s^A(\lambda)=I+\int_\mathbb{R} {\rm{e}}^{y\widehat{\mathfrak{L}(\lambda)}}(L_1^TX_+^A)(y,\lambda)\rd y.
	\end{aligned}\end{equation}
			
	In the following statement, we briefly give the relevant properties of $s^A(\lambda)$.
	
	\begin{prop}
		$($Basic properties of $s^A(\lambda)$$)$. Suppose $u_0(x), v_0(x)\in \mathcal{S}(\mathbb{R})$. Then $s^A(\lambda)$ has the following properties:
		\begin{enumerate}
			\item The domain of $s^A(\lambda)$ is
			\begin{equation*}
				\left(\begin{array}{ccc}
					-\omega^2\bar{S} & \mathbb{R}_- & \omega\mathbb{R}_-
					\\
					\mathbb{R}_- & -\omega\bar{S} & 	\omega^2 \mathbb{R}_-
					\\
					\omega \mathbb{R}_- & \omega^2 \mathbb{R}_- &-\bar{S}
				\end{array}\right)
				\backslash\{0\}.
			\end{equation*}
		The matrix function $s^A(\lambda)$ maintains continuity right up to the boundary of its domain and is analytic throughout the interior of this domain.
			\item  The function $s(\lambda)$ behaves as
			\begin{equation*}
				s^A(\lambda)=I-\sum_{j=1}^{N}\frac{s^A_j}{\lambda^j}+\mathcal{O}\left( \frac{1}{\lambda^{N+1}}\right) ,\quad \lambda\rightarrow\infty,
			\end{equation*}
			and
			\begin{equation*}
				s^A(\lambda)=I+(s^A)^{(1)}\lambda+(s^A)^{(2)}\lambda^2+\cdots,\quad \lambda\rightarrow0.
			\end{equation*}
			\item  The function $s^A(\lambda)$ satisfies the symmetries
			\begin{equation*}
				s^A(\lambda)=\tilde{\sigma}_1^{-1}s^A(\omega\lambda)\tilde{\sigma}_1=\tilde{\sigma}_2\overline{s^A(\overline{\lambda})} \tilde{\sigma}_2^{-1}.
			\end{equation*}
		\end{enumerate}
	\end{prop}		
			
	\subsection{The eigenfunctions $M_n(x, \lambda)$}		
	\quad\quad
	The objective of this subsection is to construct the piecewise analytic functions $M_n(x, \lambda)$ over the region $D_n$ for $n=1,2,\cdots,6$.
	
	For each $n=1, 2,\cdots,6$, and $\lambda\in D_n$, define a $3\times 3$ matrix-value solution by the following system of Fredholm integral equation
	\begin{equation}\label{23}\begin{aligned}
		(M_n)_{ij}(x,\lambda)=\delta_{ij}+\int_{\gamma_{ij}^n} \left({\rm{e}}^{(x-y)\widehat{\mathfrak{L}(\lambda)}}(L_1M_n)(y,\lambda)\right)_{ij}\rd y,
	\end{aligned}\end{equation}
	where the contours $\gamma_{ij}^n\ (n=1,\dots,6,\ i,\ j =1,\ 2,\ 3)$ are defined by
	\begin{equation*}
		\gamma_{ij}^n=\left\{
		\begin{aligned}
			(-\infty,x),\quad {\rm{Re}}\,l_i(\lambda)<{\rm{Re}}\,l_j(\lambda),\\
			(+\infty,x),\quad {\rm{Re}}\,l_i(\lambda)\ge{\rm{Re}}\,l_j(\lambda),
		\end{aligned}
		\quad {\rm{for}} \quad \lambda\in D_n.
		\right.
	\end{equation*}
   Define the set of zeros for the Fredholm determinants associated with Fredholm integral equations as $\mathcal{N}$,  and extend $\tilde{\mathcal{N}}:=\mathcal{N}\cup\{0\}$. The following proposition gives some properties of the matrix-valued functions defined by equation \eqref{23}.
		\begin{prop}
		$($Basic properties of $M_n$$)$. Suppose $u_0(x), v_0(x)\in \mathcal{S}(\mathbb{R})$, then the Fredholm integral equations referenced by \eqref{23} uniquely determine $M_n$ for $n=1, 2,\cdots,6$. These functions exhibit the following properties:
		\begin{enumerate}
			\item The function $M_n(x,\lambda)$ is defined for $x\in \mathbb{R}$ and $\lambda\in \bar{D}_n\setminus\tilde{\mathcal{N}}$, and is smooth and satisfies \eqref{14} for each $\lambda\in \bar{D}_n\setminus\tilde{\mathcal{N}}$.
			
			\item The function $M_n(x,\lambda)$ is continuous for $\lambda\in \bar{D}_n\setminus\tilde{\mathcal{N}}$ and analytic for $\lambda\in D_n\setminus\tilde{\mathcal{N}}$.
			
			\item For each $\epsilon>0$, the function $M_n(x,\lambda)$ satisfies the property of boundedness:
			\begin{equation*}
				\left|M_n(x,\lambda)\right|\le C(\epsilon), \quad x\in \mathbb{R},\lambda \in \bar{D}_n, {\rm{dist}}(\lambda,\tilde{\mathcal{N}})\ge\epsilon.
			\end{equation*}
			\item The partial derivative
			 $\frac{\partial^j M_n(x,\cdot)}{\partial \lambda^j}$ has the continuous extension to $\bar{D}_n\setminus\tilde{\mathcal{N}}$ for $j=1,2,\cdots$.			
			\item  $\det M_n(x,\lambda)=1$ for $x\in\mathbb{R}$ and $\lambda \in \bar{D}_n\setminus\tilde{\mathcal{N}}$.
			
			\item The functions $M(x,\lambda):=M_n(x,\lambda)$ for $\lambda\in D_n$ satisfy the symmetries
			\begin{equation*}				M(x,\lambda)=\tilde{\sigma}_1^{-1}M(x,\omega\lambda)\tilde{\sigma}_1=\tilde{\sigma}_2\overline{M(x,\overline{\lambda})}\tilde{\sigma}_2^{-1},\quad \lambda\in\mathbb{C}\setminus\tilde{\mathcal{N}}.	\end{equation*}
			\end{enumerate}
	\end{prop}
	
	\begin{proof}
		Without loss of generality, we will only prove for the first column of matrix $M_1(x,\lambda)$, and the proofs for the other cases are  analogous. Let $m_i(x,\lambda)=(M_1)_{i1}(x,\lambda)$, and write
		\begin{equation*}
			m_i(x,\lambda)=\delta_{i1}+\int_{\mathbb{R}}\sum_{j=1}^{3}K_{ij}(x,y,\lambda)m_j(y,\lambda)\rd y,\quad x\in\mathbb{R},\lambda\in\bar{D}_1\setminus\{0\},i=1,2,3,
		\end{equation*}
		where the kernel $K$ is defined by
		\begin{equation*}
			K_{ij}(x,y,\lambda)=\left\{
				\begin{aligned}
					&H(x-y){\rm{e}}^{(x-y)(l_i-l_1)}(L_1(y,\lambda))_{ij},\quad &&{\rm{if}}\,\gamma_{i1}^1=(-\infty,x),\\
					&-H(y-x){\rm{e}}^{(x-y)(l_i-l_1)}(L_1(y,\lambda))_{ij},\quad &&{\rm{if}}\,\gamma_{i1}^1=(\infty,x),
				\end{aligned}
				\right.
		\end{equation*}
		and $H(x)=1$ for $x>0$, $H(x)=0$ for $x\le0$. Fix $\epsilon>0$ small, and let $\bar{D}_1^\epsilon=\bar{D}_1\setminus\{\left|\lambda\right|<\epsilon\}$. So there is a Schwartz function $b(x)$ that makes
		\begin{equation*}
			\left|L_1(x,\lambda)\right|\le b(x),\quad x\in \mathbb{R},\lambda \in \bar{D}_1^\epsilon.
		\end{equation*}
		
		 Additionally, the Fredholm determinant corresponding to the first column is
		\begin{equation*}
			f(\lambda)=\sum_{m=0}^{\infty}\frac{(-1)^m}{m!}\sum_{i_1,i_2,\cdots,i_m=1}^{3}\int_{\mathbb{R}^m}K^{(m)}\bigg(
			\begin{aligned}
				x_1,i_1,x_2,i_2,\cdots,x_m,i_m\\
				x_1,i_1,x_2,i_2,\cdots,x_m,i_m
			\end{aligned};
			\lambda
			\bigg)\rd x_1 \rd x_2 \cdots \rd x_m,
		\end{equation*}
		where
		\begin{equation*}
			K^{(m)}\bigg(
			\begin{aligned}
				x_1,i_1,x_2,i_2,\cdots,x_m,i_m\\
				y_1,i_1^{\prime},y_2,i_2^{\prime},\cdots,y_m,i_m^{\prime}
			\end{aligned};
			\lambda
			\bigg)=\det
			\begin{pmatrix}
				K(x_1,y_1,\lambda)_{i_1i_1^{\prime}} & \cdots & K(x_1,y_m,\lambda)_{i_1i_m^{\prime}}\\
				\vdots &  &\vdots\\
				K(x_m,y_1,\lambda)_{i_mi_1^{\prime}} & \cdots &K(x_m,y_m,\lambda)_{i_mi_m^{\prime}}
			\end{pmatrix}.
		\end{equation*}
		
		Using the Hadamard's inequality,  we immediately get
		\begin{equation*}
			\left|K^{(m)}\right|\le m^{\frac{m}{2}}\prod_{j=1}^{m}b(y_j).
		\end{equation*}
		Hence, the Fredholm determinant $f(\lambda)$, which is related to the first column of \eqref{23}, is a  holomorphic function for $\lambda \in D_1\setminus\tilde{\mathcal{N}}$ and it remains continuous up to the boundary of $D_1\setminus\tilde{\mathcal{N}}$. Moreover, given that the potential matrix $L_1\rightarrow\mathcal{O}(1)$ nonvanishing as $\lambda\rightarrow\infty$, the determinant $f$ also remains bounded. This implies that the number of zeros for $f(\lambda)$ within the region $D_1$ is finite.
	\end{proof}
	
	\begin{lem}
		Suppose $u_0(x), v_0(x)\in \mathcal{S}(\mathbb{R})$ with compact support. Then
		\begin{equation}\label{24}
			\begin{aligned}
				M_n(x,\lambda)&=X_-(x,\lambda){\rm{e}}^{x\widehat{\mathfrak{L}(\lambda)}}S_n(\lambda)\\
				&=X_+(x,\lambda){\rm{e}}^{x\widehat{\mathfrak{L}(\lambda)}}T_n(\lambda), \quad n=1,2,\cdots,6,
			\end{aligned}
		\end{equation}
		where $S_n(\lambda)$ and $T_n(\lambda)$ are given in terms of the entries and the $(ij)$-th minor $m_{ij}(s)$ of the matrix $s(\lambda)$, which are expressed by
		\begin{equation*}\begin{aligned}
			S_1(\lambda) &=\begin{pmatrix}
				s_{11}	&	0	&0\\
				s_{21}	&\frac{m_{33}(s)}{s_{11}}	&0\\
				s_{31}	&\frac{m_{23}(s)}{s_{11}}&\frac{1}{m_{33}(s)}
			\end{pmatrix},&&
			S_2(\lambda)=\begin{pmatrix}
				s_{11}	&0	&0\\
				s_{21}	&\frac{1}{m_{22}(s)}	&\frac{m_{32}(s)}{s_{11}}\\
				s_{31}	&0	&\frac{m_{22}(s)}{s_{11}}
			\end{pmatrix},\\
			S_3(\lambda) &=\begin{pmatrix}
				\frac{m_{22}(s)}{s_{33}}	&	0	&s_{13}\\
				\frac{m_{12}(s)}{s_{33}}	&\frac{1}{m_{22}(s)}	&s_{23}\\
				0	&0&s_{33}
			\end{pmatrix},&&
			S_4(\lambda) =\begin{pmatrix}
				\frac{1}{m_{11}(s)}	&\frac{m_{21}(s)}{s_{33}}		&s_{13}\\
				0		&\frac{m_{11}(s)}{s_{33}}	&s_{23}\\
				0	&0&s_{33}
			\end{pmatrix},\\
			S_5(\lambda) &=\begin{pmatrix}
				\frac{1}{m_{11}(s)}	&	s_{12}	&-\frac{m_{31}(s)}{s_{22}}\\
				0	&s_{22}	&0\\
				0	&s_{32}&\frac{m_{11}(s)}{s_{22}}
			\end{pmatrix},&&
			S_6(\lambda) =\begin{pmatrix}
				\frac{m_{33}(s)}{s_{22}}	&s_{12}	&0\\
				0	&s_{22}	&0\\
				-\frac{m_{13}(s)}{s_{22}}	&s_{32}	&\frac{1}{m_{33}(s)}
			\end{pmatrix},
		\end{aligned}\end{equation*}
		and
			\begin{equation*}\begin{aligned}
			T_1(\lambda) &=\begin{pmatrix}
				1	&	-\frac{s_{12}}{s_{11}}	&\frac{m_{31}(s)}{m_{33}(s)}\\
			0	&1	&-\frac{m_{32}(s)}{m_{33}(s)}\\
			0	&0&1
			\end{pmatrix},&&
		T_2(\lambda)=\begin{pmatrix}
				1	&-\frac{m_{21}(s)}{m_{22}(s)}	&-\frac{s_{13}}{s_{11}}\\
				0	&1	&0\\
				0	&-\frac{m_{23}(s)}{m_{22}(s)}	&1
			\end{pmatrix},\\
			T_3(\lambda) &=\begin{pmatrix}
				1	&	-\frac{m_{21}(s)}{m_{22}(s)}	&0\\
				0	&1	&0\\
				-\frac{s_{31}}{s_{33}}	&-\frac{m_{23}(s)}{m_{22}(s)}  &1
			\end{pmatrix},&&
			T_4(\lambda) =\begin{pmatrix}
				1	&0		&0\\
				-\frac{m_{12}(s)}{m_{11}(s)}		&1	&0\\
				\frac{m_{13}(s)}{m_{11}(s)}	&-\frac{s_{32}}{s_{33}}&1
			\end{pmatrix},\\
			T_5(\lambda) &=\begin{pmatrix}
				1	&0	&0\\
				-\frac{m_{12}(s)}{m_{11}(s)}	&1	&-\frac{s_{23}}{s_{22}}\\
				\frac{m_{13}(s)}{m_{11}(s)}	&0&1
			\end{pmatrix},&&
			T_6(\lambda) =\begin{pmatrix}
				1	&0	&\frac{m_{31}(s)}{m_{33}(s)}\\
				-\frac{s_{21}}{s_{22}}	&1	&-\frac{m_{32}(s)}{m_{33}(s)}\\
				0	&0	&1
			\end{pmatrix}.
		\end{aligned}\end{equation*}
	\end{lem}
	\begin{proof}
		Derived from the equation \eqref{24}, we obtain
		\begin{equation*}
			\left\{
				\begin{aligned}
					S_n(\lambda)=& \lim_{x\rightarrow-\infty}{\rm{e}}^{-x\widehat{\mathfrak{L}(\lambda)}}M_n(x,\lambda),\\
					T_n(\lambda)=& \lim_{x\rightarrow\infty}{\rm{e}}^{-x\widehat{\mathfrak{L}(\lambda)}}M_n(x,\lambda),
				\end{aligned}\right.
				\quad \lambda\in \bar{D}_n\setminus\tilde{\mathcal{N}},
		\end{equation*}
		and recall that $s(k)$ is defined by equation \eqref{21} for all $\lambda\in \mathbb{C}\setminus\{0\}$ for compactly supported data, so we have
		\begin{equation*}
			s(\lambda)T_n(\lambda)=S_n(\lambda).
		\end{equation*}
	Given $s(\lambda)$, the above equation constitutes a matrix factorization problem which can be uniquely solved for $S_n(\lambda)$ and $T_n(\lambda)$.
	\end{proof}
	
	\begin{lem}
		Assume $u_0(x)$, $v_0(x)\in \mathcal{S}(\mathbb{R})$, and let $\{s(\lambda),M_n(x,\lambda)\}$ be the spectral functions and eigenfunctions associated with $(u_0(x), v_0(x))$. More importantly, $u_0^{(i)}(x)$, $v_0^{(i)}(x)\in C^{\infty}_0(\mathbb{R})$, which uniformly converge to $u_0(x)$, $v_0(x)$, respectively. Let $\{s^{(i)}(\lambda),M^{(i)}_n(x,\lambda)\}$ be the spectral functions and eigenfunctions associated with $(u_0^{(i)}(x), v_0^{(i)}(x))$. Then it follows
		\begin{equation*}
			\begin{aligned}
				&\lim_{i\rightarrow\infty}s^{(i)}(\lambda)=s(\lambda),\quad 	\lambda\in	\left(\begin{array}{ccc}
					\omega^2\bar{S} & \mathbb{R}_+ & \omega\mathbb{R}_+
					\\
					\mathbb{R}_+ & \omega\bar{S} & 	\omega^2 \mathbb{R}_+
					\\
					\omega \mathbb{R}_+ & \omega^2 \mathbb{R}_+ &\bar{S}
				\end{array}\right)
				\backslash\{0\},\\
				&\lim_{i\rightarrow\infty}(s^A)^{(i)}(\lambda)=s^A(\lambda),\quad \lambda \in\left(\begin{array}{ccc}
					-\omega^2\bar{S} & \mathbb{R}_- & \omega\mathbb{R}_-
					\\
					\mathbb{R}_- & -\omega\bar{S} & 	\omega^2 \mathbb{R}_-
					\\
					\omega \mathbb{R}_- & \omega^2 \mathbb{R}_- &-\bar{S}
				\end{array}\right)
				\backslash\{0\},\\
				&\lim_{i\rightarrow\infty}X_+^{(i)}(x,\lambda)=X_+(x,\lambda),\quad x\in \mathbb{R}, \lambda\in (\omega^2 \bar{S},\omega \bar{S},\bar{S})\setminus\{0\},\\
				&\lim_{i\rightarrow\infty}X_-^{(i)}(x,\lambda)=X_-(x,\lambda),\quad x\in \mathbb{R}, \lambda\in (-\omega^2 \bar{S},-\omega \bar{S},-\bar{S})\setminus\{0\},\\
				&\lim_{i\rightarrow\infty}M_n^{(i)}(x,\lambda)=M_n(x,\lambda),\quad x\in \mathbb{R}, \lambda\in \bar{D}_n\setminus\tilde{\mathcal{N}},n=1,2,\cdots,6.
				\end{aligned}
		\end{equation*}
		
	\end{lem}

	\begin{lem}
		$($Jump condition for $M(x,\lambda)$$)$ Assume $u_0(x)$, $v_0(x)\in \mathcal{S}(\mathbb{R})$, $M(x,\lambda)$ satisfies the condition
		\begin{equation*}
			M_+(x,\lambda)=M_-(x,\lambda)v(x,0,\lambda),\quad \lambda \in \Sigma\setminus\tilde{\mathcal{N}},
		\end{equation*}
		where $M_+(x,\lambda)$ and $M_-(x,\lambda)$ indicate the limit that $\lambda$ approaches $\Sigma$ from the left and right. If $\lambda\in {\rm{e}}^{(j-1)\pi  i/3}\mathbb{R}_+$ for $j=1,2,\cdots,6$, then $v(x,0,\lambda)=v_j(x,0,\lambda)$, and $v_j(x,0,\lambda)={\rm{e}}^{x\widehat{\mathfrak{L}(\lambda)}}J_j(\lambda)$, where $J_j=J_j(\lambda)~(j=1,2,\cdots,6)$ satisfy
		\begin{equation*}\begin{aligned}
		J_1 &=\begin{pmatrix}
				1	&	-r_1(\lambda)	&0\\
				r_1^*(\lambda)	&1-r_1(\lambda)r_1^*(\lambda)	&0\\
				0	&0&1
			\end{pmatrix},&&
			J_2=\begin{pmatrix}
				1	&0	&0\\
				0	&1-r_2(\omega\lambda)r_2^*(\omega\lambda)	&-r_2^*(\omega\lambda)\\
				0	&r_2(\omega\lambda)	&1
			\end{pmatrix},\\
			J_3 &=\begin{pmatrix}
				1-r_1(\omega^2\lambda)r_1^*(\omega^2\lambda)	&	0	&r_1^*(\omega^2\lambda)\\
				0 & 1& 0\\
				-r_1(\omega^2\lambda)	&0	&1\\
			\end{pmatrix},&&
			J_4 =\begin{pmatrix}
				1--r_2(\lambda)r_2^*(\lambda)	&-r_2^*(\lambda)		&0\\
			r_2(\lambda)	&1	&0\\
				0	&0&1
			\end{pmatrix},\\
			J_5 &=\begin{pmatrix}
				1	&0	&0\\
				0	&1	&-r_1(\omega\lambda)\\
				0	&r_1^*(\omega\lambda)&1-r_1(\omega\lambda)r_1^*(\omega\lambda)
			\end{pmatrix},&&
			J_6 =\begin{pmatrix}
				1	&0	&r_2(\omega^2\lambda)\\
				0	&1	&0\\
				-r^*_2(\omega^2\lambda)	&0	&1-r_2(\omega^2\lambda)r^*_2(\omega^2\lambda)
			\end{pmatrix},
		\end{aligned}\end{equation*}
		where $r_1(\lambda):=\frac{s_{12}(\lambda)}{s_{11}(\lambda)}$ and $r_2(\lambda):=\frac{s_{12}^A(\lambda)}{s_{11}^A(\lambda)}$.
		
	\end{lem}
	
	\begin{lem}
		Assume $u_0(x)$, $v_0(x)\in \mathcal{S}(\mathbb{R})$, then the functions $M_n$ and $M_n^A=(M_n^{-1})^T$ for $n=1,4$ can be given by
		
		\begin{equation*}
			\begin{aligned}
				&M_1=\begin{pmatrix}
					(X_+)_{11} & \frac{(X_-^A)_{31}(X_+^A)_{23}-(X_-^A)_{21}(X_+^A)_{33}}{s_{11}} & \frac{(X_-)_{13}}{s^A_{33}} \\
					(X_+)_{21} & \frac{(X_-^A)_{11}(X_+^A)_{33}-(X_-^A)_{31}(X_+^A)_{13}}{s_{11}} & \frac{(X_-)_{23}}{s^A_{33}} \\
					(X_+)_{31} & \frac{(X_-^A)_{21}(X_+^A)_{13}-(X_-^A)_{11}(X_+^A)_{23}}{s_{11}} & \frac{(X_-)_{33}}{s^A_{33}} \\
				\end{pmatrix},\\
					\end{aligned}
			\end{equation*}
			\begin{equation*}
				\begin{aligned}
				&M_1^A=\begin{pmatrix}
					\frac{(X_-^A)_{11}}{s_{11}} & \frac{(X_+)_{31}(X_-)_{23}-(X_+)_{21}(X_-)_{33}}{s^A_{33}} & (X_+^A)_{13} \\
					\frac{(X_-^A)_{21}}{s_{11}} & \frac{(X_+)_{11}(X_-)_{33}-(X_+)_{31}(X_-)_{13}}{s^A_{33}} & (X_+^A)_{23} \\
					\frac{(X_-^A)_{31}}{s_{11}} & \frac{(X_+)_{21}(X_-)_{13}-(X_+)_{11}(X_-)_{23}}{s^A_{33}} & (X_+^A)_{33} \\
				\end{pmatrix},\\
					\end{aligned}
			\end{equation*}
		\begin{equation*}
		\begin{aligned}
				&M_4=\begin{pmatrix}
					\frac{(X_-)_{11}}{s^A_{11}} & \frac{(X_-^A)_{23}(X^A_+)_{31}-(X^A_+)_{21}(X^A_-)_{33}}{s_{33}} & (X_+)_{13} \\
						\frac{(X_-)_{21}}{s^A_{11}} & \frac{(X_+^A)_{11}(X^A_-)_{33}-(X^A_-)_{13}(X^A_+)_{31}}{s_{33}} & (X_+)_{23} \\
					\frac{(X_-)_{31}}{s^A_{11}} & \frac{(X_+^A)_{21}(X^A_-)_{13}-(X^A_+)_{11}(X^A_-)_{23}}{s_{33}} & (X_+)_{33} \\
					\end{pmatrix},\\
						\end{aligned}
				\end{equation*}
			\begin{equation*}
			\begin{aligned}
				&M_4^A=\begin{pmatrix}
				(X_+^A)_{11} & \frac{(X_+)_{23}(X_-)_{31}-(X_-)_{21}(X_+)_{33}}{s_{11}^A} & \frac{(X_-^A)_{13}}{s_{33}} \\
				(X_+^A)_{21} & \frac{(X_-)_{11}(X_+)_{33}-(X_+)_{13}(X_-)_{31}}{s_{11}^A} & \frac{(X_-^A)_{23}}{s_{33}} \\
				(X_+^A)_{31} & \frac{(X_-)_{21}(X_+)_{13}-(X_-)_{11}(X_+)_{23}}{s_{11}^A} &\frac{ (X_-^A)_{33}}{s_{33}} \\
				\end{pmatrix}.\\
			\end{aligned}
		\end{equation*}
	\end{lem}
	
	\begin{remark}
		Under the Assumption \ref{assu1}, the function $M(x,t,\lambda)$ can be analytically continued to the zeros of the Fredholm determinant.
		
	\end{remark}
	
	\begin{lem}
		Assume $u_0(x)$, $v_0(x)\in \mathcal{S}(\mathbb{R})$, the function $M(x, \lambda)$ has the same expansion with the function $X_+(x, \lambda)$ as $\lambda\rightarrow0$. More importantly, the leading term in the expansion is $G(x)$.
	\end{lem}
	
	\begin{proof}
		When dealing with the problem that $\lambda=0$ is the singularity of the kernel of the Volterra integral equation \eqref{15}, we introduce the transformation $X=GY$, and expand the eigenfunctions $X_{\pm}(x,\lambda)$ to get the expansion \eqref{19} as $\lambda\rightarrow0$. Denote
		\begin{equation*}
			\mathcal{M}_p(x,\lambda):=G(x)+G(x)Y_+^{(1)}\lambda+\cdots+G(x)Y_+^{(p)}\lambda^p,
		\end{equation*}
		and it is evident that the following can be obtained
		\begin{equation*}
			\left|M(x,\lambda)-\mathcal{M}_p(x,\lambda)\right|\le C\left|\lambda\right|^p,\quad x\in \mathbb{R},\,\lambda\in \mathbb{C}\setminus\Sigma,\,\left|\lambda\right|<\epsilon,
		\end{equation*}
		with $\epsilon<1$ small enough.
		
		Let $D_n^\epsilon:=D_n\cup\{\lambda<\epsilon\}$. The upper bound of the maximum norm for $\mathcal{M}_p$ is less than 1, given that $G$ is a matrix-valued function residing within the Schwartz space. Thus, $\mathcal{M}_p^{-1}(x,\lambda)$ exists and is uniformly bounded for $\lambda\in D_n^\epsilon$.		
		
		Next suppose that $L_{(p)}:=(\partial_x\mathcal{M}_p+\mathcal{M}_p\mathfrak{L})\mathcal{M}_p^{-1}$ and set $\Delta=L-L_{(p)}$. By calculating, we get the differential equation that $\mathcal{M}_p^{-1}M$ satisfies
		\begin{equation*}
			(\mathcal{M}_p^{-1}M)_x=\mathcal{M}_p^{-1}\Delta M+[\mathfrak{L},\mathcal{M}^{-1}_pM],
		\end{equation*}
		that is
		\begin{equation*}
			({\rm{e}}^{-x\hat{\mathfrak{L}}}(\mathcal{M}_p^{-1}M))_x={\rm{e}}^{-x\hat{\mathfrak{L}}}(\mathcal{M}_p^{-1}\Delta M).
		\end{equation*}
		Each entry of $\mathcal{M}_p^{-1}M_n$ satisfies the integral equation
		\begin{equation*}
			(\mathcal{M}_p^{-1}M_n)_{ij}(x,\lambda)=\delta_{ij}+\int_{\gamma^n_{ij}}({\rm{e}}^{(x-y)\hat{\mathfrak{L}}}(\mathcal{M}_p^{-1}\Delta M_n)(y,\lambda))_{ij}\rd y.
		\end{equation*}
		Without loss of generality, consider the first column of $M_1$ and assume $m_i(x,\lambda)=(M_1)_{i1}(x,\lambda)$, and then rewrite the Fredholm integral equation as follows
		\begin{equation*}
		  m_i(x,\lambda)=(\mathcal{M}_p)_{i1}+\int_{\mathbb{R}}\sum_{j=1}^{3}K_{ij}(x,y,\lambda)m_j(y,\lambda)\rd y,
		\end{equation*}
		where
		\begin{equation*}
			K_{ij}(x,y,\lambda)=\sum_{s=1}^{3}(\mathcal{M}_p(x,\lambda))_{is}H_s(x,y){\rm{e}}^{(l_s-l_1)(x-y)}(\mathcal{M}_p^{-1}\Delta)_{sj}(y,\lambda),
		\end{equation*}
		and
		\begin{equation*}
			H_s(x,y)=\left\{
			\begin{aligned}
				&H(x-y), \quad &&{\rm{if}} \quad {\rm{Re}}\,l_i(\lambda)<{\rm{Re}}\,l_1(\lambda),\\
				&-H(y-x), \quad &&{\rm{if}}\quad {\rm{Re}}\,l_i(\lambda)\ge{\rm{Re}}\,l_1(\lambda).
				\end{aligned}
			\right.
		\end{equation*}
		Let us now rewrite $\Delta$ as follows
		\begin{equation*}
			\Delta=(L_1\mathcal{M}_p+[\mathfrak{L},\mathcal{M}_p]-\partial_x\mathcal{M}_p)\mathcal{M}_p^{-1}.
		\end{equation*}
		Since $L_1$ and $G$ are Schwartz class matrix-valued functions and $\mathcal{M}_p^{-1}$ is bounded for $\lambda\in D_n^\epsilon$, there exists a function $\alpha(x)\in \mathcal{S}(\mathbb{R})$ such that
		\begin{equation*}
			\left|\Delta(x,\lambda)\right|\le C\alpha(x)\left|\lambda\right|^{p+1}, \quad x\in \mathbb{R}, \lambda\in D_n^\epsilon,
		\end{equation*}
		and
		\begin{equation*}
			\left|K(x,y,\lambda)\right|\le \alpha(y)\left|\lambda\right|^{P+1}, \quad x,y\in \mathbb{R}, \lambda\in D_1^\epsilon.
		\end{equation*}
		In this case, set $K^{(0)}=1$ and define
		\begin{equation}\label{25}
			K^{(m)}\bigg(
			\begin{aligned}
				x_1,i_1,x_2,i_2,\cdots,x_m,i_m\\
				y_1,i_1^{\prime},y_2,i_2^{\prime},\cdots,y_m,i_m^{\prime}
			\end{aligned};
			\lambda
			\bigg)=\det
			\begin{pmatrix}
				K(x_1,y_1,\lambda)_{i_1i_1^{\prime}} & \cdots & K(x_1,y_m,\lambda)_{i_1i_m^{\prime}}\\
				\vdots &  &\vdots\\
				K(x_m,y_1,\lambda)_{i_mi_1^{\prime}} & \cdots &K(x_m,y_m,\lambda)_{i_mi_m^{\prime}}
			\end{pmatrix}.
		\end{equation}
		
		Similar to the discussion of the basic properties of $M_n$, obtain the following estimator
		\begin{equation*}
			\left|K^{(m)}\right|\le m^{{\frac{m}{2}}}\prod_{j=1}^{m}\alpha(y_j)\left|\lambda\right|^{p+1}.
		\end{equation*}
		Define the Fredholm determinant $f(\lambda)$ and Fredholm minor $F$ with $K(x,y,\lambda)$ given by \eqref{25}
		\begin{equation*}
			f(\lambda)=\sum_{m=0}^{\infty}f^{(m)}(\lambda), \quad \lambda \in D_1^\epsilon,
		\end{equation*}
		\begin{equation*}
			F_{ii^\prime}(x,y,\lambda)=\sum_{m=0}^{\infty}F^{(m)}_{ii^\prime}(x,y,\lambda),\quad x,y\in \mathbb{R}, \lambda\in D_1^\epsilon,
		\end{equation*}
		where
		\begin{equation*}
			f^{(m)}(\lambda)=\frac{(-1)^m}{m!}\sum_{i_1,i_2,\cdots,i_m=1}^{3}\int_{\mathbb{R}^m}K^{(m)}\bigg(
			\begin{aligned}
				x_1,i_1,x_2,i_2,\cdots,x_m,i_m\\
				x_1,i_1,x_2,i_2,\cdots,x_m,i_m
			\end{aligned};
			\lambda
			\bigg)\rd x_1 \rd x_2 \cdots \rd x_m,
		\end{equation*}
		\begin{equation*}
		F^{(m)}_{ii^\prime}(x,y,\lambda)=\frac{(-1)^m}{m!}\sum_{i_1,i_2,\cdots,i_m=1}^{3}\int_{\mathbb{R}^m}K^{(m+1)}\bigg(
			\begin{aligned}
				x,i,x_1,i_1,x_2,i_2,\cdots,x_m,i_m\\
				y,i^\prime,y_1,i_1,y_2,i_2,\cdots,y_m,i_m
			\end{aligned};
			\lambda
			\bigg)\rd x_1 \rd x_2 \cdots \rd x_m.
		\end{equation*}
		Therefore, one has
		\begin{equation*}
			\left|f^{(m)}(\lambda)\right|\le \frac{3^mm^{\frac{m}{2}}\Vert\alpha(y)\Vert^m_{L^1(\mathbb{R})}\left|\lambda\right|^{(p+1)m}}{m!},\quad \lambda \in D_1^\epsilon,m\ge 0,
		\end{equation*}
		\begin{equation*}
			\left|F^{(m)}_{ii^\prime}(x,y,\lambda)\right|\le \frac{3^m(m+1)^{\frac{m+1}{2}}\Vert\alpha(y)\Vert^m_{L^1(\mathbb{R})}\alpha(y)\left|\lambda\right|^{(p+1)(m+1)}}{m!},\quad x,y\in \mathbb{R}, \lambda \in D_1^\epsilon,m\ge 0.
		\end{equation*}
		This promptly suggests that
		\begin{equation*}
			\vert f(\lambda)-1\vert \le C\vert\lambda\vert^{p+1},\quad \lambda \in D_1^\epsilon,
		\end{equation*}
		\begin{equation*}
			\vert F(x,y,\lambda)\vert \le C\alpha(y)\vert\lambda\vert^{p+1},\quad \lambda\in D_1^\epsilon.
		\end{equation*}
		
		Observe that there are no zeros of the Fredholm integral equation for $\lambda\in D_1^\epsilon$. Consequently, based on the theoretical examination of Fredholm integral equations, the following conclusion can be obtained
		\begin{equation*}
			m_i(x,\lambda)=(\mathcal{M}_p)_{i1}(x,\lambda)+\frac{1}{f(\lambda)}\int_{\mathbb{R}}\sum_{j=1}^{3}F_{ij}(x,y,\lambda)(\mathcal{M}_p)_{j1}(y,\lambda) \rd y,\quad x\in \mathbb{R}, \lambda \in D_1^\epsilon,
		\end{equation*}
		and
		\begin{equation*}
			\left|  m_i(x,\lambda)-(\mathcal{M}_p)_{i1}(x,\lambda)\right|  \le C \vert \lambda\vert^{p+1}.
		\end{equation*}
		
	\end{proof}
	
	\subsection{The reconstruction formula }
	\quad\quad
	
	Define the time-dependent eigenfunctions $\{M_n(x,t,\lambda)\}_{n=1}^6$ by substituting $L_1(x,\lambda)$ into the Fredholm integral equations for $L_1(x,t,\lambda)$. Let the piecewise holomorphic function $M(x,t,\lambda)$ be represented as $M(x,t,\lambda)=M_n(x,t,\lambda)$ when $\lambda\in D_n$. Moreover, the function $M(x,t,\lambda)$ defined by the RH problem \ref{rhp} shares an equivalent expansion with $X_+(x,t,\lambda)$, indicating that
	\begin{equation*}\left\{
		\begin{aligned}
			&u(x,t)=\frac{1}{2}\lim_{\lambda\rightarrow0}\ln[(\omega,\omega^2,1)M(x,t,\lambda)]_{13},\\
			&v(x,t)=\frac{1}{6}\lim_{\lambda\rightarrow0}\ln[(\omega,\omega^2,1)M(x,t,\lambda)]_{13}-\frac{1}{3}\lim_{\lambda\rightarrow0}\ln[(\omega^2,\omega,1)M(x,t,\lambda)]_{13}.
		\end{aligned}\right.
	\end{equation*}
	\begin{lem}
		For $\lambda\in D_n$, the function $M_n(x,t,\lambda)$ is smooth for $(x,t)\in \mathbb{R}\times \left[0,T\right)$, and satisfies the Lax pair \eqref{14}.
	\end{lem}
	\begin{lem}
		For $(x,t)\in \mathbb{R}\times \left[0,T\right)$, the time-dependent function $M(x,t,\lambda)$ is a sectionally holomorphic function for $\lambda\in \mathbb{C}\setminus\Sigma$, and satisfies the same jump conditions with function $M(x,\lambda)$.
	\end{lem}
	Combining the two Lemmas mentioned above, if $M(x,t,\lambda)$ satisfies the Lax pair \eqref{2} and the jump conditions, then the $M(x,t,\lambda)$ we constructed is the solution to the RH problem \ref{rhp}. In other words, the reconstruction formula \eqref{7} provides the solution to the initial value problem of the two-component nonlinear KG equation \eqref{3}.

	\section{Higher-order asymptotics of the problem \eqref{3}}\label{sec4}
	\ \ \ \
	In the previous section, we have constructed the required RH problem \ref{rhp} and obtained the jump relation \ref{5} satisfied by the eigenfunctions.  To prepare for the analysis of the long-time asymptotic behaviors of the solution to the initial value problem of the two-component KG equation \eqref{3}, firstly, consider the RH problem \ref{rhp} and note the dispersion relation $\theta_{21}(\lambda)=\theta_{21}(x,t,\lambda)$ that
	\begin{equation}\label{26}	\theta_{21}(x,t,\lambda)=\frac{(\omega^2-\omega)t}{2}[(\lambda-\lambda^{-1})\xi+(\lambda+\lambda^{-1})],
	\end{equation}
	where $\xi=\frac{x}{t}$. After some calculations, the critical points of phase function $\theta_{21}(\lambda)$ are gotten
	\begin{equation}\label{lambda0}
		\pm \lambda_0=\pm \sqrt{\frac{\vert x-t\vert}{\vert x+t\vert}},
	\end{equation}
	and according to the relationship between the three phase functions $\theta_{21}(\lambda)$, $\theta_{31}(\lambda)$ and $\theta_{32}(\lambda)$, the critical points of $\theta_{31}(\lambda)$ are $\pm \omega \lambda_0$, while the critical points of $\theta_{32}(\lambda)$ are $\pm \omega^2 \lambda_0$.
	
	 In the following analysis, we need to define notation $B_{\lambda_0}(0):=\{\lambda\in\mathbb{C}\vert\vert\lambda\vert<\lambda_0\}$,  to represent an open disk in the complex plane with $0$ as the center and $\lambda_0$ as the radius, and $\partial B_{\lambda_0}(0):=\{\lambda\in\mathbb{C}\vert \vert \lambda\vert=\lambda_0\}$ to represent the boundary of $B_{\lambda_0}(0)$. The signatures of the real part of phase function $\theta_{21}(\lambda)$ in different cases are obtained:
	 \begin{itemize}
	 	\item For $\xi \geq1$, ${\rm{Re}}\,\theta_{21}>0$ for ${\rm{Im}}\,\lambda>0$ and ${\rm{Re}}\,\theta_{21}< 0$ for ${\rm{Im}}\,\lambda<0$.
	 	\item For $\xi \leq-1$, ${\rm{Re}}\,\theta_{21}>0$ for ${\rm{Im}}\,\lambda<0$ and ${\rm{Re}}\,\theta_{21}< 0$ for ${\rm{Im}}\,\lambda>0$.
	 	\item For $\vert\xi\vert<1$, ${\rm{Re}}\,\theta_{21}>0$ for $\lambda\in\{B_{\lambda_0}(0)\cap {\rm{Im}}\,\lambda<0\} \cup\{B_{\lambda_0}^c(0)\cap{\rm{Im}}\,\lambda>0\}$ and
	    ${\rm{Re}}\,\theta_{21}<0$ for $\lambda\in\{B_{\lambda_0}(0)\cap {\rm{Im}}\,\lambda>0\} \cup\{B_{\lambda_0}^c(0)\cap{\rm{Im}}\,\lambda<0\}$.
	 \end{itemize}
	 To more intuitively understand the case of $\vert\xi\vert<1$, we present the signature of the real part of $\theta_{21}$ in equation \eqref{26} as shown in Fig. \ref{figsignature}.
	  \begin{figure}[htbp]
	 	\centering
	 	\begin{tikzpicture}[scale=1.2]
	 		
	 		\fill[black!20!blue!20] (0,0) -- (-1.5,0) arc (180:360:1.5cm) -- cycle;
	 	     \definecolor{shadecolor}{RGB}{200,200,200}
	 	
	 	    \fill[black!20!blue!20] (-3,0) -- (3,0) -- (3,1.732*1.5) -- (-3,1.732*1.5) -- cycle;
	 	    \fill[white] (0,0) -- (1.5,0) arc (0:180:1.5cm) -- cycle;
	 	
	 		\draw [very thick,black!20!blue](0,0) circle (1.5cm);
	 		
	 		\draw [very thick,black!20!blue](0,0) -- (3,0);
	 		\draw [very thick,black!20!blue](-3,0) -- (0,0);
	 		\draw [very thick,black!20!blue](0,0) -- (1.5,1.5*1.732);
	 		\draw [very thick,black!20!blue](-1.5,-1.5*1.732) -- (0,0);
	 		\draw [very thick,black!20!blue](0,0) -- (-1.5,1.732*1.5);
	 		\draw [very thick,black!20!blue](0,0) -- (1.5,-1.5*1.732);
	 		
	 		\draw [very thick, black!20!blue, -latex](0,0) -- (2.4,0);
	 		\draw [very thick, black!20!blue, -latex](0,0) -- (-2.4,0);
	 		\draw [very thick, black!20!blue, -latex](0,0) -- (1.2,1.2*1.732);
	 		\draw [very thick, black!20!blue, -latex](0,0) -- (-1.2,1.2*1.732);
	 		\draw [very thick, black!20!blue, -latex](0,0) -- (1.2,-1.2*1.732);
	 		\draw [very thick, black!20!blue, -latex](0,0) -- (-1.2,-1.2*1.732);
	 		
	 		\node[anchor=south west, inner sep=5pt] at (0,-0.1) {0};
	 		\node[below] at (1.8,0) {$\lambda_0$};
	 		\node[right] at (0.8,1.35) {$-\omega^2\lambda_0$};
	 		\node[left] at (-0.9,1.35) {$\omega\lambda_0$};
	 		\node[below] at (-1.8,0) {$-\lambda_0$};
	 		\node[left] at (-0.8,-1.35) {$\omega^2\lambda_0$};
	 		\node[right] at (0.8,-1.35) {$-\omega\lambda_0$};
	 		\fill (0,0) circle (1.5pt);
	 		\fill (1.5,0) circle (1.5pt);
	 		\fill (0.75,1.732*0.75) circle (1.5pt);
	 		\fill (-0.75,1.732*0.75) circle (1.5pt);
	 		\fill (-1.5,0) circle (1.5pt);
	 		\fill (-0.75,-1.732*0.75) circle (1.5pt);
	 		\fill (0.75,-1.732*0.75) circle (1.5pt);
	 	\end{tikzpicture}
	 	\caption{The signature of the real part of $\theta_{21}(\lambda)$ for $\vert\xi\vert<1$. In the shaded regions, ${\rm{Re}}\,\theta_{21}>0$; in the white regions, ${\rm{Re}}\,\theta_{21}<0$}
	 	\label{figsignature}
	 \end{figure}
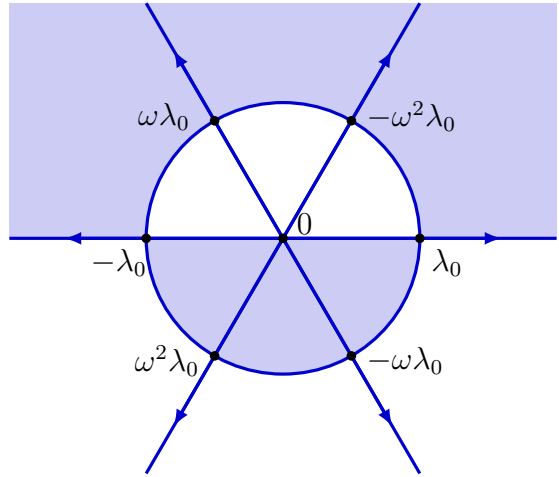

	 \subsection{Long-time asymptotic behavior outside the light cone $\vert \frac{x}{t}\vert>1$}
	 \ \ \ \
	 This subsection divides into two parts to discuss the long-time asymptotic behavior of the solution to the two-component nonlinear KG equation \eqref{3}.

	\subsubsection{The situation of $\xi>1$}\label{411}
	\ \ \ \
	For $\xi>1$, decompose the two jump matrices into the following forms
	\begin{equation*}
		\begin{aligned}
			v_1&=\begin{pmatrix}
				1	&0	&0\\
				r_1^*(\lambda) {\rm{e}}^{\theta_{21}} &1 &0\\
				0 &0&1
			\end{pmatrix}
			\begin{pmatrix}
				1	& -r_1(\lambda) {\rm{e}}^{-\theta_{21}}	&0\\
				0 &1 &0\\
				0 &0&1
			\end{pmatrix}\\
			&:=v_1^l v_1^u,
		\end{aligned}
	\end{equation*}
	and
	\begin{equation*}
		\begin{aligned}
			v_4&=\begin{pmatrix}
				1	& -r_2^*(\lambda) {\rm{e}}^{-\theta_{21}}	&0\\
				0 &1 &0\\
				0 &0&1
			\end{pmatrix}
			\begin{pmatrix}
				1	&0	&0\\
				r_2(\lambda) {\rm{e}}^{\theta_{21}} &1 &0\\
				0 &0&1
			\end{pmatrix}\\
			&:=v_2^u v_2^l.
		\end{aligned}
	\end{equation*}
	
	Without loss of generality, it is assumed that the reflection coefficient $r_1(\lambda)$ can be analytically extended to the first quarter and continuous to the boundary, and the reflection coefficient $r_2(\lambda)$ can be analytically extended to the third quarter and continuous to the boundary. In this case, it follows that $v_1^u$ and $v_2^u$ exponentially decay to identity matrix $I$ as ${\rm{Im}}\, \lambda>0$, $v_1^l$ and $v_2^l$ exponentially decay to the identity matrix $I$ as ${\rm{Im}}\, \lambda<0$ for $t$ approaches infinity. According to the condition of RH problem \ref{rhp}, it is obtained that the solution $M(x,t,\lambda)$ of the RH problem \ref{rhp} is the identity matrix $I$ under the above assumption. Using the reconstruction formula \eqref{7} given in Theorem \ref{reconstruction formula}, one can obtain the solution $u(x,t)$ and $v(x,t)$ of the two-component nonlinear KG equation \eqref{3} approach zero when $t\rightarrow\infty$. 
	
	\begin{lem} \label{lem41}
		For $\xi >1$, the solution of the RH problem \ref{rhp} satisfies $\|M(x,t,\lambda)-I\|\rightarrow 0$ and the solutions $u(x,t)$ and $v(x,t)$ of the two-component nonlinear KG equation \eqref{3} decay rapidly to zero with behavior $\mathcal{O}(x^{-N-3/2})$ in Region ${\rm{I}}$ and with behavior $\mathcal{O}(t^{-N-3/2})$ in Region ${\rm{II}}$ for $N\geq0$ and $x >0$.
	\end{lem}
		\begin{proof}
	  As constructed in Refs. \cite{tzitzeica,tzitzeica-2024} and \cite{lenells-wang-good}, the following decomposition can be obtained:
	  \begin{equation*}
	  	r_1(\lambda){\rm{e}}^{-\theta_{21}(\lambda)}:=r_{1,a}(\lambda){\rm{e}}^{-\theta_{21}(\lambda)}+	r_{1,r}(\lambda){\rm{e}}^{-\theta_{21}(\lambda)},
	  \end{equation*}
	  where $r_{1,a}(\lambda){\rm{e}}^{-\theta_{21}(\lambda)}$ and $r^*_{1,a}(\lambda){\rm{e}}^{\theta_{21}(\lambda)}$ can be analytically continued to the upper and lower half-plane, respectively, due to the Schwartz reflection principle. More importantly, it can be obtained  that $\|r_{1,r}(\lambda){\rm{e}}^{-\theta_{21}(\lambda)}\|_{L^1\cap L^\infty}\rightarrow0$ as ${\rm{Im}}\, \lambda>0$ when $t\rightarrow\infty$. Moreover, for any integer $N\geq0$, $r_{1,r}(\lambda){\rm{e}}^{-\theta_{21}(\lambda)}=\mathcal{O}(x^{-N-3/2})$ in Region ${\rm{I}}$, while in Region ${\rm{II}}$, $r_{1,r}(\lambda){\rm{e}}^{-\theta_{21}(\lambda)}=\mathcal{O}(t^{-N-3/2})$.
	
	  Similarly, by separating $r_2(\lambda){\rm{e}}^{\theta_{21}(\lambda)}:=r_{2,a}(\lambda){\rm{e}}^{-\theta_{21}(\lambda)}+	r_{2,r}(\lambda){\rm{e}}^{\theta_{21}(\lambda)}$, we can get the same properties as above: $r_{2,a}(\lambda){\rm{e}}^{\theta_{21}(\lambda)}$ can be analytically continuous to the lower half complex plane and $r^*_{2,a}(\lambda){\rm{e}}^{-\theta_{21}(\lambda)}$ can be analytically continuous to the upper half complex plane, as well as $\|r_{2,r}(\lambda){\rm{e}}^{\theta_{21}(\lambda)}\|_{L^1\cap L^\infty}\rightarrow0$ as ${\rm{Im}}\, \lambda<0$ when $t\rightarrow\infty$. In Region ${\rm{I}}$, one has $r_{2,r}(\lambda){\rm{e}}^{\theta_{21}(\lambda)}=\mathcal{O}(x^{-N-3/2})$, while in Region ${\rm{II}}$, one has $r_{2,r}(\lambda){\rm{e}}^{\theta_{21}(\lambda)}=\mathcal{O}(t^{-N-3/2})$.
	\end{proof}
	
	Through some transformations, we can always transform the jump matrices in the RH problem \ref{rhp} on the real axis into the following form, based on the decomposition of the jump matrices mentioned above and in Lemma \ref{lem41}:
	\begin{equation*}
		v_1^{(1)}(x,t,\lambda):=\left\{
		\begin{aligned}
			&\begin{pmatrix}
				1 & -r_{1,r}(\lambda){\rm{e}}^{-\theta_{21}(x,t,\lambda)} & 0\\
				r^*_{1,r}(\lambda){\rm{e}}^{\theta_{21}(x,t,\lambda)} & 1-r_{1,r}(\lambda)r^*_{1,r}(\lambda) & 0\\
				0 & 0 & 1\\
			\end{pmatrix},\quad &&\lambda\in \mathbb{R}_+,\\
				&\begin{pmatrix}
				1 & 0 & 0\\
				r^*_{1,a}(\lambda){\rm{e}}^{\theta_{21}(x,t,\lambda)} & 1 & 0\\
				0 & 0 & 1\\
			\end{pmatrix},\quad &&\lambda\in \omega_1\mathbb{R}_+,\\
				&\begin{pmatrix}
				1 & r_{1,a}(\lambda){\rm{e}}^{-\theta_{21}(x,t,\lambda)} & 0\\
				0 & 1 & 0\\
				0 & 0 & 1\\
			\end{pmatrix},\quad &&\lambda\in (\omega_1)^{-1}\mathbb{R}_+,
			\end{aligned}
		\right.
	\end{equation*}
	\begin{equation*}
	v_4^{(1)}(x,t,\lambda):=\left\{
	\begin{aligned}
		&\begin{pmatrix}
			1-r_{2,r}(\lambda)r^*_{2,r}(\lambda) & -r_{2,r}^*(\lambda){\rm{e}}^{-\theta_{21}(x,t,\lambda)} & 0\\
			r_{2,r}(\lambda){\rm{e}}^{\theta_{21}(x,t,\lambda)} & 1 & 0\\
			0 & 0 & 1\\
		\end{pmatrix},\quad &&\lambda\in \mathbb{R}_-,\\
		&\begin{pmatrix}
			1 & -r^*_{2,a}(\lambda){\rm{e}}^{-\theta_{21}(x,t,\lambda)} & 0\\
			0 & 1 & 0\\
			0 & 0 & 1\\
		\end{pmatrix},\quad &&\lambda\in \omega_1\mathbb{R}_-,\\
		&\begin{pmatrix}
			1 & 0 & 0\\
			-r^*_{2,a}(\lambda){\rm{e}}^{\theta_{21}(x,t,\lambda)} & 1 & 0\\
			0 & 0 & 1\\
		\end{pmatrix},\quad &&\lambda\in (\omega_1)^{-1}\mathbb{R}_-,
	\end{aligned}
	\right.
	\end{equation*}
	where $\omega_1={\rm{e}}^{\frac{\pi i}{8}}$. By leveraging the symmetries, it can be deduced that the jump matrices on the remaining four jump lines under the transformation can also be decomposed into the same form as above.
\par
According to the Beals-Coifman theory \cite{Beals-Coifman-1984}, the solution to the RH problem \ref{rhp} can be given, and the solution to the two-component nonlinear KG equation \eqref{3} can be ultimately derived by using the Theorem \ref{reconstruction formula}. Before performing these steps, define the Cauchy operator $\mathcal{C}$ on the contour $\Sigma$ as
	\begin{equation}\label{27}
		(\mathcal{C}f)(\lambda):=\frac{1}{2\pi i}\int_{\Sigma}\frac{f(y)}{y-\lambda} {\rm{d}} y.
	\end{equation}
	If $ f \in \dot{E}^3(\Sigma)$, then $ f^\pm $ exists almost everywhere on $\Sigma$, and $f^\pm \in \dot{L}^3(\Sigma)$ (the definitions of the symbols and the specific form of the theorem can be found in \cite{lenells-2018}). The trivial decomposition of the jump matrix $v^{(1)}$ is considered
	\begin{equation*}
		v^{(1)}:=(b_-)^{-1}b_+=Iv^{(1)},
	\end{equation*} and define $\omega=\omega_++\omega_-$, $\omega_+=b_+-I=v^{(1)}-I$, $\omega_-=I-b_-=O$. Next define the Cauchy projection operator
	\begin{equation*}
		(\mathcal{C}_\pm f)(\lambda)=\lim_{{\scriptsize {\begin{aligned}
						&\lambda'\rightarrow\lambda\in\Sigma\\[-1ex]
						&\lambda'\in \pm side\, of~ \Sigma
		\end{aligned}}}}\frac{1}{2\pi i}\int_{\Sigma}\frac{f(y)}{y-\lambda'} {\rm{d}} y,
	\end{equation*}
	and the operator $\mathcal{C}_\omega f:=\mathcal{C}_+(f\omega_-)+\mathcal{C}_-(f\omega_+)=\mathcal{C}_-(f\omega_+)$. It can also be proven that if $f\in \dot{E}^3(\Sigma) $, then $\mathcal{C}_\omega f$ is the bounded linear operator on $\dot{E}^3(\Sigma)$. Then let $\mu\in I+\dot{E}^3(\Sigma)$ be the solution of the  equation
	\begin{equation*}
		\mu=I+\mathcal{C}_\omega \mu.
	\end{equation*}
	If $I-\mathcal{C}_\omega$ is invertible, the RH problem \ref{rhp} has the unique solution of the form
	\begin{equation*}
		M(\lambda)=I+\frac{1}{2 \pi i}\int_{\Sigma}\frac{\mu(y)\omega(y)}{y-\lambda}{\rm{d}}y.
	\end{equation*}
	
	According to the above theory, if $\omega(x,t,\lambda)=v^{(1)}(x,t,\lambda)-I$, it is obvious that $\omega\in \dot{E}^3(\Sigma)$. In addition, by using the signature of the real part of $\theta_{21}(\lambda)$, we can get $\|\omega\|_{L^\infty(\lambda)}\rightarrow0$ with the behavior of $\mathcal{O}(x^{-N-3/2})$ in the Region ${\rm{I}}$  and with the behavior of $\mathcal{O}(t^{-N-3/2})$ in the Region ${\rm{II}}$ for $N\geq0$ and $x>0$ when $t\rightarrow\infty$, the condition that $I-\mathcal{C}_\omega$ is invertible as mentioned above can be obtained. In other words, the solution $M(x,t,\lambda)$ to the RH problem \ref{rhp} is unique, and $\|M_\pm -I\|_{\dot{E}^3(\Sigma')}\rightarrow I$ rapidly as $t$ approaches infinity. Therefore, it is concluded that $M(x,t,\lambda)\rightarrow I$ as $t \rightarrow \infty$, and by using equation \eqref{7}, it is obtained that $u(x,t)$ and $v(x,t)$ decay exponentially to zero in the two regions as described in the lemma. So the proof for Regions ${\rm{I}}$ and ${\rm{II}}$ in Theorem \ref{region} for the case when $x>0$ has been completed.

	\subsubsection{The situation of $\xi < -1$}\label{412}
	\ \ \ \
	Since the real part signature of $\theta_{21}(\lambda)$ for $\xi < -1$ is the opposite of that for $\xi > -1$, the jump matrix need to be decomposed differently
	\begin{equation}\label{28}
		v_1=\begin{pmatrix}
		1	&-\frac{r_1(\lambda)}{1-r_1(\lambda)r_1^*(\lambda)}{\rm{e}}^{-\theta_{21}}	&0\\
		0 &1 &0\\
		0 &0&1
		\end{pmatrix}
		\begin{pmatrix}
			\frac{1}{1-r_1(\lambda)r_1^*(\lambda)}	&	&0\\
			0 &1-r_1(\lambda)r_1^*(\lambda) &0\\
			0 &0&1
		\end{pmatrix}
		\begin{pmatrix}
			1	&0	&0\\
			\frac{r_1^*(\lambda)}{1-r_1(\lambda)r_1^*(\lambda)}{\rm{e}}^{\theta_{21}} &1 &0\\
			0 &0&1
		\end{pmatrix},
	\end{equation}
	and
		\begin{equation}\label{29}
		v_4=\begin{pmatrix}
			1	&0	&0\\
			\frac{r_2^*(\lambda)}{1-r_2(\lambda)r_2^*(\lambda)}{\rm{e}}^{\theta_{21}} &1 &0\\
			0 &0&1
		\end{pmatrix}
		\begin{pmatrix}
			1-r_2(\lambda)r_2^*(\lambda)	&	&0\\
			0 &\frac{1}{1-r_2(\lambda)r_2^*(\lambda)} &0\\
			0 &0&1
		\end{pmatrix}
		\begin{pmatrix}
			1	&-\frac{r_2(\lambda)}{1-r_2(\lambda)r_2^*(\lambda)}{\rm{e}}^{-\theta_{21}}	&0\\
			0 &1 &0\\
			0 &0&1
		\end{pmatrix}.
	\end{equation}
	To prepare for the subsequent matrix transformation, here we must introduce the functions $\delta_1(\lambda)$ and $\delta_4(\lambda)$ to satisfy the following jump and asymptotic conditions:
    \begin{equation*}
    	\left\{
    	\begin{aligned}
    		&\begin{aligned}\delta_{1+}(\lambda)
    			=&\delta_{1-}(\lambda)(1-r_1(\lambda)r_1^*(\lambda)),\quad  &&\lambda\in \mathbb{R}_+,\\
    			=&\delta_{1-}(\lambda), &&\lambda\in \mathbb{C}\backslash \mathbb{R}_+,
    		\end{aligned}\\
    		&\delta_1(\lambda)\rightarrow1,\qquad\qquad\qquad\qquad\qquad\quad\, \lambda\rightarrow\infty,
    	\end{aligned}
    	\right.
    \end{equation*}
    and
	\begin{equation*}
		\left\{
		\begin{aligned}
			&\begin{aligned}\delta_{4+}(\lambda)
				=&\delta_{4-}(\lambda)(1-r_2(\lambda)r_2^*(\lambda)),\quad  \lambda\in \mathbb{R}_-,\\
				=&\delta_{4-}(\lambda), \qquad\qquad\qquad\qquad \lambda\in \mathbb{C}\backslash \mathbb{R}_-,
			\end{aligned}\\
			&\delta_4(\lambda)\rightarrow1,\quad \qquad\qquad\qquad\qquad\quad\;\, \lambda\rightarrow\infty.
		\end{aligned}
		\right.
	\end{equation*}
	It is obvious that the functions $\delta_1(\lambda)$ and $\delta_4(\lambda)$ can be given in the form of integral equations
	\begin{equation*}
		\delta_1(\lambda)={\rm{exp}}\left\{\frac{1}{2 \pi i}\int_{0}^{\infty}\frac{\ln (1-r_1(s)r_1^*(s))}{s-\lambda}{\rm{d}}s\right\},\quad \lambda\in \mathbb{C}\backslash \mathbb{R}_+,
	\end{equation*}
	\begin{equation*}
		\delta_4(\lambda)={\rm{exp}}\left\{\frac{1}{2 \pi i}\int_{0}^{-\infty}\frac{\ln (1-r_2(s)r_2^*(s))}{s-\lambda}{\rm{d}}s\right\},\quad \lambda\in \mathbb{C}\backslash \mathbb{R}_-,
	\end{equation*}
	where the function $\ln (x)$ involved in the integral is a real-valued function.
\par	
	The other $\delta$ functions can also be given by using the symmetries
	\begin{equation}\label{30}
		\begin{aligned}
			&\delta_2(\lambda)=\delta_4(\omega\lambda), &&\lambda\in \mathbb{C}\backslash \omega^2\mathbb{R}_-,\\
			&\delta_3(\lambda)=\delta_1(\omega^2\lambda), &&\lambda\in \mathbb{C}\backslash \omega\mathbb{R}_+,\\
			&\delta_5(\lambda)=\delta_1(\omega\lambda), &&\lambda\in \mathbb{C}\backslash \omega^2\mathbb{R}_+,\\
			&\delta_6(\lambda)=\delta_4(\omega^2\lambda), &&\lambda\in \mathbb{C}\backslash \omega\mathbb{R}_-.\\
		\end{aligned}
			\end{equation}
	Based on the six $\delta$ functions constructed above, and the jump matrix decomposition given in the equations \eqref{28} and \eqref{29}, define the matrix-valued function
	\begin{equation*}
		\Delta(\lambda)=\begin{pmatrix}
			\frac{\delta_1(\lambda)\delta_6(\lambda)}{\delta_3(\lambda)\delta_4(\lambda)} & 0 & 0\\
			0 & 	\frac{\delta_4(\lambda)\delta_5(\lambda)}{\delta_1(\lambda)\delta_2(\lambda)} & 0\\
			0 & 0 & 	\frac{\delta_2(\lambda)\delta_3(\lambda)}{\delta_5(\lambda)\delta_6(\lambda)}
		\end{pmatrix},
	\end{equation*}
	and take the transformation
	\begin{equation*}
		M^{(1)}(x,t,\lambda)=M(x,t,\lambda)\Delta(\lambda).
	\end{equation*}
  Under this transformation, the jump matrices (the same notation used in subsection \ref{411} is used here, but in the different form) on the real axis corresponding to the new eigenfunction $M^{(1)}(x,t,\lambda)$ become
	\begin{equation}\label{31}
	 	\begin{aligned}
		v_1^{(1)}	&=\begin{pmatrix}
				1-r_1(\lambda)r_1^*(\lambda) & -\frac{\delta_{N1}}{\delta^2_{1-}}\frac{r_1(\lambda)}{1-r_1(\lambda)r_1^*(\lambda)}{\rm{e}}^{-\theta_{21}} & 0\\
				\frac{\delta^2_{1+}}{\delta_{N1}}\frac{r_1^*(\lambda)}{1-r_1(\lambda)r_1^*(\lambda)}{\rm{e}}^{\theta_{21}} & 1 & 0\\
				0 & 0 & 1
			\end{pmatrix}\\
			&=\begin{pmatrix}
				1 & -\frac{\delta_{N1}}{\delta^2_{1-}}\frac{r_1(\lambda)}{1-r_1(\lambda)r_1^*(\lambda)}{\rm{e}}^{-\theta_{21}} & 0\\
				0 & 1 & 0\\
				0 & 0 & 1
			\end{pmatrix}
			\begin{pmatrix}
				1 & 0 & 0\\
				\frac{\delta^2_{1+}}{\delta_{N1}}\frac{r_1^*(\lambda)}{1-r_1(\lambda)r_1^*(\lambda)}{\rm{e}}^{\theta_{21}} & 1 & 0\\
				0 & 0 & 1
			\end{pmatrix}\\
			&:=v_{1,u}^{(1)}v_{1,l}^{(1)},
		\end{aligned}
	\end{equation}
	and
	
		\begin{equation}\label{32}
		\begin{aligned}
			v_4^{(1)}	&=\begin{pmatrix}
				1 & -\frac{\delta_{4+}^2}{\delta_{N4}}\frac{r_2^*(\lambda)}{1-r_2(\lambda)r_2^*(\lambda)}{\rm{e}}^{-\theta_{21}} & 0\\
				\frac{\delta_{N4}}{\delta^2_{4-}}\frac{r_2(\lambda)}{1-r_2(\lambda)r_2^*(\lambda)}{\rm{e}}^{\theta_{21}} & 1 & 0\\
				0 & 0 & 1
			\end{pmatrix}\\
			&=\begin{pmatrix}
				1 & 0 & 0\\
				\frac{\delta_{N4}}{\delta^2_{4-}}\frac{r_2(\lambda)}{1-r_2(\lambda)r_2^*(\lambda)}{\rm{e}}^{\theta_{21}} & 1 & 0\\
				0 & 0 & 1
			\end{pmatrix}
			\begin{pmatrix}
				1 & -\frac{\delta_{4+}^2}{\delta_{N4}}\frac{r_2^*(\lambda)}{1-r_2(\lambda)r_2^*(\lambda)}{\rm{e}}^{-\theta_{21}} & 0\\
				0 & 1 & 0\\
				0 & 0 & 1
			\end{pmatrix}\\
			&:=v_{4,l}^{(1)}v_{4,u}^{(1)},
		\end{aligned}
	\end{equation}
	where $\delta_{N1}:=\frac{\delta_3\delta_4^2\delta_5}{\delta_2\delta_6}$ and $\delta_{N4}:=\frac{\delta_1^2\delta_2\delta_6}{\delta_3\delta_5}$.
	\begin{lem} \label{lem42}
		For $\xi <-1$, the solution of the RH problem \ref{rhp} satisfies $\|M(x,t,\lambda)-I\|\to0$ and the solutions $u(x,t)$ and $v(x,t)$ of the two-component nonlinear KG equation \eqref{3} decay rapidly to zero with behavior $\mathcal{O}(\left| x\right| ^{-N-3/2})$ in Region ${\rm{I}}$ and with behavior $\mathcal{O}(t^{-N-3/2})$ in Region ${\rm{II}}$ for $N\geq0$ and $x <0$.
	\end{lem}
	\begin{proof}
			The decomposition forms \eqref{31} and \eqref{32} of the jump matrices obtained after the eigenfunction transformation are similar to the decomposition of the analytic and non-analytic parts of the reflection coefficients $r_1(\lambda)$ and $r_2(\lambda)$ in subsection \ref{411}. Here we can also do similar operations with $\xi \geq 1$ to obtain the new eigenfunction $M^{(1)}$ exponentially to $I$, similar to the situation in Lemma \ref{lem41} as $t\rightarrow\infty$.
			
			In addition, since we transform the original RH problem \ref{rhp} by using $\delta_j(\lambda)$ $(j=1,\cdots,6)$ functions, and the relation \eqref{7} between the RH problem \ref{rhp} and the solution of the two-component nonlinear KG equation \eqref{3} is based on the asymptotics $\lambda\rightarrow0$, we also need to analyze the behavior of $\delta_j(\lambda)$ for $\lambda\rightarrow 0$. Theorem \ref{theo21} shows that $r_1(\lambda),r_2(\lambda)\rightarrow0$ as $\lambda\rightarrow0$, thus $\delta_1(\lambda)$ and $\delta_4(\lambda)$ are continuous at $\lambda=0$. Then $\delta_1(0)=\delta_3(0)=\delta_5(0)$ and $\delta_2(0)=\delta_4(0)=\delta_6(0)$ are derived by combining the symmetries \eqref{30} satisfied by the $\delta$ functions. It is obvious that $\Delta(\lambda)\rightarrow I$ as $\lambda\rightarrow0$. So in the case of $\xi\leq -1$, the solution of the two-component nonlinear KG equation \eqref{3} satisfies $u(x,t),v(x,t)\rightarrow0$ with behavior $\mathcal{O}(\left| x\right| ^{-N-3/2})$ in Region ${\rm{I}}$ and with behavior $\mathcal{O}(t^{-N-3/2})$ in Region ${\rm{II}}$ for $N\geq0$ and $x <0$.  Combining this with the Lemma \ref{lem41}, the proof of this lemma and the asymptotics for Regions ${\rm{I}}$ and ${\rm{II}}$ in Theorem \ref{region} is completed.
			
	\end{proof}

    \subsection{Long-time asymptotic behavior inside the light cone $\vert\frac{x}{t}\vert< 1$}\label{sec42}
    \ \ \ \
    In this section, the higher-order asymptotic solutions \eqref{8} and \eqref{9} to the initial value problem of two-component nonlinear KG equation \eqref{3} in Theorem \ref{region} are obtained by applying Deift-Zhou higher-order asymptotic theory \cite{higher-order} to the RH problem \ref{rhp}. In addition, since the classification of the real part signature of $\theta_{21}(\lambda)$ in the Fig. \ref{figsignature} is more complex than the previous subsection, this subsection will expend more effect on analyzing the properties of the new RH problem obtained by each transformation, and finally gives the expression of the higher-order asymptotic solutions.
    \par
    In subections \ref{sub421}-\ref{sub425} below, we consider the case of Region ${\rm{IV}}$ in Theorem \ref{region}, and through a series of transformations, we ultimately prove the result for Region ${\rm{IV}}$. The transformations involved in Region ${\rm{III}}$ are largely the same as those in Region ${\rm{IV}}$, but some of the error estimates are changed. Based on the results from the preceding subsections, we provide new error estimates and prove the result for Region ${\rm{III}}$ in Theorem \ref{region} in subection \ref{regionIII}.

    \subsubsection{Perform the first transformation}\label{sub421}
    \ \ \ \
    From the relations of the real part signature of $\theta_{21}(\lambda)$ in different value ranges of $\xi$ as depicted in Fig. \ref{figsignature}, it can be seen that when $\vert \lambda\vert>\lambda_0$, the jump matrix needs to be decomposed into the same form as that for $\xi>1$, in the other case, the jump matrix needs to be decomposed into the same form as that for $\xi<-1$. Next, in order to better utilize decay conditions in different regions of the jump matrix \eqref{6} for case $\xi<-1$, several scalar RH problems need to be introduced to deal with the jump matrix. For ease of writing, the notation below is the same as in the subsection \ref{412}, but here the $\delta_j(\lambda)$ $(j=1,2,\cdots,6)$ functions satisfy different properties than in the previous subsection. Next introduce two functions
     \begin{equation*}
    	\left\{
    	\begin{aligned}
    		&\begin{aligned}\delta_{1+}(\lambda)
    			=&\delta_{1-}(\lambda)(1-r_1(\lambda)r_1^*(\lambda)),\quad  0<\lambda<\lambda_0,\\
    			=&\delta_{1-}(\lambda), \qquad\qquad\qquad\qquad \lambda\in \mathbb{C}\backslash (0,\lambda_0),
    		\end{aligned}\\
    		&\delta_1(\lambda)\rightarrow1,\quad \qquad\qquad\qquad\qquad\quad\;\, \lambda\rightarrow\infty,
    	\end{aligned}
    	\right.
    \end{equation*}
    and

    \begin{equation*}
    	\left\{
    	\begin{aligned}
    		&\begin{aligned}\delta_{4+}(\lambda)
    			=&\delta_{4-}(\lambda)(1-r_2(\lambda)r_2^*(\lambda)),\quad  -\lambda_0<\lambda<0,\\
    			=&\delta_{4-}(\lambda), \qquad\qquad\qquad\qquad \lambda\in \mathbb{C}\backslash (-\lambda_0,0),
    		\end{aligned}\\
    		&\delta_4(\lambda)\rightarrow1,\quad \qquad\qquad\qquad\qquad\quad\;\, \lambda\rightarrow\infty.
    	\end{aligned}
    	\right.
    \end{equation*}
    Similarly, it yields that
    	\begin{equation}\label{33}
    	\delta_1(\lambda)={\rm{exp}}\left\{\frac{1}{2 \pi i}\int_{0}^{\lambda_0}\frac{\ln (1-r_1(s)r_1^*(s))}{s-\lambda}{\rm{d}}s\right\}:={\rm{e}}^{\chi_1(\lambda)},\quad \lambda\in \mathbb{C}\backslash (0,\lambda_0),
    \end{equation}
    \begin{equation}\label{34}
    	\delta_4(\lambda)={\rm{exp}}\left\{\frac{1}{2 \pi i}\int_{0}^{-\lambda_0}\frac{\ln (1-r_2(s)r_2^*(s))}{s-\lambda}{\rm{d}}s\right\}:={\rm{e}}^{\chi_4(\lambda)},\quad \lambda\in \mathbb{C}\backslash (-\lambda_0,0).
    \end{equation}
    By a similar method as the preceding subsection, the other four $\delta$ functions are given by using the symmetries of the jump relation, i.e.,
    \begin{equation}\label{35}
    	\begin{aligned}
    		&\delta_2(\lambda)=\delta_4(\omega\lambda), &&\lambda\in \mathbb{C}\backslash (-\omega^2\lambda_0,0),\\
    		&\delta_3(\lambda)=\delta_1(\omega^2\lambda), &&\lambda\in \mathbb{C}\backslash (0,\omega\lambda_0),\\
    		&\delta_5(\lambda)=\delta_1(\omega\lambda), &&\lambda\in \mathbb{C}\backslash (0,\omega^2\lambda_0),\\
    		&\delta_6(\lambda)=\delta_4(\omega^2\lambda), &&\lambda\in \mathbb{C}\backslash (-\omega\lambda_0,0).\\
    	\end{aligned}
    \end{equation}
    Taking $\delta_j(\lambda)$ $(j=1,4)$ as examples, some properties satisfied by the two functions are given, and efforts are also made to get the final goal. The case of $\delta_j(\lambda)$ for $j= 2,3,5,6$ can also share similar properties by using the symmetries.

    \begin{prop}\label{prop41}
    	The functions $\delta_1(\lambda)$ and $\delta_4(\lambda)$ satisfy the following properties:
    	\begin{enumerate}
    		\item  $\delta_1(\lambda)$ is analytic for $\lambda\in \mathbb{C}\backslash (0,\lambda_0)$ and $\delta_4(\lambda)$ is analytic for $\lambda\in \mathbb{C}\backslash (-\lambda_0,0)$, and both of them are continuous up to the boundary of their domains.
    		\item The functions $\chi_j$ $(j=1,4)$ can be expanded in the following forms
    		\begin{equation*}
    			\begin{aligned}
    				\chi_1(\lambda)=&i\nu_1(\lambda_0)\log_{-\pi}(\lambda-\lambda_0)+[\chi_{11}^{(1)}(\lambda_0)\log_{-\pi}(\lambda-\lambda_0)+\chi_{11}^{(2)}(\lambda_0)](\lambda-\lambda_0)+\cdots\\
    				&+\![\chi_{1N}^{(1)}(\lambda_0)\log_{-\pi}(\lambda-\lambda_0)\!+\!\chi_{1N}^{(2)}(\lambda_0)](\lambda-\lambda_0)^N\!+\!\mathcal{O}((\lambda-\lambda_0)^{N+1}\log_{-\pi}(\lambda-\lambda_0)),\\
    			\end{aligned}
    		\end{equation*}
    			\begin{equation*}
    			\begin{aligned}
    				\chi_4(\lambda)=&i\nu_4(-\lambda_0)\log_0(\lambda+\lambda_0)+[\chi_{41}^{(1)}(-\lambda_0)\log_0(\lambda+\lambda_0)+\chi_{41}^{(2)}(-\lambda_0)](\lambda+\lambda_0)+\cdots\\
    				&+\![\chi_{4N}^{(1)}(-\lambda_0)\log_0(\lambda+\lambda_0)\!+\!\chi_{4N}^{(2)}(-\lambda_0)](\lambda+\lambda_0)^N\!+\!\mathcal{O}((\lambda+\lambda_0)^{N+1}\log_0(\lambda+\lambda_0)),\\
    			\end{aligned}
    		\end{equation*}
    	where $\log_{-\pi}(z):=\ln\vert z\vert+i \arg_{-\pi}(z)$, $\arg_{-\pi}(z)\in (-\pi,\pi)$, $\log_{0}(z):=\ln \vert z\vert+i \arg_{0}(z)$, $\arg_0(z)\in (0,2\pi)$, and
    	\begin{equation}\label{nu14}
    		\nu_1(\lambda_0)=-\frac{1}{2\pi}\ln(1-\vert r_1(\lambda_0)\vert^2),\qquad \nu_4(-\lambda_0)=-\frac{1}{2\pi}\ln(1-\vert r_2(-\lambda_0)\vert^2).
    	\end{equation}
    Here $\chi_{ij}^{(k)}$ $(i=1,4,j=1,2,\cdots,N,k=1,2)$ are smooth functions that can also be obtained by integrating by parts over $\chi_1(\lambda)$ many times.
    	\item $\delta_1(\lambda)$ and $\delta_4(\lambda)$ are bounded in their domains and satisfy the conjugation symmetries
    	\begin{equation*}
    		\delta_1^{-1}(\lambda)= \overline{\delta_1(\overline{\lambda})},\qquad\lambda\in \mathbb{C}\backslash (0,\lambda_0),
    	\end{equation*}
    	and
    		\begin{equation*}
    		\delta_4^{-1}(\lambda)= \overline{\delta_4(\overline{\lambda})},\qquad\lambda\in \mathbb{C}\backslash (-\lambda_0,0).
    	\end{equation*}
    	\end{enumerate}
    \end{prop}

    The proof of the this proposition can refer to Ref. \cite{higher-order}, which are not stated much here.

    Similarly, construct the matrix transformation
    	\begin{equation*}
    	\Delta(\lambda)=\begin{pmatrix}
    		\frac{\delta_1(\lambda)\delta_6(\lambda)}{\delta_3(\lambda)\delta_4(\lambda)} & 0 & 0\\
    		0 & 	\frac{\delta_4(\lambda)\delta_5(\lambda)}{\delta_1(\lambda)\delta_2(\lambda)} & 0\\
    		0 & 0 & 	\frac{\delta_2(\lambda)\delta_3(\lambda)}{\delta_5(\lambda)\delta_6(\lambda)}
    	\end{pmatrix}
    			:=
    			\begin{pmatrix}
    					\Delta_1(\lambda) & 0 & 0\\
    					0 & 	\Delta_2(\lambda) & 0\\
    					0 & 0 & 	\Delta_3(\lambda)
    				\end{pmatrix},\quad 
    \end{equation*}
   which satisfies the symmetry $\Delta(\lambda)=\tilde{\sigma}_j^{-1}\Delta(\omega^j\lambda)\tilde{\sigma}_j$ for $j=1,2$.

    The new matrix-valued eigenfunction $M^{(1)}(x,t,\lambda)$ is obtained by transforming the eigenfunction involved in RH problem \ref{rhp} by
    	\begin{equation}\label{36}
    	M^{(1)}(x,t,\lambda)=M(x,t,\lambda)\Delta(\lambda).
    \end{equation}
    Obviously we can derive that $\delta_j(\lambda)$ $(j=1,4)$ are continuous at $\lambda=0$, and $\delta_1(0)=\delta_3(0)=\delta_5(0)$, $\delta_2(0)=\delta_4(0)=\delta_6(0)$ are derived by combining the symmetries in \eqref{35}, which implies that $\Delta(\lambda)\rightarrow I$ as $\lambda\rightarrow0$. Obviously, it follows that $\delta_j\rightarrow1$ $(j=1,2,\cdots,6)$ and $\Delta \rightarrow I$  as $\lambda\rightarrow\infty$.
    	\begin{rhp}\label{rhp1}
    	Find a $3\times3$ matrix-valued function $M^{(1)}(x,t,\lambda)$ with the following properties:
    	\begin{itemize}
    		\item $M^{(1)}(x,t,\lambda)$ is holomorphic in $\mathbb{C}\setminus\Sigma^{(1)}$, where 	the contour $\Sigma^{(1)}$ is shown in Fig. \ref{figSigma1}.
    		\item  $M^{(1)}(x,t,\lambda)$ satisfies the following jump conditions
    		\begin{equation*}
    			M^{(1)}_+(x,t,\lambda)=M^{(1)}_-(x,t,\lambda)v_j^{(1)}(x,t,\lambda),\qquad \lambda
    			\in \Sigma_j^{(1)},
    		\end{equation*}
    		where only the forms $v_j^{(1)}$ $(j=1,4,7,10)$ are given below
    		
    		\begin{align}\label{37}
    			v^{(1)}_1 &=\begin{pmatrix}
    				1 & -\frac{\delta_{N1}}{\delta^2_1}r_1(\lambda){\rm{e}}^{-\theta_{21}} & 0\\
    				\frac{\delta^2_1}{\delta_{N1}}r^*_1(\lambda){\rm{e}}^{\theta_{21}} & 1-r_1(\lambda)r_1^*(\lambda) & 0\\
    				0 & 0 & 1
    			\end{pmatrix}, \qquad&&\lambda>\lambda_0,\nonumber\\
    			v_4^{(1)}&=\begin{pmatrix}
    		   1-r_2(\lambda)r_2^*(\lambda) & -\frac{\delta^2_4}{\delta_{N4}}r_2^*(\lambda){\rm{e}}^{-\theta_{21}} & 0\\
    		   \frac{\delta_{N4}}{\delta^2_4}r_2(\lambda){\rm{e}}^{\theta_{21}} & 1 & 0\\
    			0 & 0 & 1
    		\end{pmatrix}, &&\lambda<-\lambda_0,\\
    			v^{(1)}_7 &=\begin{pmatrix}
    			1-r_1(\lambda)r_1^*(\lambda) & -\frac{\delta_{N1}}{\delta^2_{1-}}\frac{r_1(\lambda)}{1-r_1(\lambda)r_1^*(\lambda)}{\rm{e}}^{-\theta_{21}} & 0\\
    			\frac{\delta^2_{1+}}{\delta_{N1}}\frac{r_1^*(\lambda)}{1-r_1(\lambda)r_1^*(\lambda)}{\rm{e}}^{\theta_{21}} & 1 & 0\\
    			0 & 0 & 1
    		\end{pmatrix}, && 0<\lambda<\lambda_0,\nonumber\\
    		v^{(1)}_{10} &=\begin{pmatrix}
    			1 & -\frac{\delta^2_{4+}}{\delta_{N4}}\frac{r_2^*(\lambda)}{1-r_2(\lambda)r_2^*(\lambda)}{\rm{e}}^{-\theta_{21}} & 0\\
    			\frac{\delta_{N4}}{\delta^2_{4-}}\frac{r_2(\lambda)}{1-r_2(\lambda)r_2^*(\lambda)}{\rm{e}}^{\theta_{21}} & 1-r_2(\lambda)r_2^*(\lambda) & 0\\
    			0 & 0 & 1
    		\end{pmatrix}, && -\lambda_0<\lambda<0,\nonumber
    		\end{align}
    		and the rest of the jump matrices can be obtained according to the symmetries satisfied by $v_{j+2}^{(1)}(\lambda)=\tilde{\sigma}_1^{-1}v_j^{(1)}(\omega^2\lambda)\tilde{\sigma}_1$. Here $\delta_{N1}:=\frac{\delta_3\delta_4^2\delta_5}{\delta_2\delta_6}$ and $\delta_{N4}:=\frac{\delta_1^2\delta_2\delta_6}{\delta_3\delta_5}$, and $r_1(\lambda)$, $r_2(\lambda)$ are given by \eqref{4}.
    		\item  As $\lambda\rightarrow\infty$,  $M^{(1)}(x,t,\lambda)=I+\mathcal{O}(\frac{1}{\lambda})$ for $\lambda\in\mathbb{C}\setminus\Sigma^{(1)}$.
    		\item  As $\lambda\rightarrow0$,  $M^{(1)}(x,t,\lambda)=G(x,t)+\mathcal{O}({\lambda})$ for $\lambda\in\mathbb{C}\setminus\Sigma^{(1)}$.
    		\item  $M^{(1)}(x,t,\lambda)=\tilde{\sigma}_1^{-1}M^{(1)}(x,t,\omega\lambda)\tilde{\sigma}_1=\tilde{\sigma}_2\overline{M^{(1)}(x,t,\overline{\lambda})}\tilde{\sigma}_2^{-1}$.
    	\end{itemize}
    \end{rhp}
     \begin{figure}[htbp]
    	\centering
    	\begin{tikzpicture}[scale=1.2]
    		\draw [very thick,black!20!blue](0,0) -- (3,0);
    		\draw [very thick,black!20!blue](-3,0) -- (0,0);
    		\draw [very thick,black!20!blue](0,0) -- (1.5,1.5*1.732);
    		\draw [very thick,black!20!blue](-1.5,-1.5*1.732) -- (0,0);
    		\draw [very thick,black!20!blue](0,0) -- (-1.5,1.732*1.5);
    		\draw [very thick,black!20!blue](0,0) -- (1.5,-1.5*1.732);
    		
    		\draw[very thick, black!20!blue, -latex] (0,0) -- (1,0);
    		\draw[very thick, black!20!blue, -latex] (0,0) -- (0.5,0.5*1.732);
    		\draw[very thick, black!20!blue, -latex] (0,0) -- (-0.5,0.5*1.732);
    		\draw[very thick, black!20!blue, -latex] (0,0) -- (-1,0);
    		\draw[very thick, black!20!blue, -latex] (0,0) -- (-0.5,-0.5*1.732);
    		\draw[very thick, black!20!blue, -latex] (0,0) -- (0.5,-0.5*1.732);
    		
    		\draw[very thick, black!20!blue, -latex] (0,0) -- (2.4,0);
    		\draw[very thick, black!20!blue, -latex] (0,0) -- (1.2,1.2*1.732);
    		\draw[very thick, black!20!blue, -latex] (0,0) -- (-1.2,1.2*1.732);
    		\draw[very thick, black!20!blue, -latex] (0,0) -- (-2.4,0);
    		\draw[very thick, black!20!blue, -latex] (0,0) -- (-1.2,-1.2*1.732);
    		\draw[very thick, black!20!blue, -latex] (0,0) -- (1.2,-1.2*1.732);
    		
    		\node[right] at (3,0) {$\Sigma^{(1)}$};
    		\node[below] at (1.8,0) {$\lambda_0$};
    		\node[right] at (0.8,1.35) {$-\omega^2\lambda_0$};
    		\node[left] at (-0.9,1.35) {$\omega\lambda_0$};
    		\node[below] at (-1.8,0) {$-\lambda_0$};
    		\node[left] at (-0.8,-1.35) {$\omega^2\lambda_0$};
    		\node[right] at (0.8,-1.35) {$-\omega\lambda_0$};
    		
    		
    		\node[red!70!black,below] at (2.3,0) {$1$};
    		\node[red!70!black,left] at (1.1,1.2*1.6) {$2$};
    		\node[red!70!black,left] at (-1.2,1.2*1.6) {$3$};
    		\node[red!70!black,below] at (-2.3,0) {$4$};
    		\node[red!70!black,right] at (-1.05,-1.2*1.6) {$5$};
    		\node[red!70!black,right] at (1.2,-1.2*1.6) {$6$};

         	\node[red!70!black,above] at (0.8,-0.1) {$7$};
    		\node[red!70!black,left] at (0.4,0.7) {$8$};
    		\node[red!70!black,left] at (-0.5,0.6) {$9$};
    		\node[red!70!black,anchor=north east] at (-0.7,0) {$10$};
    		\node[red!70!black,right] at (-0.45,-0.9) {$11$};
    		\node[red!70!black,right] at (0.4,-0.6) {$12$};
    		\fill (0,0) circle (1.5pt);
    		\fill (1.5,0) circle (1.5pt);
    		\fill (0.75,1.732*0.75) circle (1.5pt);
    		\fill (-0.75,1.732*0.75) circle (1.5pt);
    		\fill (-1.5,0) circle (1.5pt);
    		\fill (-0.75,-1.732*0.75) circle (1.5pt);
    		\fill (0.75,-1.732*0.75) circle (1.5pt);
    	\end{tikzpicture}
    	\caption{The jump contour $\Sigma^{(1)}$ in the complex $\lambda$-plane}
    	\label{figSigma1}
    \end{figure}
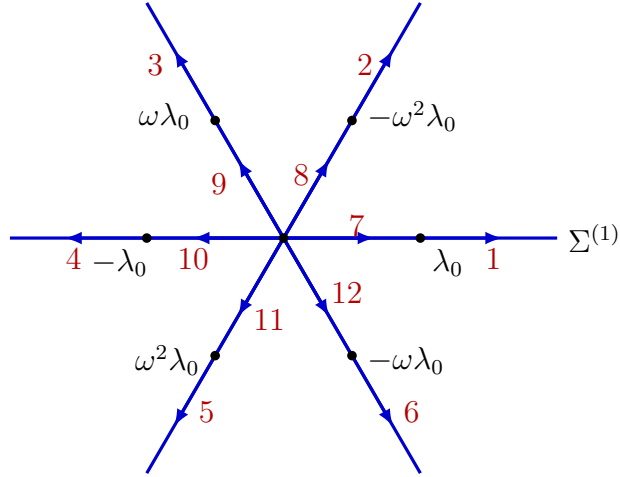

     \subsubsection{Perform the second transformation}\label{422}
    \ \ \ \
    In the last transformation, we have introduced six scalar RH problems and transformed the original RH problem \ref{rhp} to get the new RH problem \ref{rhp1}, which is also in preparation for opening the jump on $\Sigma$. For ease of writing later, introduce the notations
    \begin{equation*}
    	\begin{aligned}
    		&\hat{r}_1(\lambda)=\frac{r_1(\lambda)}{1-r_1(\lambda)r_1^*(\lambda)},\qquad 0<\lambda<\lambda_0,\\
    		&\hat{r}_2(\lambda)=\frac{r_2^*(\lambda)}{1-r_2(\lambda)r_2^*(\lambda)},\qquad -\lambda_0<\lambda<0,
    	\end{aligned}
    \end{equation*}
     then the scattering coefficients involved in the jump relation are decomposed in the following lemma.

     \begin{lem}\label{lem43}
     	The scattering coefficients $r_j(\lambda)$ and $\hat{r}_j(\lambda)$ $(j=1,2)$ can be decomposed into
     	\begin{equation*}
     		\begin{aligned}
     			&r_1(\lambda)=r_{1,a}(\lambda)+r_{1,r}(\lambda),\qquad \lambda\in\left[\lambda_0,\infty\right),\\
     			&r_2^*(\lambda)=r^*_{2,a}(\lambda)+r^*_{2,r}(\lambda),\qquad \lambda\in\left(-\infty,-\lambda_0\right],\\
     			&\hat{r}_1(\lambda)=\hat{r}_{1,a}(\lambda)+\hat{r}_{1,r}(\lambda),\qquad \lambda\in\left[0,\lambda_0\right),\\
     			&\hat{r}_2(\lambda)=\hat{r}_{2,a}(\lambda)+\hat{r}_{2,r}(\lambda),\qquad \lambda\in\left(-\lambda_0,0\right],\\
       		\end{aligned}
     	\end{equation*}
     	and the decompositions have the following properties:
     	\begin{enumerate}
     		\item The function $r_{1,a}$ can be analytically continuous to $\lambda>\lambda_0$, ${\rm{Im}}\,\lambda>0$, and $r^*_{2,a}$ can be analytically continuous to $\lambda<-\lambda_0$, ${\rm{Im}}\,\lambda>0$, respectively. The function $\hat{r}_{1,a}$ can be analytically continuous to $0<\lambda<\lambda_0$, ${\rm{Im}}\,\lambda<0$, and $\hat{r}_{2,a}$ can be analytically continuous to $-\lambda_0<\lambda<0$, ${\rm{Im}}\,\lambda<0$, respectively.
     		\item These functions satisfy the following estimators for any positive integer $N$
     		\begin{equation*}
     			\begin{aligned}
     					&\left|  r_{1,a}(\lambda)-\sum_{n=0}^{N}\frac{r^{(n)}_1(\lambda_0)}{n!}(\lambda-\lambda_0)^n \right| \leq C{\rm{e}}^{\frac{t}{4}{\vert{\rm{Re}}\,\vartheta_{21}(\lambda)\vert}}\vert \lambda-\lambda_0\vert^{N+1},\\
     					&\left|  r_{1,a}(\lambda)\right| \leq \frac{C}{1+\vert\lambda\vert} {\rm{e}}^{\frac{t}{4}\vert{\rm{Re}}\,\vartheta_{21}(\lambda)\vert}, \quad
     					 \lambda\in\{{\rm{Re}}\,\lambda>\lambda_0\}\cap\{ {\rm{Im}}\,\lambda>0\};
     			\end{aligned}
     			\end{equation*}
     			\begin{equation*}
     			\begin{aligned}
     				&\left|   r^*_{2,a}(\lambda)-\sum_{n=0}^{N}\frac{r_2^{*(n)}(-\lambda_0)}{n!}(\lambda+\lambda_0)^n\right| \leq C{\rm{e}}^{\frac{t}{4}\vert{\rm{Re}}\,\vartheta_{21}(\lambda)\vert}\vert \lambda+\lambda_0\vert^{N+1},\\
     				&	\left|  r^*_{2,a}(\lambda)\right| \leq \frac{C}{1+\vert\lambda\vert} {\rm{e}}^{\frac{t}{4}\vert{\rm{Re}}\,\vartheta_{21}(\lambda)\vert}, \quad
     				\lambda\in\{{\rm{Re}}\,\lambda<-\lambda_0\}\cap\{ {\rm{Im}}\,\lambda>0\};
     			\end{aligned}
     			\end{equation*}
     			and
     			\begin{equation*}
     				\begin{aligned}
     				&\left|  \hat{r}_{1,a}(\lambda)-\sum_{n=0}^{N}\frac{\hat{r}_1^{(n)}(\lambda_0)}{n!}(\lambda-\lambda_0)^n\right| \leq C{\rm{e}}^{\frac{t}{4}\vert{\rm{Re}}\,\vartheta_{21}(\lambda)\vert}\vert \lambda-\lambda_0\vert^{N+1},\\
     				&\left|  \hat{r}_{1,a}(\lambda)\right| \leq \frac{C}{1+\vert\lambda\vert} {\rm{e}}^{\frac{t}{4}\vert{\rm{Re}}\,\vartheta_{21}(\lambda)\vert}, \quad
     				\lambda\in\{B_{\lambda_0}(0)\}\cap\{ {\rm{Im}}\,\lambda<0\}\cap\{\frac{\lambda_0}{2}<{\rm{Re}}\,\lambda<\lambda_0\};
     				\end{aligned}
     			\end{equation*}
     			\begin{equation*}
     				\begin{aligned}
     				&\left|   \hat{r}_{2,a}(\lambda)-\sum_{n=0}^{N}\frac{\hat{r}_2(-\lambda_0)}{n!}(\lambda+\lambda_0)^n\right| \leq C{\rm{e}}^{\frac{t}{4}\vert{\rm{Re}}\,\vartheta_{21}(\lambda)\vert}\vert \lambda+\lambda_0\vert^{N+1},\\
     				&\left|  \hat{r}_{2,a}(\lambda)\right| \leq \frac{C}{1+\vert\lambda\vert} {\rm{e}}^{\frac{t}{4}\vert{\rm{Re}}\,\vartheta_{21}(\lambda)\vert}, \quad		
     				\lambda\in\{B_{\lambda_0}(0)\}\cap\{ {\rm{Im}}\,\lambda<0\}\cap\{-\lambda_0<{\rm{Re}}\,\lambda<-\frac{\lambda_0}{2}\}.
     				\end{aligned}
     			\end{equation*}
     		where $\vartheta_{21}=\theta_{21}/t$.
     		\item For any $1\leq p\leq \infty$, the reflection coefficients satisfy the following equations
     		\begin{equation*}
     			\begin{aligned}
     				&\|(1+\vert\cdot\vert)r_{1,r}(x,t,\cdot){\rm{e}}^{-t\vartheta_{21}}\|_{L^p(\lambda_0,\infty)}\leq \frac{C}{t^{N+\frac{3}{2}}},\\	&\|(1+\vert\cdot\vert)r^*_{2,r}(x,t,\cdot){\rm{e}}^{-t\vartheta_{21}}\|_{L^p(-\infty,-\lambda_0)}\leq \frac{C}{t^{N+\frac{3}{2}}},\\	&\|(1+\vert\cdot\vert)\hat{r}_{1,r}(x,t,\cdot){\rm{e}}^{-t\vartheta_{21}}\|_{L^p(0,\lambda_0)}\leq \frac{C}{t^{N+\frac{3}{2}}},\\	&\|(1+\vert\cdot\vert)\hat{r}^*_{2,r}(x,t,\cdot){\rm{e}}^{-t\vartheta_{21}}\|_{L^p(-\lambda_0,0)}\leq \frac{C}{t^{N+\frac{3}{2}}}.\\
     			\end{aligned}
     		\end{equation*}
     		In particular, $\hat{r}_j(x,t,\lambda)=\mathcal{O}(t^{-N-1/2})$ for $0<\left|\lambda \right|<\frac{\lambda_0}{2} $.
     	\end{enumerate}
     \end{lem}

    The proof of this lemma refers to the analytical framework introduced by Cheng, Venakides, and Zhou in their study of the sine-Gordon equation \cite{Zhou-sG}, which are not proved here.

    The next work is to deal with the contour $\Sigma^{(1)}$ to prepare for the higher-order asymptotics later. To do so, divide the $\lambda$-complex plane into different defined regions according to the six critical points above, as shown in the Fig. \ref{figSigma2}. Here, only the numbering of the jump contour branches and the regions at critical points $\pm\lambda_0$ are provided. The remaining four critical points, which are symmetric to $\pm\lambda_0$, can be obtained through symmetry. Then define the following transformation according to the partition in the Fig. \ref{figSigma2}. Give the matrix transformation $W(x,t,\lambda)=\bigcup\limits_{j=0}^2W_{\omega^j\lambda_0}(x,t,\lambda)\bigcup\limits_{j=0}^2W_{-\omega^j\lambda_0}(x,t,\lambda)$, which is also the second transformation in this process.

    \begin{figure}[htbp]
    	\centering
    	\begin{tikzpicture}[scale=0.6]

    		\draw [very thick,black!40!green](0,0) -- (8.5,0);
    		\draw [very thick,black!40!green](-8.5,0) -- (0,0);
    		
    		\draw [very thick,black!20!blue](0,0) -- (4,2);
    		\draw [very thick,black!20!blue](4,2) -- (8,-2);
    		\draw [very thick,black!20!blue](0,0) -- (-4,2);
    		\draw [very thick,black!20!blue](-4,2) -- (-8,-2);
    		
    	    \draw [very thick,black!20!blue](0,0) -- (4,-2);
    	    \draw [very thick,black!20!blue](4,-2) -- (8,2);
    	    \draw [very thick,black!20!blue](0,0) -- (-4,-2);
    	    \draw [very thick,black!20!blue](-4,-2) -- (-8,2);
    	
    	    \draw[very thick, black!20!blue, -latex]  (4,2) -- (2,1);
    	    \draw[very thick, black!20!blue, -latex]  (4,-2) -- (2,-1);
    	    \draw[very thick, black!20!blue, -latex]  (-4,2) -- (-2,1);
    	    \draw[very thick, black!20!blue, -latex]  (-4,-2) -- (-2,-1);
    	
    	    \draw[very thick, black!20!blue, -latex]  (6,0) -- (4.8,1.2);
    	    \draw[very thick, black!20!blue, -latex]  (6,0) -- (4.8,-1.2);
    	    \draw[very thick, black!20!blue, -latex]  (-6,0) -- (-4.8,1.2);
    	    \draw[very thick, black!20!blue, -latex]  (-6,0) -- (-4.8,-1.2);
    	
    	    \draw[very thick, black!40!green, -latex]  (0,0) -- (3,0);
    	    \draw[very thick, black!40!green, -latex]  (0,0) -- (-3,0);
    	    \draw[very thick, black!40!green, -latex]  (0,0) -- (7.3,0);
    	    \draw[very thick, black!40!green, -latex]  (0,0) -- (-7.3,0);
    	
    	    \draw[very thick, black!20!blue, -latex]  (6,0) -- (7.2,1.2);
    	    \draw[very thick, black!20!blue, -latex]  (6,0) -- (7.2,-1.2);
    	    \draw[very thick, black!20!blue, -latex]  (-6,0) -- (-7.2,1.2);
    	    \draw[very thick, black!20!blue, -latex]  (-6,0) -- (-7.2,-1.2);

	        \fill (0,0) circle (2.5pt);
            \fill (6,0) circle (2.5pt);
            \fill (-6,0) circle (2.5pt);
    		
    		\node[right] at (8.5,0) {$\Sigma^{(2)}$};
    		\node[below] at (6,0) {$\lambda_0$};
    		\node[below] at (-6,0) {$-\lambda_0$};
    		
    		\node[red!70!black,above] at (7,1) {\small$1$};
    		\node[red!70!black,above] at (4,0.75) {\small$2$};
    		\node[red!70!black,below] at (4,-0.75) {\small$3$};
    		\node[red!70!black,below] at (7,-1) {\small$4$};
    		\node[red!70!black,below] at (3,0.2) {\small$5$};
    		\node[red!70!black,below] at (7.4,0.2) {\small$6$};
    		
    		\node[red!70!black,below] at (-7,-1) {\small$7$};
    		\node[red!70!black,below] at (-4,-0.75) {\small$8$};
    		\node[red!70!black,above] at (-4,0.75) {\small$9$};
    		\node[red!70!black,above] at (-7,1) {\footnotesize$10$};
    		\node[red!70!black,above] at (-3,-0.7) {\footnotesize$11$};
    		\node[red!70!black,above] at (-7.4,-0.7) {\footnotesize$12$};
    		
    		\node[below] at (8,1.2) {\scriptsize$D_{11}$};
    		\node[below] at (6,2) {\scriptsize$D_{12}$};
    		\node[below] at (3,1) {\scriptsize$D_{13}$};
    		\node[below] at (3,-0.5) {\scriptsize$D_{14}$};
    		\node[below] at (6,-1.2) {\scriptsize$D_{15}$};
    		\node[below] at (8,-0.4) {\scriptsize$D_{16}$};
    		
    		\node[below] at (-8,-0.5) {\scriptsize$D_{41}$};
    		\node[below] at (-6,-1.2) {\scriptsize$D_{42}$};
    		\node[below] at (-3,-0.5) {\scriptsize$D_{43}$};
    		\node[below] at (-3,1.2) {\scriptsize$D_{44}$};
    		\node[below] at (-6,2) {\scriptsize$D_{45}$};
    		\node[below] at (-8,1.2) {\scriptsize$D_{46}$};

    		\coordinate (center) at (0,0);
    	
    	 \begin{scope}[rotate around={60:(center)}]

    	\draw [very thick,black!40!green](0,0) -- (8.5,0);
    	\draw [very thick,black!40!green](-8.5,0) -- (0,0);
    	
    	\draw [very thick,black!20!blue](0,0) -- (4,2);
    	\draw [very thick,black!20!blue](4,2) -- (8,-2);
    	\draw [very thick,black!20!blue](0,0) -- (-4,2);
    	\draw [very thick,black!20!blue](-4,2) -- (-8,-2);
    	
    	\draw [very thick,black!20!blue](0,0) -- (4,-2);
    	\draw [very thick,black!20!blue](4,-2) -- (8,2);
    	\draw [very thick,black!20!blue](0,0) -- (-4,-2);
    	\draw [very thick,black!20!blue](-4,-2) -- (-8,2);
    	
    	\draw[very thick, black!20!blue, -latex]  (4,2) -- (2,1);
    	\draw[very thick, black!20!blue, -latex]  (4,-2) -- (2,-1);
    	\draw[very thick, black!20!blue, -latex]  (-4,2) -- (-2,1);
    	\draw[very thick, black!20!blue, -latex]  (-4,-2) -- (-2,-1);
    	
    	\draw[very thick, black!20!blue, -latex]  (6,0) -- (4.8,1.2);
    	\draw[very thick, black!20!blue, -latex]  (6,0) -- (4.8,-1.2);
    	\draw[very thick, black!20!blue, -latex]  (-6,0) -- (-4.8,1.2);
    	\draw[very thick, black!20!blue, -latex]  (-6,0) -- (-4.8,-1.2);
    	
    	\draw[very thick, black!40!green, -latex]  (0,0) -- (3,0);
    	\draw[very thick, black!40!green, -latex]  (0,0) -- (-3,0);
    	\draw[very thick, black!40!green, -latex]  (0,0) -- (7.3,0);
    	\draw[very thick, black!40!green, -latex]  (0,0) -- (-7.3,0);
    	
    	\draw[very thick, black!20!blue, -latex]  (6,0) -- (7.2,1.2);
    	\draw[very thick, black!20!blue, -latex]  (6,0) -- (7.2,-1.2);
    	\draw[very thick, black!20!blue, -latex]  (-6,0) -- (-7.2,1.2);
    	\draw[very thick, black!20!blue, -latex]  (-6,0) -- (-7.2,-1.2);

    	\fill (0,0) circle (2.5pt);
    	\fill (6,0) circle (2.5pt);
    	\fill (-6,0) circle (2.5pt);

    	 \end{scope}
    	
    	 \node[below] at (4,3*1.8) {$-\omega^2\lambda_0$};
    	 \node[below] at (-4,-3*1.4) {$\omega^2\lambda_0$};

    	 	\coordinate (center) at (0,0);
    	
    	 \begin{scope}[rotate around={120:(center)}]
	 	
    	 \draw [very thick,black!40!green](0,0) -- (8.5,0);
    	 \draw [very thick,black!40!green](-8.5,0) -- (0,0);
    	 
    	 \draw [very thick,black!20!blue](0,0) -- (4,2);
    	 \draw [very thick,black!20!blue](4,2) -- (8,-2);
    	 \draw [very thick,black!20!blue](0,0) -- (-4,2);
    	 \draw [very thick,black!20!blue](-4,2) -- (-8,-2);
    	 
    	 \draw [very thick,black!20!blue](0,0) -- (4,-2);
    	 \draw [very thick,black!20!blue](4,-2) -- (8,2);
    	 \draw [very thick,black!20!blue](0,0) -- (-4,-2);
    	 \draw [very thick,black!20!blue](-4,-2) -- (-8,2);
    	 
    	 \draw[very thick, black!20!blue, -latex]  (4,2) -- (2,1);
    	 \draw[very thick, black!20!blue, -latex]  (4,-2) -- (2,-1);
    	 \draw[very thick, black!20!blue, -latex]  (-4,2) -- (-2,1);
    	 \draw[very thick, black!20!blue, -latex]  (-4,-2) -- (-2,-1);
    	 
    	 \draw[very thick, black!20!blue, -latex]  (6,0) -- (4.8,1.2);
    	 \draw[very thick, black!20!blue, -latex]  (6,0) -- (4.8,-1.2);
    	 \draw[very thick, black!20!blue, -latex]  (-6,0) -- (-4.8,1.2);
    	 \draw[very thick, black!20!blue, -latex]  (-6,0) -- (-4.8,-1.2);
    	 
    	 \draw[very thick, black!40!green, -latex]  (0,0) -- (3,0);
    	 \draw[very thick, black!40!green, -latex]  (0,0) -- (-3,0);
    	 \draw[very thick, black!40!green, -latex]  (0,0) -- (7.3,0);
    	 \draw[very thick, black!40!green, -latex]  (0,0) -- (-7.3,0);
    	 
    	 \draw[very thick, black!20!blue, -latex]  (6,0) -- (7.2,1.2);
    	 \draw[very thick, black!20!blue, -latex]  (6,0) -- (7.2,-1.2);
    	 \draw[very thick, black!20!blue, -latex]  (-6,0) -- (-7.2,1.2);
    	 \draw[very thick, black!20!blue, -latex]  (-6,0) -- (-7.2,-1.2);

    	 \fill (0,0) circle (2.5pt);
    	 \fill (6,0) circle (2.5pt);
    	 \fill (-6,0) circle (2.5pt);

    	 \end{scope}
    		
    		\node[below] at (-4,5.2) {$\omega\lambda_0$};
    		\node[below] at (4,-4.3) {$-\omega\lambda_0$};
    		
    	\end{tikzpicture}
    	\caption{The jump contour $\Sigma^{(2)}$ in the complex $\lambda$-plane}
    	\label{figSigma2}
    \end{figure}
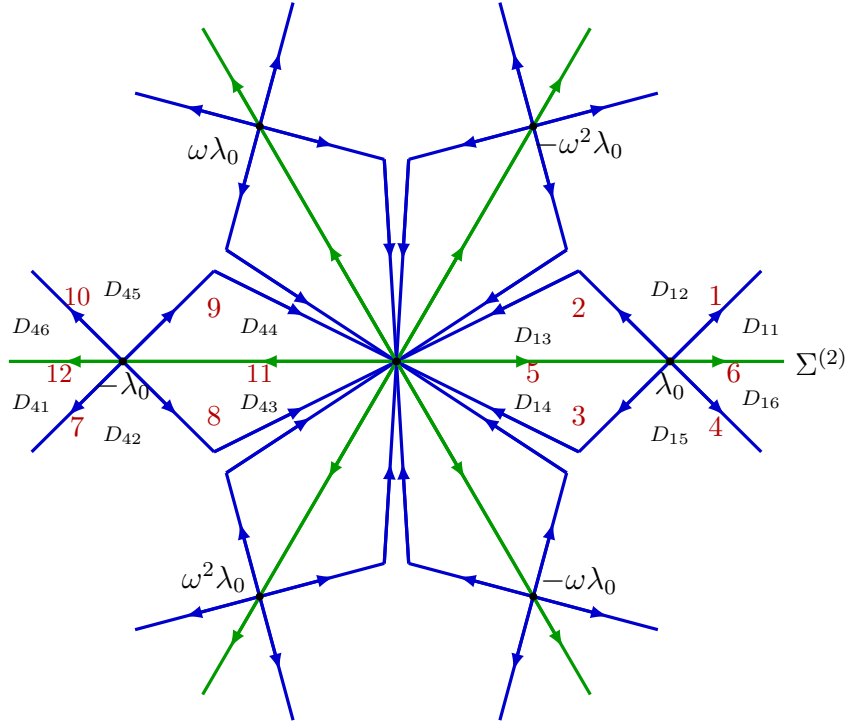

    Firstly, define the following matrix transformation at $\lambda_0$ as

    \begin{align}\label{38}
    	W_{\lambda_0}(x,t,\lambda)=\left\{\begin{aligned}
    		 &\begin{pmatrix}
    		 	1 & \frac{\delta_{N1}}{\delta^2_1}r_{1,a}(\lambda){\rm{e}}^{-\theta_{21}} & 0\\
    		 	0 & 1 & 0\\
    		 	0 & 0 & 1
    		 \end{pmatrix}, \qquad &&\lambda\in D_{11},\\
    		 &\begin{pmatrix}
    		 	1 & 0 & 0\\
    		 	-\frac{\delta^2_{1+}}{\delta_{N1}}\hat{r}^*_{1,a}(\lambda){\rm{e}}^{\theta_{21}} & 1 & 0\\
    		 	0 & 0 & 1
    		 \end{pmatrix}, \qquad &&\lambda\in D_{13},\\
    		 &\begin{pmatrix}
    		 	1 & -\frac{\delta_{N1}}{\delta^2_{1-}}\hat{r}_{1,a}(\lambda){\rm{e}}^{-\theta_{21}} & 0\\
    		 	0 & 1 & 0\\
    		 	0 & 0 & 1
    		 \end{pmatrix},\qquad &&\lambda\in D_{14},\\
    		 &\begin{pmatrix}
    		 	1 & 0 & 0\\
    		 	\frac{\delta^2_1}{\delta_{N1}}r^*_{1,a}(\lambda){\rm{e}}^{\theta_{21}} & 1 & 0\\
    		 	0 & 0 & 1
    		 \end{pmatrix}, \qquad &&\lambda\in D_{16},\\
    		 & \quad I, \qquad && \lambda\in D_{12}\cup D_{15}.
    	\end{aligned}
    	\right.
    \end{align}
    The transformation defined at $-\lambda_0$ is of the form
     \begin{equation}\label{39}
    	W_{-\lambda_0}(x,t,\lambda)=\left\{\begin{aligned}
    		&\begin{pmatrix}
    			1 & 0 & 0\\
    			-\frac{\delta_{N4}}{\delta^2_4}r_{2,a}(\lambda){\rm{e}}^{\theta_{21}} & 1 & 0\\
    			0 & 0 & 1
    		\end{pmatrix}, \qquad &&\lambda\in D_{41},\\
    		&\begin{pmatrix}
    			1 & \frac{\delta^2_{4+}}{\delta_{N4}}\hat{r}_{2,a}{\rm{e}}^{-\theta_{21}} & 0\\
    			0 & 1 & 0\\
    			0 & 0 & 1
    		\end{pmatrix}, \qquad &&\lambda\in D_{43},\\
    		&\begin{pmatrix}
    			1 & 0 & 0\\
    			\frac{\delta_{N4}}{\delta^2_{4-}}\hat{r}^*_{2,a}{\rm{e}}^{\theta_{21}} & 1 & 0\\
    			0 & 0 & 1
    		\end{pmatrix},\qquad &&\lambda\in D_{44},\\
    		&\begin{pmatrix}
    			1 & -\frac{\delta^2_4}{\delta_{N4}}r_{2,a}^*(\lambda){\rm{e}}^{-\theta_{21}} & 0\\
    			0 & 1 & 0\\
    			0 & 0 & 1
    		\end{pmatrix}, \qquad &&\lambda\in D_{46},\\
    		 & \quad I, \qquad && \lambda\in D_{42}\cup D_{45}.
    	\end{aligned}
    	\right.
    \end{equation}
     In addition, the definitions of $W(x,t,\lambda)$ at several other critical points satisfy the relations:
     \begin{equation*}
  	   \begin{aligned}
  		&W_{\omega\lambda_0}(x,t,\lambda)=\tilde{\sigma}_1^{-1}W_{\lambda_0}(x,t,\omega^2\lambda)\tilde{\sigma}_1,\\
  		&W_{\omega^2\lambda_0}(x,t,\lambda)=\tilde{\sigma}_1^{-1}W_{\lambda_0}(x,t,\omega\lambda)\tilde{\sigma}_1,\\
  		&W_{-\omega\lambda_0}(x,t,\lambda)=\tilde{\sigma}_1^{-1}W_{-\lambda_0}(x,t,\omega^2\lambda)\tilde{\sigma}_1,\\
  		&W_{-\omega^2\lambda_0}(x,t,\lambda)=\tilde{\sigma}_1^{-1}W_{-\lambda_0}(x,t,\omega\lambda)\tilde{\sigma}_1.\\
  	   \end{aligned}
     \end{equation*}
     From the properties satisfied by the reflection coefficients $r_1(\lambda)$ and $r_2(\lambda)$ in Theorem \ref{theo21}, it follows that $W \rightarrow I$ as $\lambda\rightarrow0$ and $\lambda\rightarrow\infty$.

     After defining the matrix transformation $W(x,t,\lambda)$ in \eqref{38} and \eqref{39}, the eigenfunction $M^{(1)}$ is operated by
      \begin{equation}\label{40}
     	M^{(2)}(x,t,\lambda)=M^{(1)}(x,t,\lambda)W(x,t,\lambda).
     \end{equation}
     \par
    Under this transformation, the RH problem for the eigenfunction $M^{(2)}(x,t,\lambda)$ is gotten immediately.

    \begin{rhp}\label{rhp2}
    	Find a $3\times3$ matrix-valued function $M^{(2)}(x,t,\lambda)$ with the following properties:
    	\begin{itemize}
    		\item $M^{(2)}(x,t,\lambda)$ is holomorphic in $\mathbb{C}\setminus\Sigma^{(2)}$, where 	$\Sigma^{(2)}$ is  shown in Fig. \ref{figSigma2}.
    		\item  $M^{(2)}(x,t,\lambda)$ satisfies the following jump relationship
    		\begin{equation*}
    			M^{(2)}_+(x,t,\lambda)=M^{(2)}_-(x,t,\lambda)v_j^{(2)}(x,t,\lambda),\qquad \lambda
    			\in \Sigma_j^{(2)},
    		\end{equation*}
    		where only the forms $v_j^{(2)}$ $(j=1,\cdots,12)$ are given
    		\begin{equation}\label{41}\begin{aligned}
    			&v^{(2)}_1 =\begin{pmatrix}
    				1 & -\frac{\delta_{N1}}{\delta^2_1}r_{1,a}(\lambda){\rm{e}}^{-\theta_{21}} & 0\\
    				0 & 1 & 0\\
    				0 & 0 & 1
    			\end{pmatrix},\quad
    			&& v^{(2)}_2=\begin{pmatrix}
    				1 & 0 & 0\\
    				-\frac{\delta^2_{1+}}{\delta_{N1}}\hat{r}^*_{1,a}(\lambda){\rm{e}}^{\theta_{21}} & 1 & 0\\
    				0 & 0 & 1
    			\end{pmatrix},\\
    			&v^{(2)}_3=\begin{pmatrix}
    				1 & \frac{\delta_{N1}}{\delta^2_{1-}}\hat{r}_{1,a}(\lambda){\rm{e}}^{-\theta_{21}} & 0\\
    				0 & 1 & 0\\
    				0 & 0 & 1
    			\end{pmatrix},\quad
    			&&v^{(2)}_4=\begin{pmatrix}
    				1 & 0 & 0\\
    				\frac{\delta^2_1}{\delta_{N1}}r^*_{1,a}(\lambda){\rm{e}}^{\theta_{21}} & 1 & 0\\
    				0 & 0 & 1
    			\end{pmatrix},\\
    			&v^{(2)}_5=\begin{pmatrix}
    				1-r_{1,r}r^*_{1,r} & -\frac{\delta_{N1}}{\delta^2_{1-}}\hat{r}_{1,r}{\rm{e}}^{-\theta_{21}} & 0\\
    				\frac{\delta^2_{1+}}{\delta_{N1}}\hat{r}^*_{1,r}{\rm{e}}^{\theta_{21}} & 1 & 0\\
    				0 & 0 & 1
    			\end{pmatrix},
    			&& v^{(2)}_6=\begin{pmatrix}
    				1 & -\frac{\delta_{N1}}{\delta^2_1}r_{1,r}{\rm{e}}^{-\theta_{21}} & 0\\
    				\frac{\delta^2_1}{\delta_{N1}}r^*_{1,r}{\rm{e}}^{\theta_{21}} & 1-r_{1,r}r^*_{1,r} & 0\\
    				0 & 0 & 1
    			\end{pmatrix},\\
    			&v^{(2)}_7 =\begin{pmatrix}
    				1 & 0 & 0\\
    				\frac{\delta_{N4}}{\delta^2_4}r_{2,a}(\lambda){\rm{e}}^{\theta_{21}} & 1 & 0\\
    				0 & 0 & 1
    			\end{pmatrix},\quad
    			&& v^{(2)}_8=\begin{pmatrix}
    				1 & \frac{\delta^2_{4+}}{\delta_{N4}}\hat{r}_{2,a}{\rm{e}}^{-\theta_{21}} & 0\\
    				0 & 1 & 0\\
    				0 & 0 & 1
    			\end{pmatrix},\\
    			&v^{(2)}_9=\begin{pmatrix}
    				1 & 0 & 0\\
    				-\frac{\delta_{N4}}{\delta^2_{4-}}\hat{r}^*_{2,a}{\rm{e}}^{\theta_{21}} & 1 & 0\\
    				0 & 0 & 1
    			\end{pmatrix},\quad
    			&&v^{(2)}_{10}=\begin{pmatrix}
    				1 & -\frac{\delta^2_4}{\delta_{N4}}r_{2,a}^*(\lambda){\rm{e}}^{-\theta_{21}} & 0\\
    				0 & 1 & 0\\
    				0 & 0 & 1
    			\end{pmatrix},\\
    			&v^{(2)}_{11}=\begin{pmatrix}
    				1 & -\frac{\delta^2_{4+}}{\delta_{N4}}\hat{r}_{2,r}{\rm{e}}^{-\theta_{21}} & 0\\
    				\frac{\delta_{N4}}{\delta^2_{4-}}\hat{r}^*_{2,r}{\rm{e}}^{\theta_{21}} & 1-r_{2,r}r^*_{2,r} & 0\\
    				0 & 0 & 1
    			\end{pmatrix},
    			&& v^{(2)}_{12}=\begin{pmatrix}
    				1-r_{2,r}r^*_{2,r} & -\frac{\delta_4^2}{\delta_{N4}}r_{2,r}^*{\rm{e}}^{-\theta_{21}} & 0\\
    				\frac{\delta_{N4}}{\delta_4^2}r_{2,r}{\rm{e}}^{\theta_{21}} & 1 & 0\\
    				0 & 0 & 1
    			\end{pmatrix},\\
    		\end{aligned}\end{equation}
    		and the rest of the jump matrices can be obtained according to the symmetry  $v_{j+12}^{(2)}(\lambda)=\tilde{\sigma}_1^{-1}v_j^{(2)}(\omega^2\lambda)\tilde{\sigma}_1$ for $j=1,2,\cdots,24$.
    		\item  As $\lambda\rightarrow\infty$, $M^{(1)}(x,t,\lambda)=I+\mathcal{O}(\frac{1}{\lambda})$ for $\lambda\in\mathbb{C}\setminus\Sigma^{(2)}$.
    		\item  As $\lambda\rightarrow0$, $M^{(2)}(x,t,\lambda)=G(x,t)+\mathcal{O}({\lambda})$ for $\lambda\in\mathbb{C}\setminus\Sigma^{(2)}$.
    		\item  $M^{(2)}(x,t,\lambda)=\tilde{\sigma}_1^{-1}M^{(2)}(x,t,\omega\lambda)\tilde{\sigma}_1=\tilde{\sigma}_2\overline{M^{(2)}(x,t,\overline{\lambda})}\tilde{\sigma}_2^{-1}$.
    	\end{itemize}
    \end{rhp}

    \begin{lem}
    	The jump matrices in \eqref{41} of RH problem \ref{rhp2} satisfies the following properties:
    	\begin{enumerate}
    		\item The jump matrices $v_j^{(2)}\rightarrow I$ $(j=1, 2,\cdots,36)$ uniformly as $t\rightarrow\infty$ for $\lambda\in \Sigma^{(2)}\setminus\{\pm \omega^j\lambda_0,j=0,1,2\}$.
    		\item For the index $j$ that $\Sigma_j^{(2)}\in \mathbb{R}\cup\omega\mathbb{R}\cup\omega^2\mathbb{R}$, it follows that
    		\begin{equation*}
    			\|(1+\vert\cdot\vert)(v^{(2)}_j-I)\|_{L^1\cap L^\infty(\Sigma^{(2)}_j)}\leq \frac{C}{t^{N+\frac{3}{2}}}.
    		\end{equation*}
    	\end{enumerate}
    \end{lem}

   The proof can be completed according to the properties of $\delta$ function and Lemma \ref{lem43}, which will not take up space to give a detailed proof here.

   \subsubsection{Local parametrix at the critical points $\pm\omega^j\lambda_0$}\label{423}
   \ \ \ \
   The RH problem \ref{rhp2}, constructed from the previous subsection, corresponds to the jump matrices $v^{(2)}_j(x,t,\lambda)$ that asymptotically approaches the identity matrix $I$ as $t\rightarrow\infty$ except for the critical points $\pm\omega^j\lambda_0$ $(j=0,1,2)$ on the jump contours $\Sigma^{(2)}$. Next, we deal two critical points $\pm\lambda_0$ on the real axis, and establish the relationship between the local RH problem at the two critical points and the parabolic cylinder functions. The other four critical points can be dealt with correspondingly by using the symmetries.
   \par
   For any $\epsilon>0$, consider two small discs $B_{\epsilon}(\pm\lambda_0):=\{\lambda\in\mathbb{C}|\vert\lambda\pm\lambda_0\vert\leq \epsilon \}$. The problem with this subsection lies in the interior of the light cone $\vert\frac{x}{t}\vert< 1$, that is, Region {\rm{IV}}. Thus, there exists a fixed constant $M$ such that $0<\lambda_0<M$. Define two scaling transformations at $\pm\lambda_0$ by
   \begin{equation}\label{42}
   	\begin{aligned}
   		&z_1=(\lambda-\lambda_0)\sqrt{\frac{2\sqrt{3}t}{\lambda_0(1+\lambda_0^2)}},\\
   		&z_2=(\lambda+\lambda_0)\sqrt{\frac{2\sqrt{3}t}{\lambda_0(1+\lambda_0^2)}}.\\
   	\end{aligned}
   \end{equation}
   \par
   Then the phase function $\theta_{21}(\lambda)$ involved in the jump matrix is expanded as follows at $\lambda_0$ and $-\lambda_0$, respectively
   \begin{equation*}
   	\begin{aligned}	&\theta_{21}(\lambda)=\theta_{21}(\lambda_0)-\frac{iz_1^2}{2}\left(1-\frac{\lambda_0^3z_1}{a\eta_1^4}\right),\\
   		&\theta_{21}(\lambda)=\theta_{21}(-\lambda_0)+\frac{iz_2^2}{2}\left( 1+\frac{\lambda_0^3z_2}{a\eta_2^4}\right) ,\\
   	\end{aligned}
   \end{equation*}
   where $a=\sqrt{\frac{\lambda_0(1+\lambda_0^2)}{2\sqrt{3}t}}=(bt)^{-1/2}$, $b=\frac{2 \sqrt{3}}{\lambda_0(1+\lambda_0^2)}$, $\eta_1=\lambda_0+ak_1z_1$ and $\eta_2=-\lambda_0+ak_2z_2$, $0<k_1,k_2<1$.

    Recall the definition of $\delta_j(\lambda)$ $(j=1,4)$ in \eqref{33} and \eqref{34}, which can be written as
    \begin{equation*}
    	\delta_1(\lambda)={\rm{e}}^{\chi_1(\lambda)}={\rm{e}}^{i\nu_1(\lambda_0)\log_{-\pi}(\lambda-\lambda_0)+\chi_1'(\lambda)},
    \end{equation*}
     \begin{equation*}
    	\delta_4(\lambda)={\rm{e}}^{\chi_4(\lambda)}={\rm{e}}^{i\nu_4(-\lambda_0)\log_0(\lambda+\lambda_0)+\chi_4'(\lambda)}.
    \end{equation*}
    Take the element $\frac{\delta^2_{1+}}{\delta_{N1}}$ involved in matrix $v_2^{(2)}$ as an example, it is expressed as
    \begin{equation*}
    		    \frac{\delta^2_{1+}(\lambda)}{\delta_{N1}(\lambda)}={\rm{e}}^{2i\nu_1(\lambda_0)\log_{-\pi}(z_1)}\frac{a^{2i\nu_1(\lambda_0)}{\rm{e}}^{-2\chi_1'(\lambda_0)}}{\delta_{N1}(\lambda_0)}\frac{{\rm{e}}^{2\chi_1'(\lambda_0)-2\chi_1'(\lambda)}\delta_{N1}(\lambda_0)}{\delta_{N1}(\lambda)}.
    \end{equation*}
    Define $\delta_{\lambda_0}^0(\lambda_0)=\frac{a^{2i\nu_1(\lambda_0)}{\rm{e}}^{-2\chi_1'(\lambda_0)}}{\delta_{N1}(\lambda_0)}$ and $\delta_{\lambda_0}^1(\lambda)=\frac{{\rm{e}}^{2\chi_1'(\lambda_0)-2\chi_1'(\lambda)}\delta_{N1}(\lambda_0)}{\delta_{N1}(\lambda)}$.
     Since $\delta_{N1}(\lambda)$ is not the function of $\delta_1(\lambda)$, it is holomorphic at $\lambda_0$ and can be performed as a Taylor expansion at $\lambda_0$. Through sorting out the expansion, it is derived that
    \begin{equation}\label{43}
    	\begin{aligned}
          \delta_{\lambda_0}^1(\lambda) &={\rm{exp}}\{[\alpha_{11}(\lambda_0)\log_{-\pi}(\lambda\!-\!\lambda_0)\!+\!\alpha_{12}(\lambda_0)](\lambda\!-\!\lambda_0)+[\alpha_{21}(\lambda_0)\log_{-\pi}(\lambda\!-\!\lambda_0)\\
          &\quad\!+\!\alpha_{22}(\lambda_0)](\lambda\!-\!\lambda_0)^2+\cdots+[\alpha_{N1}(\lambda_0)\log_{-\pi}(\lambda\!-\!\lambda_0)\!+\!\alpha_{N2}(\lambda_0)](\lambda\!-\!\lambda_0)^N\\
          &\quad\!+\!\mathcal{O}[(\lambda\!-\!\lambda_0)^{N+1}\log_{-\pi}(\lambda\!-\!\lambda_0)]\},
    	\end{aligned}
    \end{equation}
    where $\alpha_{ij}$ $(i=1, 2,\cdots,N,j=1,2)$ are smooth functions of $\lambda_0$ and are independent of $\lambda$.

    Similarly, the element $\frac{\delta_{N4}}{\delta^2_{4}}$ involved in matrix $v_7^{(2)}$ can be expressed as
    \begin{equation*}
    	\begin{aligned}		    \frac{\delta_{N4}(\lambda)}{\delta^2_{4}(\lambda)}&={\rm{e}}^{-2i\nu_4(-\lambda_0)\log_0(z_2)}
    \frac{\delta_{N4}(-\lambda_0)}{a^{2i\nu_4(-\lambda_0)}{\rm{e}}^{-2\chi_4'(-\lambda_0)}}
    \frac{\delta_{N4}(\lambda)}{{\rm{e}}^{2\chi_4'(-\lambda_0)-2\chi_4'(\lambda)}\delta_{N4}(-\lambda_0)}.\\
    	\end{aligned}
    \end{equation*}
      Further denote $\delta_{-\lambda_0}^0(-\lambda_0)=\frac{a^{2i\nu_4(-\lambda_0)}{\rm{e}}^{-2\chi_4'(-\lambda_0)}}{\delta_{N4}(-\lambda_0)}$, $\delta_{-\lambda_0}^1(\lambda)=\frac{{\rm{e}}^{2\chi_4'(-\lambda_0)-2\chi_4'(\lambda)}
      \delta_{N4}(-\lambda_0)}{\delta_{N4}(\lambda)}$, which can be expanded into the following form
       \begin{equation}\label{44}
       	\begin{aligned}
       	\delta_{-\lambda_0}^1(\lambda)
       	&={\rm{exp}}\{[\beta_{11}(\!-\!\lambda_0)\log_0(\lambda+ \lambda_0)+\beta_{12}(\!-\!\lambda_0)](\lambda\!+\!\lambda_0)+[\beta_{21}(\!-\!\lambda_0)\log_0(\lambda\!+\! \lambda_0)\\
       		&\!+\!\beta_{22}(\!-\!\lambda_0)](\lambda\!+\!\lambda_0)^2
       		+\cdots+[\beta_{N1}(\!-\!\lambda_0)\log_0(\lambda\!+\!\lambda_0)+\beta_{N2}(\!-\!\lambda_0)](\lambda\!+\!\lambda_0)^N\\
       		&+\mathcal{O}[(\lambda\!+\!\lambda_0)^{N+1}\log_0(\lambda\!+\!\lambda_0)]\},
       	\end{aligned}
       \end{equation}
    where $\beta_{ij}$ $(i=1, 2,\cdots,N,j=1,2)$ are smooth functions of $-\lambda_0$ and are independent of $\lambda$.

    Then introduce the diagonal matrix $H$ at $\pm\lambda_0$ as
    \begin{equation}\label{45}
    	H(\pm\lambda_0,t)=\begin{pmatrix}
     [\delta_{\pm\lambda_0}^0(\lambda)]^{\mp\frac{1}{2}}\,{\rm{e}}^{-\frac{1}{2}\theta_{21}(\pm\lambda_0)} & 0 & 0\\
    		0 &  [\delta_{\pm\lambda_0}^0(\lambda)]^{\pm\frac{1}{2}}\,{\rm{e}}^{\frac{1}{2}\theta_{21}(\pm\lambda_0)} & 0\\
    		0 & 0 & 1
    	\end{pmatrix},
    \end{equation}
    which satisfy the following properties.

    \begin{lem}\label{lem45}
    	For $t\geq1$ and $0<\lambda_0<M$, fix $z_1$ and $z_2$. The matrix functions $H(\pm\lambda_0,t)$ defined by equation \eqref{45} satisfy
    	\begin{equation*}
    		\sup_{t\geq1}\vert H(\pm\lambda_0,t)\vert\leq C.
    	\end{equation*}
     The functions $\delta_{\pm\lambda_0}^0(\pm\lambda_0)$ satisfy
    	\begin{equation*}
    		\vert\delta_{\lambda_0}^0(\lambda_0)\vert=1, \quad
    		\vert\delta_{-\lambda_0}^0(-\lambda_0)\vert={\rm{e}}^{-2\pi \nu_4(-\lambda_0)}.
    	\end{equation*}
    	Moreover, the function $\frac{\lambda_0^3z_j}{a\eta_j^4}\rightarrow0$ for $j=1,2$.
    \end{lem}

    \begin{proof}
    	From $a\in\mathbb{R}$ and $a>0$, by calculation it follows that
    	\begin{equation*}
    		\vert a^{2i\nu_1(\lambda_0)}\vert=\vert {\rm{e}}^{2 i \nu_1(\lambda_0)\ln a}\vert=1.
    	\end{equation*}
    	According to the definitions of \(\delta_{N1} = \frac{\delta_3 \delta_5 \delta_4^2}{\delta_2 \delta_6}\) and \(\delta_{N4} = \frac{\delta_1^2 \delta_2 \delta_6}{\delta_3 \delta_5}\), and by leveraging the conjugate symmetries of the \(\delta\) functions as stated in Proposition \ref{prop41}, along with the symmetric relationships \eqref{35} among $\delta_j$, we deduce that \(|\delta_{N1}| = 1\) and \(|\delta_{N4}| = 1\). Additionally, one has
    	\begin{equation*}
    		{\rm{Re}}\,\chi_1'(\lambda_0)=0, \quad {\rm{Re}}\,\chi_4'(-\lambda_0)=-\pi\nu_4(-\lambda_0),
    	\end{equation*}
    	so we can get
    	\begin{equation*}
    		\vert\delta_{\lambda_0}^0(\lambda_0)\vert=1, \quad
    		\vert\delta_{-\lambda_0}^0(-\lambda_0)\vert={\rm{e}}^{-2\pi \nu_4(-\lambda_0)}.
    	\end{equation*}

     Finally, the boundedness of $H(\pm\lambda_0,t)$ as $t\geq1$ can also be proven from the above equation, as well as the values taken by the phase function $\theta_{21}(\lambda)$ at $\pm\lambda_0$.

    \end{proof}

    Consider the following transformation of the eigenfunction $M^{(2)}$ with respect to $H(\pm\lambda_0,t)$:
    \begin{equation*}
    	M^{(\epsilon)}_{\pm\lambda_0}(x,t,\lambda)=M^{(2)}(x,t,\lambda)H(\pm\lambda_0,t),\qquad \lambda\in B_{\epsilon}(\pm\lambda_0),
    \end{equation*}
     and define $\Sigma^{\lambda_0}_j=\Sigma^{(2)}_j\cap B_{\lambda_0}(\epsilon)$ $(j=1,2,3,4)$, $\Sigma^{-\lambda_0}_j=\Sigma^{(2)}_j\cap B_{-\lambda_0}(\epsilon)$ $(j=7,8,9,10)$ as shown in Fig. \ref{figSigmalam}.
     After integrating the scale transformation \eqref{42}, we obtain the jump relations associated with the local RH problem at $\lambda_0$, and based on the expansion of $\delta_{\lambda_0}^1$ in \eqref{43}, the expansion forms of the jump matrix $v^{(\epsilon)}$ at $\Sigma^{\lambda_0}$ are derived:
      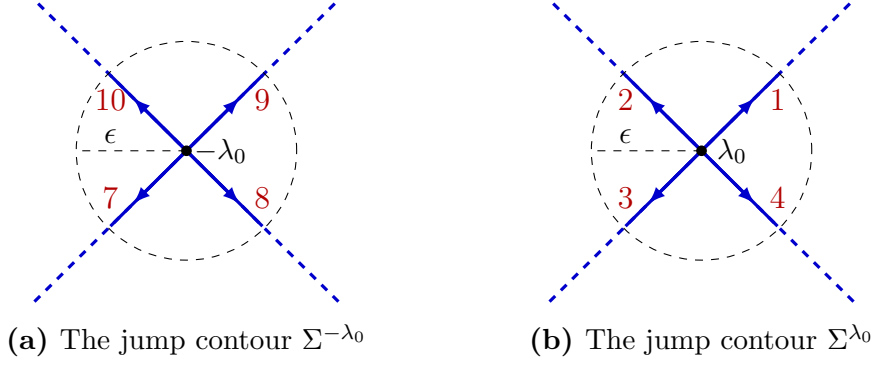
\begin{figure}[htbp]
     	\centering
     	\begin{subfigure}[b]{0.4\textwidth}
     		\centering
     	\begin{tikzpicture}
     		\draw[dashed,very thick,black!20!blue] (-2,-2) -- (2,2);
     		\draw[dashed,very thick,black!20!blue] (2,-2) -- (-2,2);
     		
     		\draw[dashed] (0,0) -- (-1.45,0);
     		
     		\draw [very thick,black!20!blue](-1,-1) -- (1,1);
     		\draw [very thick,black!20!blue](-1,1) -- (1,-1);
     		
     		\draw[very thick, black!20!blue, -latex]  (0,0) -- (0.7,0.7);
     		\draw[very thick, black!20!blue, -latex]  (0,0) -- (-0.7,0.7);
     		\draw[very thick, black!20!blue, -latex]  (0,0) -- (0.7,-0.7);
     		\draw[very thick, black!20!blue, -latex]  (0,0) -- (-0.7,-0.7);
     		
     		\draw [ dashed] (0,0) circle (1.45);
     		\fill (0,0) circle (2pt);
     		
     		\node at (0.45,0) {\small$-\lambda_0$};
     		\node at (-1,0.2) {$\epsilon$};
     		
     		\node[red!70!black,above] at (1,0.4) {$9$};
     		\node[red!70!black,above] at (-1,0.4) {$10$};
     		\node[red!70!black,above] at (-1,-0.9) {$7$};
     		\node[red!70!black,above] at (1,-0.9) {$8$};
     		
     	\end{tikzpicture}
     		\caption{The jump contour $\Sigma^{-\lambda_0}$}
     	\end{subfigure}
     	\begin{subfigure}[b]{0.4\textwidth}
     		\centering
     		\begin{tikzpicture}
     			\draw[dashed,very thick,black!20!blue] (-2,-2) -- (2,2);
     			\draw[dashed,very thick,black!20!blue] (2,-2) -- (-2,2);
     			
     			\draw[dashed] (0,0) -- (-1.45,0);
     			
     			\draw [very thick,black!20!blue](-1,-1) -- (1,1);
     			\draw [very thick,black!20!blue](-1,1) -- (1,-1);
     			
     			\draw[very thick, black!20!blue, -latex]  (0,0) -- (0.7,0.7);
     			\draw[very thick, black!20!blue, -latex]  (0,0) -- (-0.7,0.7);
     			\draw[very thick, black!20!blue, -latex]  (0,0) -- (0.7,-0.7);
     			\draw[very thick, black!20!blue, -latex]  (0,0) -- (-0.7,-0.7);
     			
     			\draw [ dashed] (0,0) circle (1.45);
     			\fill (0,0) circle (2pt);
     			
     			\node at (0.4,0) {\small$\lambda_0$};
     			\node at (-1,0.2) {$\epsilon$};
     			
     			\node[red!70!black,above] at (1,0.4) {$1$};
     			\node[red!70!black,above] at (-1,0.4) {$2$};
     			\node[red!70!black,above] at (-1,-0.9) {$3$};
     			\node[red!70!black,above] at (1,-0.9) {$4$};
     			
     		\end{tikzpicture}
     		\caption{The jump contour $\Sigma^{\lambda_0}$}
     	\end{subfigure}
     	\caption{The jump contours $\Sigma^{\pm\lambda_0}$ in the complex $\lambda$-plane}
     	\label{figSigmalam}
     \end{figure}
     \begin{align}\label{vep}
     			v^{(\epsilon)}_1(z_1) &=\begin{pmatrix}
     				1 & -{\rm{e}}^{-2i\nu_1(\lambda_0)\log_{-\pi}(z_1)}[\delta_{\lambda_0}^1(z_1)]^{-1}r_{1,a}(z_1){\rm{e}}^{\frac{iz_1^2}{2}(1-\frac{\lambda_0^3z_1}{a\eta_1^4})} & 0\\
     				0 & 1 & 0\\
     				0 & 0 & 1
     			\end{pmatrix}\nonumber\\
     			&=I  -{\rm{e}}^{-2i\nu_1(\lambda_0)\log_{-\pi}(z_1)+\frac{iz_1^2}{2}\left(1-\frac{\lambda_0^3z_1}{a\eta_1^4}\right)}\left( \sum_{0\leq q\leq p\leq N}\frac{(\ln t)^q}{(bt)^{\frac{p}{2}}} v_{1,pq}^{(\epsilon)}\right)+E_{v_1}^{(\epsilon)}(z_1,t,\lambda_0) ,\nonumber\\
     			v^{(\epsilon)}_2(z_1)&=\begin{pmatrix}
     				1 & 0 & 0\\
     				-{\rm{e}}^{2i\nu_1(\lambda_0)\log_{-\pi}(z_1)}\delta_{\lambda_0}^1(z_1)\hat{r}^*_{1,a}(z_1){\rm{e}}^{-\frac{iz_1^2}{2}(1-\frac{\lambda_0^3z_1}{a\eta_1^4})} & 1 & 0\\
     				0 & 0 & 1
     			\end{pmatrix}\nonumber\\
     			&=I  -{\rm{e}}^{2i\nu_1(\lambda_0)\log_{-\pi}(z_1)-\frac{iz_1^2}{2}\left(1-\frac{\lambda_0^3z_1}{a\eta_1^4}\right)}\left( \sum_{0\leq q\leq p\leq N}\frac{(\ln t)^q}{(bt)^{\frac{p}{2}}} v_{2,pq}^{(\epsilon)}\right)+E_{v_2}^{(\epsilon)}(z_1,t,\lambda_0),\\
     			v^{(\epsilon)}_3(z_1)&=\begin{pmatrix}
     				1 & {\rm{e}}^{-2i\nu_1(\lambda_0)\log_{-\pi}(z_1)}[\delta_{\lambda_0}^1(z_1)]^{-1}\hat{r}_{1,a}(z_1){\rm{e}}^{\frac{iz_1^2}{2}(1-\frac{\lambda_0^3z_1}{a\eta_1^4})} & 0\\
     				0 & 1 & 0\\
     				0 & 0 & 1
     			\end{pmatrix}\nonumber\\
     			&=I + {\rm{e}}^{-2i\nu_1(\lambda_0)\log_{-\pi}(z_1)+\frac{iz_1^2}{2}\left(1-\frac{\lambda_0^3z_1}{a\eta_1^4}\right)}\left( \sum_{0\leq q\leq p\leq N}\frac{(\ln t)^q}{(bt)^{\frac{p}{2}}} v_{3,pq}^{(\epsilon)}\right)+E_{v_3}^{(\epsilon)}(z_1,t,\lambda_0),\nonumber\\
     			v^{(\epsilon)}_4(z_1)&=\begin{pmatrix}
     				1 & 0 & 0\\
     				{\rm{e}}^{2i\nu_1(\lambda_0)\log_{-\pi}(z_1)}\delta_{\lambda_0}^1(z_1)r^*_{1,a}(z_1){\rm{e}}^{-\frac{iz_1^2}{2}(1-\frac{\lambda_0^3z_1}{a\eta_1^4})} & 1 & 0\\
     				0 & 0 & 1
     			\end{pmatrix}\nonumber\\
     			&=I + 	{\rm{e}}^{2i\nu_1(\lambda_0)\log_{-\pi}(z_1)-\frac{iz_1^2}{2}\left(1-\frac{\lambda_0^3z_1}{a\eta_1^4}\right)}\left( \sum_{0\leq q\leq p\leq N}\frac{(\ln t)^q}{(bt)^{\frac{p}{2}}} v_{4,pq}^{(\epsilon)}\right)+E_{v_4}^{(\epsilon)}(z_1,t,\lambda_0),\nonumber
     	\end{align}
     where
     \begin{equation}\label{46}
     	\|E_{v_j}^{(\epsilon)}(\cdot,t,\lambda_0)\|_{L^1\cap L^\infty}=\mathcal{O}\left( \frac{(\ln t)^{N+1}}{t^{(N+1)/2}}\right) ,\quad j=1,\cdots,4.
     \end{equation}
     The elements in matrices $v_{1,pq}^{(\epsilon)}(z_1,\lambda_0)$ and $v_{3,pq}^{(\epsilon)}(z_1,\lambda_0)$ are zero except for the $(1,2)$ positions, and the elements in matrices $v_{2,pq}^{(\epsilon)}(z_1,\lambda_0)$ and $v_{4,pq}^{(\epsilon)}(z_1,\lambda_0)$ are zero except for the $(2,1)$ positions. The elements at those positions involve only $\alpha_{ij}(\lambda_0)$, the relevant reflection coefficient $r_{1,a}$ and $\hat{r}_{1,a}$, as well as their conjugates.

      Similarly, by utilizing equation \eqref{44},
      the jump relations corresponding to the local RH problem at $-\lambda_0$ are as follows

      \begin{align}\label{vep1}
      		v^{(\epsilon)}_7(z_2)&= \begin{pmatrix}
      			1 & 0 & 0\\
      			{\rm{e}}^{-2i\nu_4(-\lambda_0)\log_0(z_2)}[\delta_{-\lambda_0}^1(z_2)]^{-1}r_{2,a}(z_2){\rm{e}}^{\frac{iz_2^2}{2}(1+\frac{\lambda_0^3z_2}{a\eta_2^4})} & 1 & 0\\
      			0 & 0 & 1
      		\end{pmatrix}\nonumber\\
      		&=I + {\rm{e}}^{-2i\nu_4(-\lambda_0)\log_0(z_2)+\frac{iz_2^2}{2}(1+\frac{\lambda_0^3z_2}{a\eta_2^4})}\left( \sum_{0\leq q\leq p\leq N}\frac{(\ln t)^q}{(bt)^{\frac{p}{2}}} v_{7,pq}^{(\epsilon)}\right)
      		+E_{v_7}^{(\epsilon)}(z_2,t,\lambda_0),\nonumber \\
      		v^{(\epsilon)}_8(z_2)&=\begin{pmatrix}
      			1 & {\rm{e}}^{2i\nu_4(-\lambda_0)\log_0(z_2)}\delta_{-\lambda_0}^1(z_2)\hat{r}_{2,a}(z_2){\rm{e}}^{-\frac{iz_2^2}{2}(1+\frac{\lambda_0^3z_2}{a\eta_2^4})} & 0\\
      			0 & 1 & 0\\
      			0 & 0 & 1
      		\end{pmatrix}\nonumber\\
      		&=I + {\rm{e}}^{2i\nu_4(-\lambda_0)\log_0(z_2)-\frac{iz_2^2}{2}(1+\frac{\lambda_0^3z_2}{a\eta_2^4})}\left( \sum_{0\leq q\leq p\leq N}\frac{(\ln t)^q}{(bt)^{\frac{p}{2}}} v_{8,pq}^{(\epsilon)}\right)
      		+E_{v_8}^{(\epsilon)}(z_2,t,\lambda_0), \\
      		v^{(\epsilon)}_9(z_2)&=\begin{pmatrix}
      			1 & 0 & 0\\
      			-{\rm{e}}^{-2i\nu_4(-\lambda_0)\log_0(z_2)}[\delta_{-\lambda_0}^1(z_2)]^{-1}\hat{r}^*_{2,a}(z_2){\rm{e}}^{\frac{iz_2^2}{2}(1+\frac{\lambda_0^3z_2}{a\eta_2^4})} & 1 & 0\\
      			0 & 0 & 1
      		\end{pmatrix}\nonumber\\
      		&=I 	-{\rm{e}}^{-2i\nu_4(-\lambda_0)\log_0(z_2)+\frac{iz_2^2}{2}(1+\frac{\lambda_0^3z_2}{a\eta_2^4})}\left( \sum_{0\leq q\leq p\leq N}\frac{(\ln t)^q}{(bt)^{\frac{p}{2}}} v_{9,pq}^{(\epsilon)}\right)
      		+E_{v_9}^{(\epsilon)}(z_2,t,\lambda_0),\nonumber\\
      		v^{(\epsilon)}_{10}(z_2)&=\begin{pmatrix}
      			1 & -	{\rm{e}}^{2i\nu_4(-\lambda_0)\log_0(z_2)}\delta_{-\lambda_0}^1(z_2)r^*_{2,a}(z_2){\rm{e}}^{-\frac{iz_2^2}{2}(1+\frac{\lambda_0^3z_2}{a\eta_2^4})} & 0\\
      			0 & 1 & 0\\
      			0 & 0 & 1
      		\end{pmatrix}\nonumber\\
      		&=I -	{\rm{e}}^{2i\nu_4(-\lambda_0)\log_0(z_2)-\frac{iz_2^2}{2}(1+\frac{\lambda_0^3z_2}{a\eta_2^4})}\left( \sum_{0\leq q\leq p\leq N}\frac{(\ln t)^q}{(bt)^{\frac{p}{2}}} v_{10,pq}^{(\epsilon)}\right)
      		+E_{v_{10}}^{(\epsilon)}(z_2,t,\lambda_0),\nonumber
      	\end{align}
       where
       \begin{equation}\label{47}
       	\|E_{v_j}^{(\epsilon)}(\cdot,t,\lambda_0)\|_{L^1\cap L^\infty}=\mathcal{O}\left( \frac{(\ln t)^{N+1}}{t^{(N+1)/2}}\right) ,\quad j=7, 8,\cdots,10.
       \end{equation}
     The elements in matrices $v_{7,pq}^{(\epsilon)}(z_1,\lambda_0)$ and $v_{9,pq}^{(\epsilon)}(z_1,\lambda_0)$ are zero except for the $(2,1)$ positions, and the elements in matrices $v_{8,pq}^{(\epsilon)}(z_1,\lambda_0)$ and $v_{10,pq}^{(\epsilon)}(z_1,\lambda_0)$ are zero except for the $(1,2)$ positions. The elements at those positions involve only $\beta_{ij}(\lambda_0)$, the relevant reflection coefficient $r_{2,a}$ and $\hat{r}_{2,a}$, as well as their conjugates.

     More importantly, the aforementioned expansions also meet the following conditions to ensure the rationality of the expansions:
      \begin{equation}\label{48}
      	\|{\rm{e}}^{\pm2i\nu_1(\lambda_0)\log_{-\pi}(\cdot)\mp\frac{i(\cdot)^2}{2}\left(1-\frac{\lambda_0^3(\cdot)}{a\eta_1^4}\right)}v_{j,pq}^{(\epsilon)}(\cdot,\lambda_0)\|_{L^1\cap L^\infty}\leq C, \quad j=1,\cdots,4,\quad 0\leq q\leq p\leq N,
      \end{equation}
      \begin{equation}\label{49}
      	\|{\rm{e}}^{\pm2i\nu_4(-\lambda_0)\log_0(\cdot)\mp\frac{i(\cdot)^2}{2}\left(1+\frac{\lambda_0^3(\cdot)}{a\eta_2^4}\right)}v_{j,pq}^{(\epsilon)}(\cdot,\lambda_0)\|_{L^1\cap L^\infty}\leq C, \quad j=7,\cdots,10,\quad 0\leq q\leq p\leq N.
      \end{equation}

     Fix $z_1$ and $z_2$, and then analyze some terms as $t\rightarrow \infty$. According to Lemma \ref{lem43}, it can be immediately obtained that $\frac{\lambda_0^3z_1}{a\eta_1^4},\frac{\lambda_0^3z_2}{a\eta_2^4}\rightarrow0$, and
     $$r_{1,a}\rightarrow \sum_{n=0}^{N}\frac{r^{(n)}_1(\lambda_0)}{n!}(az_1)^n,\quad \hat{r}_{1,a}\rightarrow \sum_{n=0}^{N}\frac{\hat{r}^{(n)}_1(\lambda_0)}{n!}(az_1)^n,$$
     and
     $$r_{2,a}\rightarrow \sum_{n=0}^{N}\frac{r^{(n)}_2(-\lambda_0)}{n!}(az_2)^n,\quad \hat{r}_{2,a}\rightarrow \sum_{n=0}^{N}\frac{\hat{r}^{(n)}_2(-\lambda_0)}{n!}(az_2)^n.$$

     According to the Beals-Coifman theory \cite{Beals-Coifman-1984}, the trivial decomposition of the jump matrix $v^{(\epsilon)}$  corresponding to point $\lambda_0$ is performed here. Setting $\omega^{(1)}_-=O$, $\omega^{(1)}_+=v^{(\epsilon)}-I$, and using the Cauchy operator $\mathcal{C}$ defined by equation \eqref{27}, we have
     \begin{equation*}
     	\begin{aligned}
     		\mathcal{C}_{\omega^{(1)}}f&=\mathcal{C}_-(f(v^{(\epsilon)}-I))\\
     		&=\sum_{0\leq q \leq p\leq N}\frac{(\ln t)^q}{(bt)^{p/2}}\mathcal{C}_{pq}f+E_c(t,\lambda(z_1))f,
     	\end{aligned}
     \end{equation*}
     where
     \begin{equation*}
     	\begin{aligned}
     		\mathcal{C}_{pq}f=&\mathcal{C}_-(f{\rm{e}}^{\pm2i\nu_1(\lambda_0)\log_{-\pi}(\cdot)\mp\frac{i(\cdot)^2}{2}(1-\frac{\lambda_0^3(\cdot)}{a\eta1^4})}v^{(\epsilon)}_{pq}),
     	\end{aligned}
     \end{equation*}
     and
     \begin{equation*}
     E_c(t,z_1)f=\mathcal{C}_-(fE^{(\epsilon)}_v(\cdot,t,\lambda_0)).
     \end{equation*}
    From the aforementioned estimates in the equations (\ref{46})-(\ref{49}), it can be concluded that for $0< \lambda_0< M$ the following equation holds uniformly
     \begin{equation*}
     	\begin{aligned}
     	&\|\mathcal{C}_{pq}\|_{\dot{L}^3\cap L^\infty\rightarrow \dot{L}^3}\leq c_{pq},\\
     	&\|E_c(t,z_1)\|_{\dot{L}^3\cap L^\infty\rightarrow \dot{L}^3}=\mathcal{O}\left( \frac{(\ln t)^{N+1}}{t^{(N+1)/2}}\right).\\
     	\end{aligned}
     \end{equation*}
     In addition, for $t$ sufficiently large and $0< \lambda_0< M$, the operator $1-\mathcal{C}_{00}$ is invertible and $(1-\mathcal{C}_{00})^{-1}$ is a bounded operator. Suppose $\mu^{(\epsilon)}_{\lambda_0}\in I+\dot{L}^3(\Sigma^{\lambda_0})$ is the solution to $(1-\mathcal{C}_{\omega^{(1)}})\mu^{(\epsilon)}_{\lambda_0}=I$, thus
     \begin{equation*}
    	\begin{aligned}
    	\mu^{(\epsilon)}_{\lambda_0}&=(1-\mathcal{C}_{\omega^{(1)}})^{-1}I\\
    	&=I+\sum_{0\leq q \leq p\leq N}\frac{(\ln t)^q}{(bt)^{p/2}}\mathcal{C}_{pq}I+\mathcal{O}\left( \frac{(\ln t)^{N+1}}{t^{(N+1)/2}}\right),
    	\end{aligned}
    \end{equation*}
    and we can always take a sufficiently large $t$ such that $\|\omega^{(1)}\|_{L^p(\Sigma^{\lambda_0}}<\|\mathcal{C}_-\|_{L^p(\Sigma^{\lambda_0})\rightarrow L^p(\Sigma^{\lambda_0})}^{-1}$, hence it follows that

     \begin{equation*}
    	\begin{aligned}
    		\|\mu^{(\epsilon)}_{\lambda_0}-I\|_{L^p(\Sigma^{\lambda_0})}&=\|(1-\mathcal{C}_{\omega^{(1)}})^{-1}\mathcal{C}_{\omega^{(1)}}I\|_{L^p(\Sigma^{\lambda_0})}\\
    		&\leq  \sum_{j=1}^{\infty}\|\mathcal{C}_{\omega^{(1)}}\|^{j-1}_{L^p(\Sigma^{\lambda_0})\rightarrow L^p(\Sigma^{\lambda_0})}\|\mathcal{C}_{\omega^{(1)}}I\|_{L^P(\Sigma^{\lambda_0})}\\
    		&\leq \sum_{j=1}^{\infty} \|\mathcal{C}_-\|_{L^p(\Sigma^{\lambda_0})\rightarrow L^p(\Sigma^{\lambda_0})}^j\|\omega^{(1)}\|_{L^\infty(\Sigma^{\lambda_0})}^{j-1}\|\omega^{(1)}\|_{L^p(\Sigma^{\lambda_0})}\\
    		&=\frac{\|\mathcal{C}_-\|_{L^p(\Sigma^{\lambda_0})\rightarrow L^p(\Sigma^{\lambda_0})}\|\omega^{(1)}\|_{L^p(\Sigma^{\lambda_0})}}{1-\|\mathcal{C}_-\|_{L^p(\Sigma^{\lambda_0})\rightarrow L^p(\Sigma^{\lambda_0})}\|\omega^{(1)}\|_{L^\infty(\Sigma^{\lambda_0})}}\\
    		&\leq \frac{C(\ln t)^\frac{1}{p}}{t^\frac{1}{2}}.
    	\end{aligned}
    \end{equation*}
     In other words, one has
     \begin{equation}\label{50}
     	\begin{aligned}
     		M_{\lambda_0}^{(\epsilon)}(x,t,z_1)&=I+\mathcal{C}(\mu^{(\epsilon)}_{\lambda_0}\omega_+^{(1)})\\
     		&=I+\frac{1}{2\pi i}\int_{\Sigma^{\lambda_0}}\frac{\mu^{(\epsilon)}_{\lambda_0}(x,t,\zeta)\omega_{+,j}^{(1)}(x,t,\zeta)}{\zeta-z_1}{\rm{d}}\zeta\\
     		&=I+\frac{1}{2\pi i}\int_{\Sigma^{\lambda_0}}\frac{\omega_{+}^{(1)}(x,t,\zeta)}{\zeta-z_1}{\rm{d}}\zeta+\frac{1}{2\pi i}\int_{\Sigma^{\lambda_0}}\frac{(\mu^{(\epsilon)}_{\lambda_0}(x,t,\zeta)-I)\omega_{+,j}^{(1)}(x,t,\zeta)}{\zeta-z_1}{\rm{d}}\zeta\\
     		&=I+\frac{1}{2\pi i}\sum_{j=1}^{4}
     		\!\int_{\Sigma^{\lambda_0}_j}\!\!\!\frac{v^{(\epsilon)}_j(x,t,\zeta)-I}{\zeta-z_1}{\rm{d}}\zeta\!+\!\frac{1}{2\pi i}\!\int_{\Sigma^{\lambda_0}}\!\!\!\!\frac{(\mu^{(\epsilon)}_{\lambda_0}(x,t,\zeta)\!-\!I)(v^{(\epsilon)}_j(x,t,\zeta)\!-\!I)}{\zeta-z_1}{\rm{d}}\zeta\\
     		&=I+\frac{\sum\limits_{0\leq q\leq p\leq N}\frac{(\ln t)^q}{(bt)^{p/2}}(M_{\lambda_0}^{(1)})_{pq}(x,t)}{z_1}+\mathcal{O}\left( \frac{(\ln t)^{N+1}}{t^{(N+2)/2}}\right),
     	\end{aligned}
     \end{equation}
     where
     \begin{equation}\label{51}
     	(M_{\lambda_0}^{(1)})_{pq}(x,t)=\sum_{j=1}^{4}\frac{1}{2 \pi i}\int_{\Sigma^{\lambda_0}_j}{\rm{e}}^{\pm2i\nu_1(\lambda_0)\log_{-\pi}\zeta\mp\frac{i\zeta^2}{2}\left(1-\frac{\lambda_0^3\zeta}{a\eta_1^4}\right)}v_{j,pq}^{(\epsilon)}(\zeta,\lambda_0){\rm{d}}\zeta,
     \end{equation}
     and the choice of the $\pm$ sign in the above depends on the value of the index $j$.

    Similarly, $M^{(\epsilon)}_{-\lambda_0}$ can be expanded as
     \begin{equation}\label{52}
     		M_{-\lambda_0}^{(\epsilon)}(x,t,z_2)=I+\frac{\sum\limits_{0\leq q\leq p\leq N}\frac{(\ln t)^q}{(bt)^{p/2}}(M_{-\lambda_0}^{(1)})_{pq}(x,t)}{z_2}+\mathcal{O}\left( \frac{(\ln t)^{N+1}}{t^{(N+2)/2}}\right),
     \end{equation}
     and
      \begin{equation}\label{53}
     	(M_{-\lambda_0}^{(1)})_{pq}(x,t)=\sum_{j=7}^{10}\frac{1}{2 \pi i}\int_{\Sigma^{-\lambda_0}_j}{{\rm{e}}^{\pm2i\nu_4(-\lambda_0)\log_0\zeta\mp\frac{i\zeta^2}{2}\left(1-\frac{\lambda_0^3\zeta}{a\eta_2^4}\right)}v_{j,pq}^{(\epsilon)}(\zeta,-\lambda_0)}{\rm{d}}\zeta,
     \end{equation}
      the choice of the $\pm$ sign depends on the value of $j$.

      \subsubsection{Small norm RH problem}
\ \ \ \
   Based on the scaling transformations made according to equation \eqref{42}, the functions $M_{\pm\lambda_0}^{(\epsilon)}$ in the equations (\ref{50})-(\ref{53}) concerning $z_1$ and $z_2$ are transformed into the following functions concerning $\lambda$:
    \begin{equation*}\left\lbrace
     	\begin{aligned}
	        	&	M_{\lambda_0}^{(\epsilon)}(x,t,\lambda)=I+\frac{\sum\limits_{0\leq q\leq p\leq N}\frac{(\ln t)^q}{(bt)^{(p+1)/2}}(M_{\lambda_0}^{(1)})_{pq}(x,t)}{\lambda-\lambda_0}+\mathcal{O}\left( \frac{(\ln t)^{N+1}}{t^{(N+2)/2}}\right),\\
	        		&M_{-\lambda_0}^{(\epsilon)}(x,t,\lambda)=I+\frac{\sum\limits_{0\leq q\leq p\leq N}\frac{(\ln t)^q}{(bt)^{(p+1)/2}}(M_{-\lambda_0}^{(1)})_{pq}(x,t)}{\lambda+\lambda_0}+\mathcal{O}\left( \frac{(\ln t)^{N+1}}{t^{(N+2)/2}}\right).
    	\end{aligned}\right.
\end{equation*}
\par
   To transform the jumps on the cross into jumps on the disc, the following definitions are made
   \begin{equation*}
   	\begin{aligned}
     M_{\lambda_0}(x,t,\lambda)&=H(\lambda_0,t)M_{\lambda_0}^{(\epsilon)}(x,t,\lambda)H(\lambda_0,t)^{-1}\\
     &=I+H(\lambda_0,t)\left( \frac{\sum\limits_{0\leq q\leq p\leq N}\frac{(\ln t)^q}{(bt)^{(p+1)/2}}(M_{\lambda_0}^{(1)})_{pq}(x,t)}{\lambda-\lambda_0}\right) H(\lambda_0,t)^{-1}+\mathcal{O}\left( \frac{(\ln t)^{N+1}}{t^{(N+2)/2}}\right),\\
   	\end{aligned}
   \end{equation*}
   \begin{equation*}
   	\begin{aligned}
   	M_{-\lambda_0}(x,t,\lambda)&=H(-\lambda_0,t)M_{-\lambda_0}^{(\epsilon)}(x,t,\lambda)H(-\lambda_0,t)^{-1}\\
   	&=I\!+\!H(-\lambda_0,t)\!\left( \frac{\sum\limits_{0\leq q\leq p\leq N}\!\!\frac{(\ln t)^q}{(bt)^{(p+1)/2}}(M_{-\lambda_0}^{(1)})_{pq}(x,t)}{\lambda+\lambda_0}\right)\! H(-\lambda_0,t)^{-1}\!+\!\mathcal{O}\!\left( \frac{(\ln t)^{N+1}}{t^{(N+2)/2}}\right).\\
   \end{aligned}
      \end{equation*}
      \par
Furthermore, the eigenfunctions at the other four critical points can be given by the symmetries below
     \begin{equation*}
    	M_{\pm\omega\lambda_0}(x,t,\lambda)=\tilde{\sigma}_1M_{\pm\lambda_0}(x,t,\omega^2\lambda)\tilde{\sigma}_1^{-1},
    \end{equation*}
    \begin{equation*}
    	M_{\pm\omega^2\lambda_0}(x,t,\lambda)=\tilde{\sigma}_1^2M_{\pm\lambda_0}(x,t,\omega\lambda)\tilde{\sigma}_1^{-2}.
    \end{equation*}

     For $j=0,1,2$, define the matrix functions that are non-constant only on the local circle by
    \begin{equation*}\hat{M}_{\pm\omega^j\lambda_0}(\lambda)=\left\lbrace
    	\begin{aligned}
    		&M_{\pm\omega^j\lambda_0}(\lambda), \quad && \lambda\in \partial B_\epsilon(\pm\omega^j\lambda_0),\\
    		&I, &&\lambda\in \bigcup\limits_{i=0}^{2}\partial B_\epsilon(\pm\omega^i\lambda_0)\setminus\partial B_\epsilon(\pm\omega^j\lambda_0).
    	\end{aligned}\right.
    \end{equation*}
    The above efforts are made to construct the following matrix-valued eigenfunction $M^{(3)}(x,t,\lambda)$ without jumps on contours $\Sigma^{\pm\omega^j\lambda_0}$ for $j=0,1,2$:
    \begin{equation}\label{54}
    	M^{(3)}(x,t,\lambda):=\left\lbrace\begin{aligned}
    	&	M^{(2)} (\hat{M}_{\pm\omega^j\lambda_0}(x,t,\lambda))^{-1},\quad &&\lambda\in B_{\epsilon}(\pm\omega^j\lambda_0),\\
    	&	M^{(2)},&& otherelse.
    	\end{aligned}
    	 \right.
    \end{equation}
    \par
    Define the contour $\Sigma^{(3)}=\bigcup\limits_{j=0}^{2}\partial B_{\epsilon}(\pm\omega^j\lambda_0)\cup\Sigma^{(2)}\setminus\bigcup\limits_{j=0}^{2}\Sigma^{\pm\omega^j\lambda_0}$ as shown in Fig. \ref{figSigma3}.
    The new jump relation corresponding to the eigenfunction $M^{(3)}(x,t,\lambda)$ in \eqref{54} is obtained by combining the transformation
    \begin{equation}\label{55}
    	v^{(3)}:=\left\lbrace
    	\begin{aligned}
    	&	v^{(2)},\quad &&\lambda\in \Sigma^{(2)}\setminus\bigcup_{j=0}^{2}\Sigma^{\pm\omega^j\lambda_0},\\
    	&	(\hat{M}_{\pm\omega^j\lambda_0}(x,t,\lambda))^{-1}, &&\lambda\in\bigcup_{j=0}^{2}\partial B_{\epsilon}(\pm\omega^j\lambda_0),\\
    	& I, && \lambda\in \Sigma^{\pm\omega^j\lambda_0}.
    	\end{aligned}
    	\right.
    \end{equation}

    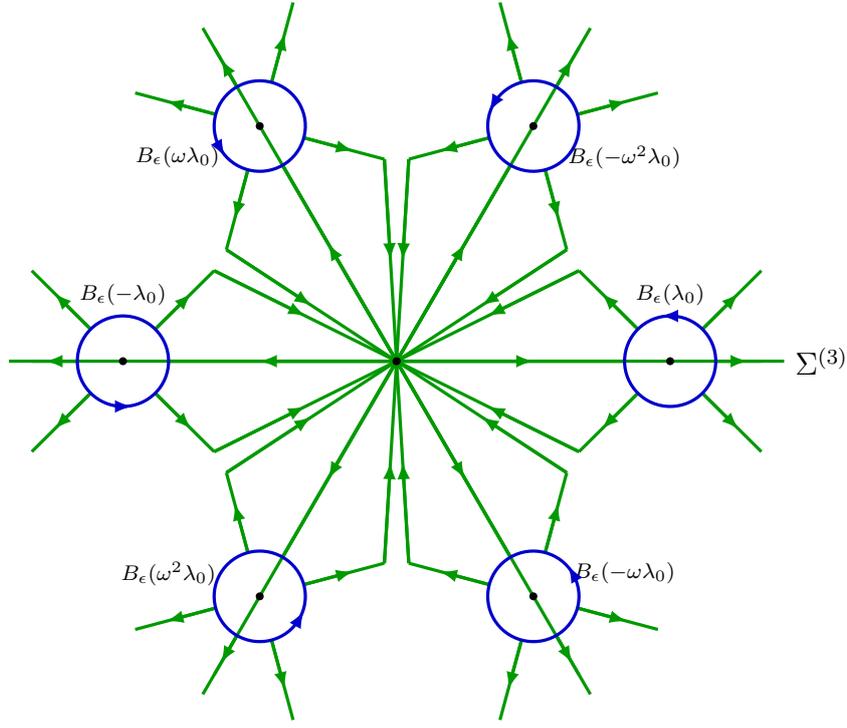
\begin{figure}[htbp]
    	\centering
    	\begin{tikzpicture}[scale=0.6]
    		
            \draw [very thick,black!40!green](0,0) -- (8.5,0);
            \draw [very thick,black!40!green](-8.5,0) -- (0,0);

    		\draw [very thick,black!40!green](0,0) -- (4,2);
    		\draw [very thick,black!40!green](4,2) -- (8,-2);
    		\draw [very thick,black!40!green](0,0) -- (-4,2);
    		\draw [very thick,black!40!green](-4,2) -- (-8,-2);
    		
    		\draw [very thick,black!40!green](0,0) -- (4,-2);
    		\draw [very thick,black!40!green](4,-2) -- (8,2);
    		\draw [very thick,black!40!green](0,0) -- (-4,-2);
    		\draw [very thick,black!40!green](-4,-2) -- (-8,2);
    		
    		\draw[very thick, black!40!green, -latex]  (4,2) -- (2,1);
    		\draw[very thick, black!40!green, -latex]  (4,-2) -- (2,-1);
    		\draw[very thick, black!40!green, -latex]  (-4,2) -- (-2,1);
    		\draw[very thick, black!40!green, -latex]  (-4,-2) -- (-2,-1);
    		
    		\draw[very thick, black!40!green, -latex]  (6,0) -- (4.5,1.5);
    		\draw[very thick, black!40!green, -latex]  (6,0) -- (4.5,-1.5);
    		\draw[very thick, black!40!green, -latex]  (-6,0) -- (-4.5,1.5);
    		\draw[very thick, black!40!green, -latex]  (-6,0) -- (-4.5,-1.5);
    		
    		\draw[very thick, black!40!green, -latex]  (0,0) -- (3,0);
    		\draw[very thick, black!40!green, -latex]  (0,0) -- (-3,0);
    		\draw[very thick, black!40!green, -latex]  (0,0) -- (7.7,0);
    		\draw[very thick, black!40!green, -latex]  (0,0) -- (-7.7,0);
    		
    		\draw[very thick, black!40!green, -latex]  (6,0) -- (7.5,1.5);
    		\draw[very thick, black!40!green, -latex]  (6,0) -- (7.5,-1.5);
    		\draw[very thick, black!40!green, -latex]  (-6,0) -- (-7.5,1.5);
    		\draw[very thick, black!40!green, -latex]  (-6,0) -- (-7.5,-1.5);

    		\fill[white] (6,0) circle (1cm);
    	
    		\draw [very thick,black!40!green](0,0) -- (8,0);
    		\draw[black!20!blue, very thick] (6,0) circle (1cm);
    		
    		\fill[white] (-6,0) circle (1cm);
    		
    		\draw [very thick,black!40!green](0,0) -- (-8,0);
    		\draw[black!20!blue, very thick] (-6,0) circle (1cm);
    		
    		\fill (0,0) circle (2.5pt);
    		\fill (6,0) circle (2.5pt);
    		\fill (-6,0) circle (2.5pt);

    		\node[right] at (8.5,0) {$\Sigma^{(3)}$};
    		\node[above] at (6,1) {\scriptsize$B_{\epsilon}(\lambda_0)$};
    		\node[above] at (-6,1) {\scriptsize$B_{\epsilon}(-\lambda_0)$};

    		\begin{scope}[rotate around={60:(center)}]

    		\draw [very thick,black!40!green](0,0) -- (8.5,0);
    		\draw [very thick,black!40!green](-8.5,0) -- (0,0);

    		\draw [very thick,black!40!green](0,0) -- (4,2);
    		\draw [very thick,black!40!green](4,2) -- (8,-2);
    		\draw [very thick,black!40!green](0,0) -- (-4,2);
    		\draw [very thick,black!40!green](-4,2) -- (-8,-2);
    		
    		\draw [very thick,black!40!green](0,0) -- (4,-2);
    		\draw [very thick,black!40!green](4,-2) -- (8,2);
    		\draw [very thick,black!40!green](0,0) -- (-4,-2);
    		\draw [very thick,black!40!green](-4,-2) -- (-8,2);
    		
    		\draw[very thick, black!40!green, -latex]  (4,2) -- (2,1);
    		\draw[very thick, black!40!green, -latex]  (4,-2) -- (2,-1);
    		\draw[very thick, black!40!green, -latex]  (-4,2) -- (-2,1);
    		\draw[very thick, black!40!green, -latex]  (-4,-2) -- (-2,-1);
    		
    		\draw[very thick, black!40!green, -latex]  (6,0) -- (4.5,1.5);
    		\draw[very thick, black!40!green, -latex]  (6,0) -- (4.5,-1.5);
    		\draw[very thick, black!40!green, -latex]  (-6,0) -- (-4.5,1.5);
    		\draw[very thick, black!40!green, -latex]  (-6,0) -- (-4.5,-1.5);
    		
    		\draw[very thick, black!40!green, -latex]  (0,0) -- (3,0);
    		\draw[very thick, black!40!green, -latex]  (0,0) -- (-3,0);
    		\draw[very thick, black!40!green, -latex]  (0,0) -- (7.7,0);
    		\draw[very thick, black!40!green, -latex]  (0,0) -- (-7.7,0);
    		
    		\draw[very thick, black!40!green, -latex]  (6,0) -- (7.5,1.5);
    		\draw[very thick, black!40!green, -latex]  (6,0) -- (7.5,-1.5);
    		\draw[very thick, black!40!green, -latex]  (-6,0) -- (-7.5,1.5);
    		\draw[very thick, black!40!green, -latex]  (-6,0) -- (-7.5,-1.5);

    		\fill[white] (6,0) circle (1cm);
    		
    		\draw [very thick,black!40!green](0,0) -- (8,0);
    		\draw[black!20!blue, very thick] (6,0) circle (1cm);
    		
    		\fill[white] (-6,0) circle (1cm);
    		
    		\draw [very thick,black!40!green](0,0) -- (-8,0);
    		\draw[black!20!blue, very thick] (-6,0) circle (1cm);
    		
    		\fill (0,0) circle (2.5pt);
    		\fill (6,0) circle (2.5pt);
    		\fill (-6,0) circle (2.5pt);

    		\end{scope}
    		
    		\node[below] at (5,5) {\scriptsize$B_{\epsilon}(-\omega^2\lambda_0)$};
    		\node[below] at (-5,-4.2) {\scriptsize$B_{\epsilon}(\omega^2\lambda_0)$};

    		\coordinate (center) at (0,0);
    		
    		\begin{scope}[rotate around={120:(center)}]

    		\draw [very thick,black!40!green](0,0) -- (8.5,0);
    		\draw [very thick,black!40!green](-8.5,0) -- (0,0);

    		\draw [very thick,black!40!green](0,0) -- (4,2);
    		\draw [very thick,black!40!green](4,2) -- (8,-2);
    		\draw [very thick,black!40!green](0,0) -- (-4,2);
    		\draw [very thick,black!40!green](-4,2) -- (-8,-2);
    		
    		\draw [very thick,black!40!green](0,0) -- (4,-2);
    		\draw [very thick,black!40!green](4,-2) -- (8,2);
    		\draw [very thick,black!40!green](0,0) -- (-4,-2);
    		\draw [very thick,black!40!green](-4,-2) -- (-8,2);
    		
    		\draw[very thick, black!40!green, -latex]  (4,2) -- (2,1);
    		\draw[very thick, black!40!green, -latex]  (4,-2) -- (2,-1);
    		\draw[very thick, black!40!green, -latex]  (-4,2) -- (-2,1);
    		\draw[very thick, black!40!green, -latex]  (-4,-2) -- (-2,-1);
    		
    		\draw[very thick, black!40!green, -latex]  (6,0) -- (4.5,1.5);
    		\draw[very thick, black!40!green, -latex]  (6,0) -- (4.5,-1.5);
    		\draw[very thick, black!40!green, -latex]  (-6,0) -- (-4.5,1.5);
    		\draw[very thick, black!40!green, -latex]  (-6,0) -- (-4.5,-1.5);
    		
    		\draw[very thick, black!40!green, -latex]  (0,0) -- (3,0);
    		\draw[very thick, black!40!green, -latex]  (0,0) -- (-3,0);
    		\draw[very thick, black!40!green, -latex]  (0,0) -- (7.7,0);
    		\draw[very thick, black!40!green, -latex]  (0,0) -- (-7.7,0);
    		
    		\draw[very thick, black!40!green, -latex]  (6,0) -- (7.5,1.5);
    		\draw[very thick, black!40!green, -latex]  (6,0) -- (7.5,-1.5);
    		\draw[very thick, black!40!green, -latex]  (-6,0) -- (-7.5,1.5);
    		\draw[very thick, black!40!green, -latex]  (-6,0) -- (-7.5,-1.5);

    		\fill[white] (6,0) circle (1cm);
    		
    		\draw [very thick,black!40!green](0,0) -- (8,0);
    		\draw[black!20!blue, very thick] (6,0) circle (1cm);
    		
    		\fill[white] (-6,0) circle (1cm);
    		
    		\draw [very thick,black!40!green](0,0) -- (-8,0);
    		\draw[black!20!blue, very thick] (-6,0) circle (1cm);
    		
    		\fill (0,0) circle (2.5pt);
    		\fill (6,0) circle (2.5pt);
    		\fill (-6,0) circle (2.5pt); 
    		
    		\end{scope}
    		
    		\node[below] at (-4.8,5) {\scriptsize$B_{\epsilon}(\omega\lambda_0)$};
    		\node[below] at (5,-4.2) {\scriptsize$B_{\epsilon}(-\omega\lambda_0)$};
    		
    		\draw[very thick, black!20!blue, -latex] (6,1) -- (5.8,1);
    		
    		\begin{scope}[rotate around={60:(center)}]
    		\draw[very thick, black!20!blue, -latex] (6,1) -- (5.8,1);
    		\end{scope}
    		
    		\begin{scope}[rotate around={120:(center)}]
    			\draw[very thick, black!20!blue, -latex] (6,1) -- (5.8,1);
    		\end{scope}
    		
    		\begin{scope}[rotate around={180:(center)}]
    			\draw[very thick, black!20!blue, -latex] (6,1) -- (5.8,1);
    		\end{scope}
    		
    		\begin{scope}[rotate around={240:(center)}]
    			\draw[very thick, black!20!blue, -latex] (6,1) -- (5.8,1);
    		\end{scope}
    		
    		\begin{scope}[rotate around={300:(center)}]
    			\draw[very thick, black!20!blue, -latex] (6,1) -- (5.8,1);
    		\end{scope}

    	\end{tikzpicture}
    	\caption{The jump contour $\Sigma^{(3)}$ in the complex $\lambda$-plane}
    	\label{figSigma3}
    \end{figure}

   Based on the properties of $v^{(2)}$, we can directly obtain the following lemma regarding the estimate of the jump relationship over different regions.
    \begin{lem}
    	Define $\omega^{(3)}=v^{(3)}-I$, the jump relationship \eqref{55} satisfies the following estimates for $0<\lambda_0<M$ and for $t$ sufficiently large
    	\begin{equation*}
    		\|(1+\vert\cdot\vert)\omega^{(3)}\|_{(L^1\cap L^\infty)(\Sigma^{(1)})}\leq \frac{C}{t^{N+\frac{3}{2}}},
    	\end{equation*}
    		\begin{equation*}
    		\|(1+\vert\cdot\vert)\omega^{(3)}\|_{(L^1\cap L^\infty)(\Sigma^{(2)}\setminus (\cup_{j=0}^{2}\Sigma^{\pm\omega^j\lambda_0}  \cup\Sigma^{(1)}))}\leq C{\rm{e}}^{-ct},
    	\end{equation*}
    	\begin{equation*}
    			\|\omega^{(3)}\|_{(L^1\cap L^\infty)(\bigcup_{j=0}^{2}\partial B_{\epsilon}(\pm\omega^j\lambda_0))}\leq C t^{-\frac{1}{2}}.
    	\end{equation*}
    \end{lem}

    Thus the asymptotic solution in Region ${\rm IV}$ is mainly given by the information on $\bigcup\limits_{j=0}^{2}\partial B_{\epsilon}(\pm\omega^j\lambda_0)$ as $t$ large enough.

    \subsubsection{The higher-order asymptotic solutions}\label{sub425}
    \ \ \ \
   The operator $\mathcal{C}_{\omega^{(3)}}$  on $\bigcup\limits_{j=0}^{2}\partial B_{\epsilon}(\pm\omega^j\lambda_0)$ is given by
   \begin{equation*}
   	\begin{aligned}
   		 \mathcal{C}_{\omega^{(3)}} f&=\sum_{j=0}^{2} \mathcal{C}_-\{f[(\hat{M}_{\omega^j\lambda_0})^{-1}-I]\}+\sum_{j=0}^{2} \mathcal{C}_-\{f[(\hat{M}_{-\omega^j\lambda_0})^{-1}-I]\}\\
   		 &:=\sum_{j=0}^{2}\mathcal{A}_jf+\sum_{j=0}^{2}\mathcal{B}_jf,
   	\end{aligned}
   \end{equation*}
   where $\mathcal{A}_jf=\mathcal{C}_-\{f[(\hat{M}_{\omega^j\lambda_0})^{-1}-I]\}$ and $\mathcal{B}_jf=\mathcal{C}_-\{f[(\hat{M}_{-\omega^j\lambda_0})^{-1}-I]\}$.
    \begin{lem}
       The operator $1-\mathcal{C}_{\omega^{(3)}}$ is invertible and $(1-\mathcal{C}_{\omega^{(3)}})^{-1}$ is a bounded linear operator.
    \end{lem}

    Let $\mu^{(3)}\in I+\dot{L}^3(\bigcup_{j=0}^{2}\partial B_{\epsilon}(\pm\omega^j\lambda_0))$  be the solution to the equation $\mu^{(3)}=I+\mathcal{C}_{\omega^{(3)}}\mu^{(3)}$, which is \begin{equation*}
    	\begin{aligned}
    	\mu^{(3)}(x,t,\lambda)&=(1-\mathcal{C}_{\omega^{(3)}})^{-1}I\\
    		&=\left( 1-\sum_{j=0}^{2}\mathcal{A}_j-\sum_{j=0}^{2}\mathcal{B}_j\right) ^{-1}I\\
    		&=I+\sum_{n=1}^{N}\left( \sum_{j=0}^{2}\mathcal{A}_j+\sum_{j=0}^{2}\mathcal{B}_j\right) ^nI+E^{(3)}(\lambda,t,\lambda_0),
    	\end{aligned}
    \end{equation*}
    where $\|E^{(3)}(\cdot,t,\lambda_0)\|_{\dot{L}^3(\bigcup_{j=0}^{2}\partial B_{\epsilon}(\pm\omega^j\lambda_0))}=\mathcal{O}\left( \frac{1}{t^{(N+2)/2}}\right)    $ and $\|\mu^{(3)}-I\|_{L^p(\Sigma^{(3)})}\leq {C}{t^{-\frac{1}{2}}}$.

    Therefore, the RH problem concerning $M^{(3)}(x,t,\lambda)$ has a unique solution
    \begin{equation*}
    	\begin{aligned}
    			M^{(3)}(x,t,\lambda)&=I+\mathcal{C}(\mu^{(3)} \omega^{(3)})\\
    			&=I+\frac{1}{2 \pi i}\int_{\Sigma^{(3)}}\frac{\mu^{(3)}(x,t,\zeta)\omega^{(3)}(x,t,\zeta)}{\zeta-\lambda}{\rm{d}}\zeta\\
    			&=I+\frac{1}{2 \pi i}\int_{\bigcup\limits_{j=0}^{2}\partial B_{\epsilon}(\pm\omega^j\lambda_0)}\frac{\mu^{(3)}(x,t,\zeta)\omega^{(3)}(x,t,\zeta)}{\zeta-\lambda}{\rm{d}}\zeta+\mathcal{O}\left( \frac{C}{t^{N+\frac{3}{2}}}\right) .\\
    	\end{aligned}
    \end{equation*}
    \par
    In the reconstruction formula \eqref{7}, it is necessary to take the limit of the eigenfunction $M(x,t,\lambda)$ with respect to $\lambda\rightarrow0$ to obtain the solution of the two-component nonlinear KG equation \eqref{3}.
    Continuing with the non-tangential limit of $M^{(3)}(x,t,\lambda)$ as $\lambda\rightarrow0$, define it as $S(x,t)$ and calculate to get
   \begin{equation*}
   	\begin{aligned}
   	  S(x,t)&=\lim_{\lambda\rightarrow0}	M^{(3)}(x,t,\lambda)\\
   	  &=I+\frac{1}{2 \pi i}\int_{\bigcup\limits_{j=0}^{2}\partial B_{\epsilon}(\pm\omega^j\lambda_0)}\frac{\mu^{(3)}(x,t,\zeta)\omega^{(3)}(x,t,\zeta)}{\zeta}{\rm{d}}\zeta+\mathcal{O}\left( \frac{C}{t^{N+\frac{3}{2}}}\right) \\
   	  &=I+\frac{1}{2 \pi i}\int_{\bigcup\limits_{j=0}^{2}\partial B_{\epsilon}(\pm\omega^j\lambda_0)}\frac{\omega^{(3)}(x,t,\zeta)}{\zeta}{\rm{d}}\zeta\\
   	  &\quad +\frac{1}{2 \pi i}\int_{\bigcup\limits_{j=0}^{2}\partial B_{\epsilon}(\pm\omega^j\lambda_0)}\frac{(\mu^{(3)}(x,t,\zeta)-I)\omega^{(3)}(x,t,\zeta)}{\zeta}{\rm{d}}\zeta+\mathcal{O}\left( \frac{C}{t^{N+\frac{3}{2}}}\right) .\\
   	\end{aligned}
   \end{equation*}
   \par
   Since $\omega^{(3)}(x,t,\lambda)=	(\hat{M}_{\pm\omega^j\lambda_0}(x,t,\lambda))^{-1}-I$ holds on region $\lambda\in\bigcup\limits_{j=0}^{2}\partial B_{\epsilon}(\pm\omega^j\lambda_0)$, so we have
   \begin{equation*}
   	\begin{aligned}
   			&\frac{1}{2 \pi i}\int_{\partial B_{\epsilon}(\lambda_0)}\frac{\omega^{(3)}(x,t,\zeta)}{\zeta}{\rm{d}}\zeta\\
   			&\quad=-\frac{1}{2 \pi i}\int_{\partial B_{\epsilon}(\omega^j\lambda_0)}\frac{H(\lambda_0,t)\left( \frac{\sum\limits_{0\leq q\leq p\leq N}\frac{(\ln t)^q}{(bt)^{(p+1)/2}}(M_{\lambda_0}^{(1)})_{pq}(x,t)}{\zeta-\lambda_0}\right) H(\lambda_0,t)^{-1}}{\zeta}{\rm{d}}\zeta+\mathcal{O}\left( \frac{(\ln t)^{N+1}}{t^{(N+2)/2}}\right) \\
   			 &\quad=-\frac{H(\lambda_0,t)\left( {\sum\limits_{0\leq q\leq p\leq N}\frac{(\ln t)^q}{(bt)^{p/2}}(M_{\lambda_0}^{(1)})_{pq}(x,t)}\right) H(\lambda_0,t)^{-1}}{\lambda_0\sqrt{bt}}+\mathcal{O}\left( \frac{(\ln t)^{N+1}}{t^{(N+2)/2}}\right) ,
   	\end{aligned}
   \end{equation*}
   and similarly
   \begin{equation*}
   	\begin{aligned}
   		&\frac{1}{2 \pi i}\int_{\partial B_{\epsilon}(-\lambda_0)}\frac{\omega^{(3)}(x,t,\zeta)}{\zeta}{\rm{d}}\zeta\\
   		&\quad=-\frac{1}{2 \pi i}\int_{\partial B_{\epsilon}(\omega^j\lambda_0)}\!\!\!\!\!\!\frac{H(-\lambda_0,t)\left( \frac{\sum\limits_{0\leq q\leq p\leq N}\frac{(\ln t)^q}{(bt)^{(p+1)/2}}(M_{-\lambda_0}^{(1)})_{pq}(x,t)}{\zeta+\lambda_0}\right) H(-\lambda_0,t)^{-1}}{\zeta}{\rm{d}}\zeta+\mathcal{O}\left( \frac{(\ln t)^{N+1}}{t^{(N+2)/2}}\right) \\
   		&\quad=\frac{H(-\lambda_0,t)\left( {\sum\limits_{0\leq q\leq p\leq N}\frac{(\ln t)^q}{(bt)^{p/2}}(M_{-\lambda_0}^{(1)})_{pq}(x,t)}\right) H(-\lambda_0,t)^{-1}}{\lambda_0\sqrt{bt}}+\mathcal{O}\left( \frac{(\ln t)^{N+1}}{t^{(N+2)/2}}\right) .
   	\end{aligned}
   \end{equation*}
   \par
   Finally, by utilizing the symmetry of $\hat{M}_{\pm\omega^j\lambda_0}(\lambda)$ on other circles, it is derived that

      \begin{align}\label{56}
     		S(x,t)
     		&=I -\sum_{j=0}^{2}\frac{\tilde{\sigma}_1^jH(\lambda_0,t)\left( {\sum\limits_{0\leq q\leq p\leq N}\frac{(\ln t)^q}{(bt)^{p/2}}(M_{\lambda_0}^{(1)})_{pq}(x,t)}\right) H(\lambda_0,t)^{-1}\tilde{\sigma}_1^{-j}}{\sqrt{\frac{2\sqrt{3}t\lambda_0}{1+\lambda_0^2}}}\nonumber\\
     		&\qquad +\sum_{j=0}^{2}\frac{\tilde{\sigma}_1^jH(-\lambda_0,t)\left( {\sum\limits_{0\leq q\leq p\leq N}\frac{(\ln t)^q}{(bt)^{p/2}}(M_{-\lambda_0}^{(1)})_{pq}(x,t)}\right) H(-\lambda_0,t)^{-1}\tilde{\sigma}_1^{-j}}{\sqrt{\frac{2\sqrt{3}t\lambda_0}{1+\lambda_0^2}}}\\
     		&\qquad+\mathcal{O}\left( \frac{(\ln t)^{N+1}}{t^{(N+2)/2}}\right) \nonumber.
     	\end{align}

     Since the jump matrix $v^{(\epsilon)}$ involves only elements in $(1,2)$ and $(2,1)$ positions in the equations \eqref{vep} and \eqref{vep1}, and the two position elements are conjugate to each other, it is reasonable to denote
     \begin{equation}\label{57}
     		(M_{\lambda_0}^{(1)})_{pq}(x,t)=\begin{pmatrix}
     		0 & a_{pq}(x,t) & 0\\
     		a^*_{pq}(x,t) & 0 & 0\\
     		0 & 0 & 0
     		\end{pmatrix},\quad
     		(M_{-\lambda_0}^{(1)})_{pq}(x,t)=\begin{pmatrix}
     			0 & b_{pq}(x,t) & 0\\
     			b^*_{pq}(x,t) & 0 & 0\\
     			0 & 0 & 0
     		\end{pmatrix}.
     \end{equation}
   Under this definition, substituting equation \eqref{57} into equation \eqref{56} yields
   \begin{equation*}
   	S(x,t)=I+\frac{1}{\sqrt{\frac{2\sqrt{3}t\lambda_0}{1+\lambda_0^2}}}\begin{pmatrix}
   	0 & s_1 & s_2\\
   	s_2 & 0 & s_1\\
   	s_1 & s_2 & 0
   	\end{pmatrix}+\mathcal{O}\left( \frac{(\ln t)^{N+1}}{t^{(N+2)/2}}\right),
   \end{equation*}
   where
    \begin{equation*}
    	\begin{aligned}
    			&s_1=\delta^0_{-\lambda_0}{\rm{e}}^{-\theta_{21}(-\lambda_0)}\left( \sum_{p=0}^{N}\sum_{q=0}^{p}\frac{(\ln t)^q}{(bt)^{p/2}}b_{pq}\right)-(\delta^0_{\lambda_0})^{-1}{\rm{e}}^{-\theta_{21}(\lambda_0)}\left( \sum_{p=0}^{N}\sum_{q=0}^{p}\frac{(\ln t)^q}{(bt)^{p/2}}a_{pq}\right),\\
    			&s_2=(\delta^0_{-\lambda_0})^{-1}{\rm{e}}^{\theta_{21}(-\lambda_0)}\left( \sum_{p=0}^{N}\sum_{q=0}^{p}\frac{(\ln t)^q}{(bt)^{p/2}}b_{pq}^*\right)-\delta^0_{\lambda_0}{\rm{e}}^{\theta_{21}(\lambda_0)}\left( \sum_{p=0}^{N}\sum_{q=0}^{p}\frac{(\ln t)^q}{(bt)^{p/2}}a^*_{pq}\right),\\
    	\end{aligned}
    \end{equation*}
    and the equality $s_1=s_2^*$ holds.
    \par
    Recalling the three transformations in \eqref{36}, \eqref{40} and \eqref{54} that we have made on the eigenfunction $M(x,t,\lambda)$ previously, it is obvious that
     \begin{equation*}\left\lbrace
     	\begin{aligned}
     		&\Delta(\lambda)\rightarrow I,\\
     		&W(x,t,\lambda)\rightarrow I, \qquad \quad as\, \lambda\rightarrow0.\\
     		&\hat{M}_{\pm\omega^j\lambda_0}(x,t,\lambda)\rightarrow I,
     	\end{aligned}\right.
     \end{equation*}
     \par
     \begin{lem}\label{RegionIV}
     	According to the reconstruction formula \eqref{7}, the higher-order asymptotic solutions to the initial value problem of the two-component nonlinear KG equation \eqref{3} in Region {\rm{IV}} can be expressed by
     	\begin{equation}\label{58}
     		\begin{aligned}
     			u(x,t)&=\frac{1}{2}\lim_{\lambda\rightarrow0}\ln [(\omega,\omega^2,1)M(x,t,\lambda)]_{13}\\
     			&=\frac{1}{2}\ln \left( 1+\!\!\!\!\sum\limits_{0\leq q\leq p\leq N}\!\!\!\frac{(\ln t)^q}{(bt)^{p/2}}\frac{2{\rm{Re}}\,\left( \omega^2\delta^0_{-\lambda_0}{\rm{e}}^{-\theta_{21}(-\lambda_0)}b_{pq}(x,t)\right) -2{\rm{Re}}\,\left(\omega \delta^0_{\lambda_0}{\rm{e}}^{\theta_{21}(\lambda_0)}a^*_{pq}(x,t)\right)}{\sqrt{\frac{2\sqrt{3}t\lambda_0}{1+\lambda_0^2}}}\right)\\
     			&\quad +\mathcal{O}\left( \frac{(\ln t)^{N+1}}{t^{(N+2)/2}}\right)	,	
     		\end{aligned}
     	\end{equation}
     	\begin{equation}\label{59}
     		\begin{aligned}
     			v(x,t)&=\frac{1}{6}\lim_{\lambda \rightarrow0} \ln[(\omega,\omega^2,1)M(x,t,\lambda)]_{13}-\frac{1}{3}\lim_{\lambda \rightarrow0} \ln[(\omega^2,\omega,1)M(x,t,\lambda)]_{13}\\
     			&=\frac{1}{6}\ln \left( 1+\!\!\!\!\sum\limits_{0\leq q\leq p\leq N}\!\!\!\frac{(\ln t)^q}{(bt)^{p/2}}\frac{2{\rm{Re}}\,\left( \omega^2\delta^0_{-\lambda_0}{\rm{e}}^{-\theta_{21}(-\lambda_0)}b_{pq}(x,t)\right) -2{\rm{Re}}\,\left(\omega \delta^0_{\lambda_0}{\rm{e}}^{\theta_{21}(\lambda_0)}a^*_{pq}(x,t)\right)}{\sqrt{\frac{2\sqrt{3}t\lambda_0}{1+\lambda_0^2}}}\right)\\
     			& \quad -\frac{1}{3}\ln \left( 1+\!\!\!\!\sum\limits_{0\leq q\leq p\leq N}\!\!\!\frac{(\ln t)^q}{(bt)^{p/2}}\frac{2{\rm{Re}}\,\left( \omega\delta^0_{-\lambda_0}{\rm{e}}^{-\theta_{21}(-\lambda_0)}b_{pq}(x,t)\right) -2{\rm{Re}}\,\left(\omega^2 \delta^0_{\lambda_0}{\rm{e}}^{\theta_{21}(\lambda_0)}a^*_{pq}(x,t)\right)}{\sqrt{\frac{2\sqrt{3}t\lambda_0}{1+\lambda_0^2}}}\right)\\
     			&\quad +\mathcal{O}\left( \frac{(\ln t)^{N+1}}{t^{(N+2)/2}}\right)	,
     		\end{aligned}
     	\end{equation}
     	where $\delta_{\lambda_0}^0(\lambda)=\frac{a^{2i\nu_1(\lambda_0)}{\rm{e}}^{-2\chi_1'(\lambda_0)}}{\delta_{N1}(\lambda_0)}$,  $\delta_{-\lambda_0}^0(-\lambda_0)=\frac{a^{2i\nu_4(-\lambda_0)}{\rm{e}}^{-2\chi_4'(-\lambda_0)}}{\delta_{N4}(-\lambda_0)}$, $\theta_{21}(\lambda_0)=\frac{-2 \sqrt{3}t \lambda_0}{1+\lambda_0^2}$ and $\theta_{21}(-\lambda_0)=\frac{2 \sqrt{3}t \lambda_0}{1+\lambda_0^2}$.
     \end{lem}

Finally, taking $N=0$ for an example, it is immediate that
\begin{equation*}
	\begin{aligned}
		&v_{1,00}^{\epsilon}=\begin{pmatrix}
			0 & r_1(\lambda_0) & 0\\
			0 & 0 & 0\\
			0 & 0 & 0
		\end{pmatrix},\quad 	&&v_{2,00}^{\epsilon}=\begin{pmatrix}
		0 & 0 & 0\\
		\frac{r^*_1(\lambda_0)}{1-r_1(\lambda_0)r_1^*(\lambda_0)} & 0 & 0\\
		0 & 0 & 0
		\end{pmatrix},\\
			&v_{3,00}^{\epsilon}=\begin{pmatrix}
			0 & \frac{r_1(\lambda_0)}{1-r_1(\lambda_0)r_1^*(\lambda_0)} & 0\\
			0 & 0 & 0\\
			0 & 0 & 0
		\end{pmatrix},\quad
			&&v_{4,00}^{\epsilon}=\begin{pmatrix}
			0 & 0 & 0\\
			r_1^*(\lambda_0) & 0 & 0\\
			0 & 0 & 0
		\end{pmatrix},
	\end{aligned}
\end{equation*}
and
\begin{equation*}
	\begin{aligned}
		&v_{7,00}^{\epsilon}=\begin{pmatrix}
			0 & 0 & 0\\
			r_2(-\lambda_0) & 0 & 0\\
			0 & 0 & 0
		\end{pmatrix},\quad 	&&v_{8,00}^{\epsilon}=\begin{pmatrix}
			0 & \frac{r^*_2(-\lambda_0)}{1-r_2(-\lambda_0)r_2^*(-\lambda_0)} & 0\\
			0 & 0 & 0\\
			0 & 0 & 0
		\end{pmatrix},\\
		&v_{9,00}^{\epsilon}=\begin{pmatrix}
			0 & 0 & 0\\
			\frac{r_2(-\lambda_0)}{1-r_2(-\lambda_0)r_2^*(-\lambda_0)} & 0 & 0\\
			0 & 0 & 0
		\end{pmatrix},\quad
		&&v_{10,00}^{\epsilon}=\begin{pmatrix}
			0 & r_2^*(-\lambda_0) & 0\\
			0 & 0 & 0\\
			0 & 0 & 0
		\end{pmatrix}.
	\end{aligned}
\end{equation*}
\par
Similar calculations as above indicate that
\begin{equation*}
       a_{00}(x,t)=-\frac{\sqrt{2\pi}{\rm{e}}^{\frac{\pi i}{4}-\frac{\pi \nu_1(\lambda_0)}{2}}}{r_1^*(\lambda_0)\Gamma(i\nu_1(\lambda_0))},\quad b_{00}(x,t)=-\frac{\sqrt{2\pi}{\rm{e}}^{-\frac{\pi i}{4}-\frac{5\pi \nu_4(-\lambda_0)}{2}}}{r_2(-\lambda_0)\Gamma(-i\nu_4(-\lambda_0))}.
\end{equation*}
Thus, under the case $N=0$ for $\vert\frac{x}{t}\vert<1$, the long-time asymptotic form of the solution to the two-component nonlinear KG equation \eqref{3} can be obtained, as shown in equations \eqref{N01} and \eqref{N02}.

Henceforth, we have completed the proof of the long-time asymptotics in high-order form for the solutions in Region ${\rm{IV}}$ of Theorem \ref{region}.

\subsubsection{Long-time asymptotic behavior in Region ${\rm{III}}$: $\vert\frac{x}{t}\vert\to1^{-}$}\label{regionIII}
\ \ \ \
In Region ${\rm{III}}$, where $\vert\frac{x}{t}\vert\to1^{-}$, the critical point $\lambda_0=\sqrt{\frac{\vert x-t\vert}{\vert x+t\vert}}$ approaches $0$ for $x>0$ and diverges to $\infty$ for $x<0$. More importantly, according to Theorem \ref{theo21}, the reflection coefficients $r_1(\lambda)$ and $r_2(\lambda)$ vanish to all orders as $\lambda\to0$ and $\lambda\to\infty$. Hence, in Region ${\rm{III}}$, we make the following modifications to the estimates given in Lemma \ref{lem43}.
\begin{lem} Under the same assumptions as Lemma \ref{lem43}, when $\lambda_0$ approaches $0$ for $x>0$ and diverges to $\infty$ for $x<0$, a nonnegative smooth function $C_N(\lambda)$ exists, which vanishes to all orders at $\lambda=0$ and $\lambda=\infty$. Consequently, we obtain
\begin{enumerate}
	\item The scattering coefficients $r_j(\lambda)$ and $\hat{r}_j(\lambda)$ $(j=1,2)$ satisfy the following estimators for $N\geq0$
	
	\begin{equation*}
		\begin{aligned}
			&\left|  r_{1,a}(\lambda)-\sum_{n=0}^{N}\frac{r^{(n)}_1(\lambda_0)}{n!}(\lambda-\lambda_0)^n \right| \leq C_N(\lambda_0){\rm{e}}^{\frac{t}{4}{\vert{\rm{Re}}\,\vartheta_{21}(\lambda)\vert}}\vert \lambda-\lambda_0\vert^{N+1},\\
			&\left|   r^*_{2,a}(\lambda)-\sum_{n=0}^{N}\frac{r_2^{*(n)}(-\lambda_0)}{n!}(\lambda+\lambda_0)^n\right| \leq C_N(-\lambda_0){\rm{e}}^{\frac{t}{4}\vert{\rm{Re}}\,\vartheta_{21}(\lambda)\vert}\vert \lambda+\lambda_0\vert^{N+1},
		\end{aligned}
	\end{equation*}
	and
	\begin{equation*}
		\begin{aligned}
			&\left|  \hat{r}_{1,a}(\lambda)-\sum_{n=0}^{N}\frac{\hat{r}_1^{(n)}(\lambda_0)}{n!}(\lambda-\lambda_0)^n\right| \leq C_N(\lambda_0){\rm{e}}^{\frac{t}{4}\vert{\rm{Re}}\,\vartheta_{21}(\lambda)\vert}\vert \lambda-\lambda_0\vert^{N+1},\\
			&\left|   \hat{r}_{2,a}(\lambda)-\sum_{n=0}^{N}\frac{\hat{r}_2(-\lambda_0)}{n!}(\lambda+\lambda_0)^n\right| \leq C_N(-\lambda_0){\rm{e}}^{\frac{t}{4}\vert{\rm{Re}}\,\vartheta_{21}(\lambda)\vert}\vert \lambda+\lambda_0\vert^{N+1},
		\end{aligned}
	\end{equation*}
	where $\vartheta_{21}=\theta_{21}/t$ and $C_N(\lambda_0)$ is smooth non-negative function that vanishes to any order as $\lambda=0$ and $\lambda=\infty$.
	\item For any $1\leq p\leq \infty$, the scattering coefficients satisfy the following inequalities
	\begin{equation*}
		\begin{aligned}
			&\|(1+\vert\cdot\vert)r_{1,r}(x,t,\cdot){\rm{e}}^{-t\vartheta_{21}}\|_{L^p(\lambda_0,\infty)}\leq \frac{C}{t^{N+\frac{3}{2}}},\\	&\|(1+\vert\cdot\vert)r^*_{2,r}(x,t,\cdot){\rm{e}}^{-t\vartheta_{21}}\|_{L^p(-\infty,-\lambda_0)}\leq \frac{C}{t^{N+\frac{3}{2}}},\\	&\|(1+\vert\cdot\vert)\hat{r}_{1,r}(x,t,\cdot){\rm{e}}^{-t\vartheta_{21}}\|_{L^p(0,\lambda_0)}\leq \frac{C}{t^{N+\frac{3}{2}}},\\	&\|(1+\vert\cdot\vert)\hat{r}^*_{2,r}(x,t,\cdot){\rm{e}}^{-t\vartheta_{21}}\|_{L^p(-\lambda_0,0)}\leq \frac{C}{t^{N+\frac{3}{2}}}.\\
		\end{aligned}
	\end{equation*}
\end{enumerate}
\end{lem}
The other transformations are similar to the above. Therefore, we can obtaion the following lemma.
\begin{lem}\label{RegionIII}
	When $\vert\frac{x}{t}\vert\to1^{-}$, that is, in the case of Region ${\rm{III}}$, the solutions $u(x,t)$ and $v(x,t)$ to the initial value problem of the two-component nonlinear KG equation (\ref{3}) have the same leading term as in Region ${\rm{IV}}$, as shown in equations \eqref{58} and \eqref{59} , with an error term of $\mathcal{O}\left(t^{-N-3/2} +\frac{C_N(\lambda_0)(\ln t)^{N+1}}{t^{(N+2)/2}}\right) $.
\end{lem}
\par
In conclusion, we have completed all the proofs of Theorem \ref{region}.

\bigskip
\par
\subsection*{Acknowledgements}
\ \ \ \
We extend our sincere gratitude to Xiaodong Zhu for his enthusiastic support and insightful discussions, which have significantly propelled the progress and successful completion of this research. Additionally, we acknowledge the generous financial support from the National Natural Science Foundation of China, Grant Nos. 12371247 and 12431008.

\bigskip

\noindent{\bf Data availability}
Data sharing not applicable to this article as no data sets were generated or analysed during the current study.

\subsection*{Statements and Declarations}

\noindent{\bf Competing Interests}
The authors declare that they have no conflict of interest.

\bibliographystyle{amsplain}

\begin{thebibliography}{10}
	
\bibitem{Messiah-1999} Messiah, A.: Quantum Mechanics. Dover Publications, (1999).

\bibitem{Kleinert-2001}Kleinert, H., Schulte-Frohlinde, V.: Critical Properties of $\phi^4$ Theories. World Scientific, Singapore, (2001).

\bibitem{Takhatajian-1974} Takhatajian, L. A., Faddeev, L. D.: Essentially Nonlinear One-Dimensional
Model of Classical Field Theory. Theor. Math. Phys. 21, 1046-1057 (1974).

\bibitem{Zhou-sG} Cheng, P. J., Venakides, S., Zhou, X.: Long-time asymptotics for the pure radiation solution of the sine-Gordon equation. Commun. Partial Differ. Equ. 24, 1195-1262 (1999).

\bibitem{tzitzeica} Tzitz\'eica, G.: Sur une nouvelle classe de surfaces. C. R. Acad. Sci. Paris 144, 1257-1259 (1907).

\bibitem{tzitzeica-2024}Huang, L., Wang, D. S., Zhu, X. D.: Long-time asymptotics of the Tzitz\'eica equation on the line. arXiv preprint arXiv:2404.04999 (2024).

\bibitem{Dunajski-2009} Dunajski, M., Plansangkate, P.: Strominger-Yau-Zaslow Geometry, Affine Spheres and Painlev\'e III. Commun. Math. Phys. 290, 997-1024 (2009).

\bibitem{Martel-2021}C${\rm\hat{o}}$te, R., Martel, Y., Yuan, X.: Long-time asymptotics of the one-dimensional damped nonlinear Klein-Gordon equation. Arch. Ration. Mech. Anal. 239, 1837-1874  (2021).


\bibitem{Fordy1}Fordy, A. P., Gibbons, J.: Integrable Nonlinear Klein-Gordon Equations and Toda Lattices. Commun. Math. Phys. 77, 21-30 (1980).

\bibitem{Fordy2} Fordy, A. P., Gibbons, J.: Nonlinear Klein-Gordon equations and simple Lie algebras. Proc. R. Ir. Acad. A 83, 33-44 (1983).
    
    
\bibitem{Toda-2012} Toda, M.: Theory of nonlinear lattices. Springer Science and Business Media. Vol. 20 (2012).

\bibitem{Mikhailov-1981} Mikhailov, A.V.: The reduction problem and the inverse scattering method. Phys. D 3, 73-117  (1981).

\bibitem{Hitchin-1987} Hitchin, N. J.: The self-duality equations on a Riemann surface. Proc. Lond. Math. Soc. 55, 59-126 (1987).

\bibitem{Ozer-1998}\"{O}zer, M. N.: A new integrable reduction of the coupled NLS equation. Turkish J. Math. 22, 319-334 (1998).

\bibitem{Wu-Geng-2013} Wu, L. H.,  He, G. L. Geng, X. G.: Quasi-periodic solutions to the two-component nonlinear Klein-Gordon equation. J. Geom. Phys. 66, 1-17 (2013).
    
\bibitem{Kotlyarov-Shepelsky-2017} Kotlyarov, V., Shepelsky, D.: Planar unimodular Baker-Akhiezer function for the nonlinear Schr\"{o}dinger equation. Ann. Math. Sci. Appl. 2, 343-384 (2017).
    
\bibitem{Constantin-2010}
Constantin, A., Ivanov, R. I.,  Lenells, J.: Inverse scattering transform for the Degasperis-Procesi equation. Nonlinearity 23, 2559-2575 (2010).
    
\bibitem{Deift-2008} Deift,  P.: Some open problems in random matrix theory and the theory of integrable systems. Contemp. Math. 458, 419-430 (2008).

\bibitem{good-boussinesq}  Charlier, C.,  Lenells, J.: The ``good" Boussinesq equation: a RH approach. Indiana Univ. Math. J. 71, 1505-1562 (2022).

\bibitem{lenells-wang-good} Charlier, C.,  Lenells, J.,  Wang, D. S.: The ``good" Boussinesq equation: long-time asymptotics. Anal. PDE 16, 1351-1388 (2023).

\bibitem{Deift-Zhou-1993}  Deift, P.,  Zhou, X.: A steepest-descent method for oscillatory RH problem. Asymptotics of the MKdV equation. Ann. Math. 137,  295-368 (1993).

\bibitem{Charlier-Lenells-2022}
 Charlier, C.,  Lenells, J.: On Boussinesq's equation for water waves. arXiv:2204.02365 (2022).

\bibitem{Grava-Minakov-2020}
 Grava, T.,  Minakov, A.: On the Long-Time Asymptotic Behavior of the Modified Korteweg-de Vries Equation with Step-like Initial Data. SIAM J. Math. Anal. 52, 5892-5993 (2020).

\bibitem{Girotti-Grava-2021}
 Girotti, M.,  Grava, T.,  Jenkins, R.,  McLaughlin, K. D. T.-R.: Rigorous asymptotics of a KdV soliton gas. Comm. Math. Phys. 384, 733-784 (2021).

\bibitem{WangJMP} Wang, D. S.,  Zhu, X. D.: Long-time asymptotics of the good Boussinesq equation with $q_{xx}$-term and its modified version. J. Math. Phys.  63, 123501 (2022).

\bibitem{Charlier-Lenells-miura}
 Charlier, C.,  Lenells, J.: Miura transformation for the ``good'' Boussinesq equation. Stud. Appl. Math. 152, 73-110 (2024).

\bibitem{Charlier-Lenells-2024} Charlier, C.,  Lenells, J.: The soliton resolution conjecture for the Boussinesq equation. J. Math. Pures Appl. 191, 103621 (2024).
       
\bibitem{Wang-Zhu-2025}
 Wang, D. S.,  Zhu, C.,  Zhu, X. D.: Miura transformations and large-time behaviors of the Hirota-Satsuma equation. J. Diff. Equations, 416, 642-699 (2025).

\bibitem{higher-order} Deift, P. A.,  Zhou, X.: Long-time asymptotics for integrable systems: higher order theory. Commun. Math. Phys. 165, 175-191 (1994).

\bibitem{Beals-Coifman-1984} Beals, R.,  Coifman, R.: Scattering and inverse scattering for first order systems. Comm. Pure Appl. Math. 37, 39-90 (1984).

\bibitem{lenells-2018} Lenells, J.: Matrix RH problems with jumps across Carleson contours. Monatsh. Math.  186, 111-152 (2018).
	
\bibitem{lenells-sG} Huang, L.,  Lenells, J.: Nonlinear Fourier transforms for the sine-Gordon equation in the quarter plane. J. Differ. Equations 264, 3445-3499 (2018).
	





\end{thebibliography}

\end{document}